\RequirePackage{etex}
\documentclass[11pt]{article}
\usepackage[letterpaper, margin=1in]{geometry}
\usepackage{amsmath,amssymb,amsthm}
\usepackage{enumitem}

\newtheorem{theorem}{Theorem}
\newtheorem{lemma}[theorem]{Lemma}
\newtheorem{proposition}[theorem]{Proposition}

\newtheorem{definition}[theorem]{Definition}
\newtheorem{remark}[theorem]{Remark}
\newtheorem{claim}[theorem]{Claim}

\usepackage{booktabs} 
\usepackage[ruled, linesnumbered]{algorithm2e} 

\SetAlFnt{\small}
\SetAlCapFnt{\small}
\SetAlCapNameFnt{\small}
\SetAlCapHSkip{0pt}
\IncMargin{-\parindent}

\usepackage{tikz-network}
\usetikzlibrary{graphs}
\usepackage{subcaption}

\usepackage{xcolor}
\usepackage[colorlinks=true, allcolors=blue]{hyperref}
\usepackage{xspace}

\newcommand{\score}{\mathrm{score}}
\newcommand{\opt}{\ensuremath{\operatorname{OPT}}\xspace}

\newcommand{\alg}{\ensuremath{\operatorname{ALG}}\xspace}
\newcommand{\ranking}{\ensuremath{\operatorname{RANKING}}\xspace}
\newcommand{\greedy}{\ensuremath{\textsc{Greedy}}\xspace}
\newcommand{\omp}{\ensuremath{\operatorname{OMP}}\xspace}
\newcommand{\womp}{\ensuremath{\operatorname{WOMP}}\xspace} 
\newcommand{\dual}{\ensuremath{\operatorname{DUAL}}\xspace}
\newcommand{\lcp}{\ensuremath{\operatorname{LCP}}\xspace}
\usepackage{xparse}
\NewDocumentCommand{\apath}{m m g g g g}{%
  \IfNoValueTF{#3}%
    {[#1~#2]}%
    {\IfNoValueTF{#5}%
       {[#1~#2~#3~#4]}%
       {[#1~#2~#3~#4~#5~#6]}%
    }%
}
\newcommand{\spgreedy}{\textsc{SP-Greedy}\xspace}

\newcommand{\rootof}{\operatorname{root}} 
\newcommand{\act}{\operatorname{act}} 
\newcommand{\wit}{\operatorname{wit}} 
\newcommand{\mate}{\operatorname{mate}}

\title{Deterministic online matching under short augmenting paths and restricted vertex reassignments}

\author{Mariano Llancamán\thanks{Department of Mathematical Engineering, Universidad de Chile.} \and José A. Soto\footnotemark[1] \thanks{Center for Mathematical Modeling CNRS-IRL 2807, Universidad de Chile.}}
\begin{document}
\date{}
\hypersetup{pageanchor=false}
\begin{titlepage}

\maketitle

\begin{abstract}
We study deterministic online bipartite matching with local recourse. Online vertices arrive one by one and reveal edges to a fixed offline set. After each arrival, the algorithm maintains a matching and may update it only by an augmenting path of length $1$ or $3$ starting at the new vertex. In the budgeted model $\omp_3(s,t)$, each offline vertex can be reassigned at most $s$ times and each online vertex at most $t$ times.

Our main result is a deterministic algorithm, Lowest-Cost-Path (\lcp), that achieves the optimal competitive ratio for every choice of budgets. If $s=0$ or $t=0$, the best possible ratio is $1/2$. For every $s\ge 1$ and finite $t$, the optimal ratio is
\[
\gamma_t=\frac{2\cdot 2^t-1}{3\cdot 2^t-1}.
\]
This value equals $3/5$ for $t=1$ and converges to $2/3$ as $t$ increases. For
$t=\infty$, the optimal ratio is exactly $2/3$. In particular, allowing a single reassignment per offline vertex already matches the best guarantee achievable with any larger offline budget. We also prove matching upper bounds for all deterministic algorithms.

Our analysis is primal--dual and uses a structural description of the history graph of \lcp, which explains how exhausted reassignment budgets can block length-$3$ augmentations. We also study offline-weighted variants. In particular, we determine the exact optimal deterministic ratios for two unit-offline-budget cases: $2-\sqrt2$ for $\womp_3(1,1)$, and $(\sqrt5-1)/2$ for $\womp_3(1,\infty)$. Both are strictly below their unweighted counterparts. If weights are placed on online vertices instead, no deterministic algorithm has a positive competitive ratio.
\end{abstract}
        
\vspace{1cm}
\setcounter{tocdepth}{2} 

\end{titlepage}
\hypersetup{pageanchor=true}
\setcounter{page}{1}

\section{Introduction}\label{sec:intro}

Online matching is a central model for decision-making under uncertainty, motivated by marketplaces such as sponsored search and online advertising, where a platform must allocate arriving opportunities to a fixed pool of agents subject to constraints; see, e.g., the survey and monograph by Mehta~\cite{Mehta13} and the AdWords and budgeted-allocation literature~\cite{Mehta13,BuchbinderJN07,GoelM08}.

In \emph{online bipartite matching}, we are given a bipartite graph $G=(L,R,E)$, where the offline side $L$ is known in advance and vertices of the online side $R$ arrive one by one in an adversarial order. Upon the arrival of an online vertex $j\in R$, its incident edges to $L$ are revealed, and the algorithm must immediately decide whether to match $j$ to a currently unmatched neighbor in $L$; in the classical model, decisions are irrevocable. The goal is to maximize the final matching size.

Karp, Vazirani, and Vazirani~\cite{Karp1990} introduced this model and showed that randomization breaks the deterministic barrier: \textsc{Ranking} achieves ratio $1-1/e$.
In contrast, against adversarial arrivals no deterministic algorithm can guarantee ratio better than $1/2$~\cite{Karp1990,BirnbaumM08}.

A natural way to overcome this deterministic barrier is to allow \emph{recourse}: upon each arrival, the algorithm may modify the current matching.
Unrestricted alternating-path recourse has two degrees of freedom: it may use long augmenting paths (changing many matches at once), and it may reassign the same vertex many times. We study a restricted model that bounds both: updates must be augmenting paths of length at most $3$ starting at the new vertex, and each vertex has a budget on how many times it may be reassigned.

\paragraph{Recourse via short augmenting paths.}
We focus on deterministic algorithms that, upon the arrival of a vertex $j\in R$, may increase the current matching only via an alternating path that starts at $j$ and has length at most $3$. Equivalently, the algorithm may either match $j$ directly to a free neighbor $x\in L$, or augment along a $3$-path $\apath{j}{x}{y}{i}$, where $x,i\in L$, $y\in R$, and $xy$ is the currently matched edge; this increases the matching size by $1$.

Operationally, a length-$3$ augmentation changes the partner of the previously arrived online vertex $y$ from $x$ to $i$, and the partner of the offline vertex $x$ from $y$ to $j$.
We call these partner changes \emph{reassignments}, and we denote by $\omp_3$ the model with no limit on reassignments.

Without reassignment limits, $\omp_3$ admits the classical deterministic $2/3$ guarantee: if the final matching has no augmenting path of length at most $3$, then the Hopcroft--Karp lemma~\cite{Hopcroft1973} implies a $2/3$-approximation. This certificate is not automatic in an online process; nevertheless, Shin et al.~\cite{ShinKLA2020} show that repeatedly augmenting along a shortest augmenting path from the new vertex exhausts all length-$1$ and length-$3$ augmentations, even online. A central question we ask is whether the $2/3$ benchmark survives under per-vertex reassignment budgets.

\paragraph{Budgets on vertex reassignments.}

For parameters $s,t\in \mathbb{N}\cup\{\infty\}$ (where $\mathbb{N}$ denotes the nonnegative integers), let $\omp_3(s,t)$ be the variant of $\omp_3$ in which each offline vertex may be reassigned at most $s$ times and each online vertex at most $t$ times. Thus $\omp_3(\infty,\infty)$ coincides with $\omp_3$, while if $s=0$ or $t=0$ then no length-$3$ augmentation is feasible and the model reduces to classical online matching.

These budgets differ from replacement bounds studied in prior work, where the cost is the number of matching edges modified, often measured amortized per arrival (e.g., \cite{BernsteinHR19}).
In particular, replacement bounds do not control how many times a fixed vertex may change partners, whereas our model imposes this restriction directly.

Asymmetric budgets are natural in two-sided platforms, where one side may tolerate fewer reassignments than the other; this motivates allowing different budgets on the two sides.

\paragraph{Deterministic focus.}
We focus on deterministic algorithms under adversarial arrivals. In this setting, classical online matching has a sharp $1/2$ barrier, so any improvement must come from recourse rather than randomization.

This leads to the main question of the paper: how much local recourse is needed to beat the deterministic $1/2$ barrier? More concretely, for a fixed online budget $t$, how much offline reassignment budget is useful? Our results show that even the smallest positive offline budget already suffices.

\subsection{Our results}

We characterize the optimal deterministic competitive ratio of $\omp_3(s,t)$ and give a deterministic algorithm that achieves it.
As a reference point, the optimal deterministic ratio for $\omp_3(\infty,\infty)$ is $2/3$~\cite{ShinKLA2020}. This benchmark follows from the Hopcroft--Karp certificate: if a matching has no augmenting path of length at most $3$, then it is a $2/3$-approximation.

We reprove this benchmark via a primal--dual analysis, which is also the starting point for the budgeted setting, where the short-augmenting-path certificate does not apply directly.

\paragraph{Optimal unweighted ratios.}
For the unweighted problem with per-vertex reassignment budgets, our main result gives a complete characterization:
\begin{itemize}
    \item If $s=0$ or $t=0$, the model reduces to classical deterministic online matching, and the optimal ratio is $1/2$.
    \item If $s\ge 1$ and $t<\infty$, the optimal deterministic ratio is
    $\displaystyle     \gamma_t=\frac{2\cdot 2^t-1}{3\cdot 2^t-1}.$
    \item If $s\ge 1$ and $t=\infty$, the optimal deterministic ratio is $2/3$.
\end{itemize}

The guarantee is achieved by \emph{Lowest-Cost-Path} (\lcp). It takes a direct match whenever one is available. Otherwise, among feasible length-$3$ augmentations $\apath{j}{x}{y}{i}$, it chooses one that minimizes the reassignment count already used by the middle online vertex $y$. The direct-first rule and the minimum-cost choice are both needed; Appendix~\ref{app:bad-examples} gives small instances showing this.

Although the analysis is carried out for $s=1$, Lemma~\ref{lemma:lcp-spgreedy} shows that \lcp never reassigns an offline vertex more than once. Hence the same guarantee holds for every $s\ge1$.

For finite $t$, the analysis does not follow from the Hopcroft--Karp short-augmenting-path certificate. The final matching produced by \lcp may contain length-$3$ augmenting paths, but these paths can be infeasible because some reassignment budgets are exhausted.

Our analysis is primal--dual. We construct dual variables for the matching LP and show that, after scaling, they are feasible. The key ingredient is the \emph{middle-edge lemma} (Lemma~\ref{lemma:middle}): every remaining length-$3$ augmenting path has a frozen middle edge. This local obstruction replaces the missing short-augmenting-path certificate and determines the dual construction.

\paragraph{Upper bounds.}
We prove matching upper bounds for deterministic algorithms. For
$\omp_3(\infty,\allowbreak\infty)$, a small finite adaptive construction gives the
$2/3$ bound. For $\omp_3(\infty,1)$, another small finite adaptive
construction gives the $3/5$ bound. For $\omp_3(s,t)$ with $s\ge1$ and
$1\le t<\infty$, the tight bound $\gamma_t$ is proved by a different
adaptive construction: for arbitrarily large $N$, we build a finite
instance, and the bound follows as $N\to\infty$.

A consequence is that, for length-$3$ augmentations, the offline budget becomes irrelevant as soon as it is positive: for every $t$, the optimal ratio for $\omp_3(s,t)$ is the same for all $s\ge1$. This collapse is specific to length-$3$ augmentations. Allowing paths of length at most $5$ gives ratio $3/4$ for $\omp_5(\infty,\infty)$, but Appendix~\ref{app:bad-examples} shows that the analogous collapse fails for $\omp_5(1,\infty)$.

The obstruction is structural. For length-$3$ augmentations, every remaining augmenting path has a single middle matching edge, and the middle-edge lemma can identify a local budget obstruction on that edge. For longer augmenting paths, the obstruction can be distributed over several matched edges, so the same local certificate does not directly extend.

\paragraph{Weighted results.}
Finally, we study offline-weighted variants of the unit-offline-budget problem, where offline vertices have nonnegative weights and the objective is to maximize the total matched offline weight. We determine the exact optimal deterministic competitive ratios for two cases:
\begin{itemize}
    \item $2-\sqrt2 \approx 0.585786$ for $\womp_3(1,1)$, and
    \item $(\sqrt5-1)/2 \approx 0.618034$ for $\womp_3(1,\infty)$.
\end{itemize}
Both are strictly smaller than their unweighted counterparts ($3/5$ and $2/3$, respectively) under the same reassignment limits. Thus offline weights make the budgeted problem strictly harder.

The weighted guarantees are achieved by greedy algorithms and proved via primal--dual constructions. We also give matching upper bounds via adaptive instances. With arbitrary weights on online vertices, no deterministic algorithm has a positive competitive ratio.

\subsection{Related work}

\paragraph{Replacements, augmentations, and switching costs.}
A related line of work studies online matching under recourse, with the goal of maintaining an optimal or near-optimal matching while controlling the number of changes. In the \emph{replacements} model, the algorithm may modify the current matching after each arrival, and the cost is the number of modified matching edges. Bosek et al.~\cite{BosekLSZ14} maintain a maximum matching with $O(\sqrt{n})$ \emph{amortized} replacements per insertion. Bernstein, Holm, and Rotenberg~\cite{BernsteinHR19} analyze augmenting along a shortest augmenting path from the new vertex and prove an $O(\log^2 n)$ \emph{amortized} bound, nearly matching the $\Omega(\log n)$ lower bound.

This line of work differs from ours in its objective and in its recourse primitive. It aims to maintain exact optimality while bounding update cost, often amortized over time. In contrast, we study worst-case competitive ratios under adversarial arrivals. Moreover, these algorithms may use long augmenting paths, whereas $\omp_3(s,t)$ permits only local length-$3$ augmentations and imposes a hard reassignment budget on each vertex over the entire execution. The closest hard-budget models known to us impose budgets on edges, updates, or total recourse. Our model differs in that the budget is persistent and local to each vertex.

A complementary formulation charges a \emph{switching cost} for reassignment. This problem was introduced by Grove et al.~\cite{GroveKKV95}. Chaudhuri, Daskalakis, Kleinberg, and Lin~\cite{ChaudhuriDKL09} study switching costs for online bipartite perfect matching with augmentations: assuming that a perfect matching exists, the algorithm must maintain a matching while minimizing the number of client reassignments. They obtain substantially smaller switching costs than the worst-case $O(n^2)$ bound in several settings, including random arrival order and forests. Our model instead imposes hard per-vertex budgets, restricts augmentations to length at most $3$, and maximizes matching size under adversarial arrivals.

\paragraph{Hard-budget recourse and competitive guarantees.}
Hard (worst-case) recourse budgets have also been studied from a competitive
analysis viewpoint. Shin, Kim, Lee, and An~\cite{ShinKLA2020} study online
bipartite matching with bounded augmentations from the newly arrived vertex.
If $\omp_k$ denotes the analogue of $\omp_3$ in which augmenting paths may have
length at most $k$, then their model corresponds to $\omp_k$. They do not impose
persistent budgets on the vertices. Their budget limits the size of one update;
ours limits how many times each vertex can change partner during the whole
execution. For odd $k$, allowing paths of length at most $k$ gives ratio
$(k+1)/(k+3)$.

Angelopoulos et al.~\cite{AngelopoulosDJ18} study edge arrivals with a recourse bound on each edge. Bernstein and Dudeja~\cite{BernsteinD20} consider random edge arrivals with recourse. Both models budget changes to matching edges, rather than partner changes experienced by vertices over the entire execution.

A different direction is the online batch-arrival model of Lee and Singla~\cite{LeeS20}, where edges are revealed in batches and the algorithm must irrevocably extend its matching at each batch.

\paragraph{Vertex-weighted online matching.} Our weighted extension is related to the classical vertex-weighted online bipartite matching problem, in which offline vertices have weights, and the goal is to maximize the total matched offline weight. Without recourse, Aggarwal, Goel, Karande, and Mehta~\cite{AggarwalGKM11} give the optimal $1-1/e$ competitive ratio under adversarial arrivals. Under random-order arrivals, stronger guarantees are possible: Huang, Tang, Wu, and Zhang~\cite{HuangTWZ19} adapt \textsc{Ranking} to achieve $0.6534$, later improved to $0.662$ by Jin and Williamson~\cite{JinW21} and to $0.6862$ by Peng and Tang~\cite{PengT25}. These results are randomized and do not allow reassignment. Our weighted result is in a different regime: deterministic algorithms, adversarial arrivals, and strictly local recourse with unit reassignment budgets on both sides.

\section{Preliminaries and Baseline Techniques}
\subsection{Model and Primal--Dual Template}
\begin{definition}[Online Matching problem (OMP)]
An instance of \emph{online bipartite matching} consists of a bipartite graph
$G=(L,R,E)$, where $L$ denotes the offline vertices and $R$ the online vertices,
and an arrival order $\pi$ of $R$. The vertices in $L$ are known in advance.
When a vertex $j\in R$ arrives, all incident edges are revealed, and the
algorithm must irrevocably decide whether to match $j$ to a currently free
(i.e., unmatched) neighbor in $L$. The objective is to maximize the final
matching size.
\end{definition}

\begin{definition}[Offline-vertex-weighted Online Matching problem (WOMP)]
In the offline-vertex-weighted variant, each offline vertex $i\in L$ has a
weight $w(i)\ge 0$, and the objective is to maximize the total weight of the
matched offline vertices.
\end{definition}

\begin{definition}[Online Matching with augmenting paths of length at most $K$]
Fix $K\ge 1$. In $\omp_K$, upon the arrival of $j\in R$, the algorithm may either leave $j$ unmatched or
augment along an augmenting path that starts at $j$ and ends at a free vertex in $L$, with length at most $K$.
\end{definition}
Note that $\omp_1$ coincides with OMP; for $K\ge 3$, $\omp_K$ allows recourse via augmenting paths.

Let $M$ be the current matching. An allowed augmentation uses an $M$-augmenting path $P$ that starts at the arriving vertex $j$ and ends at a free vertex in $L$. The augmentation updates the matching as $M \leftarrow M \triangle E(P)$. We only allow augmentations at arrivals, so the path used must start at the arriving vertex.
An augmentation may change the matched neighbor of a previously matched vertex; we call this a \emph{reassignment}. Formally, a vertex $v$ is \emph{reassigned} if it is matched both before and after an augmentation, but to different neighbors.

In particular, in a length-$3$ augmentation $\apath{j}{x}{y}{i}$, the reassigned vertices are exactly the internal vertices $x$ and $y$ (the endpoints $j$ and $i$ are newly matched).

Note that vertices that are matched before an augmentation remain matched after it. They either keep their matched neighbor or are reassigned to a new one. In particular, an augmentation never makes a previously matched vertex free. Furthermore, once an online vertex arrives, its full neighborhood is revealed and remains fixed.
Figure~\ref{fig:example} illustrates a length-$3$ augmentation.

\begin{figure}[h]
\centering
\begin{tikzpicture}

  \Vertex[x=0,y=0, size=.5, color=white,
          label=$i$, position=above]{iA}
  \Vertex[x=0,y=1.5, size=.5, color=black,
          label=$x$, position=above]{xA}

  \Vertex[x=2,y=1.5, size=.5, color=black,
          label=$y$, position=above]{yA}
  \Vertex[x=2,y=0, size=.5, color=white,
          label=$j$, position=above]{jA}

  \Edge[style=solid](xA)(yA)
  \Edge[style=dashed](jA)(xA)
  \Edge[style=dashed](yA)(iA)

  \Vertex[x=3,y=1.5,style={color=white}]{arrowA}
  \Vertex[x=5,y=1.5,style={color=white}]{arrowB}
  \Edge[Direct](arrowA)(arrowB)

  \Vertex[x=6,y=0, size=.5, color=black,
          label=$i$, position=above]{iB}
  \Vertex[x=6,y=1.5, size=.5, color=black,
          label=$x$, position=above]{xB}

  \Vertex[x=8,y=1.5, size=.5, color=black,
          label=$y$, position=above]{yB}
  \Vertex[x=8,y=0, size=.5, color=black,
          label=$j$, position=above]{jB}

  \Edge[style=solid](jB)(xB)
  \Edge[style=solid](yB)(iB)
  \Edge[style=dashed](xB)(yB)

\end{tikzpicture}
\caption{Upon the arrival of $j$, we augment the current matching $M=\{xy\}$ along the $3$-path $\apath{j}{x}{y}{i}$, removing $xy$ and adding $xj$ and $yi$. The vertices $x$ and $y$ are reassigned.}
\label{fig:example}
\end{figure}
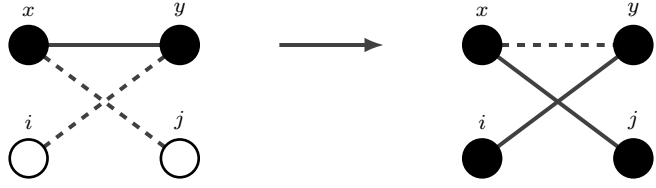

\begin{definition}[$\omp_K(s,t)$]
Let $s,t\in \mathbb N\cup\{\infty\}$. In $\omp_K(s,t)$, each vertex in $L$ may be reassigned at most $s$ times, and each vertex in $R$ may be reassigned at most $t$ times over the execution.
A vertex is \emph{exhausted} once it reaches its budget.
An augmentation is \emph{feasible} if every vertex that would be reassigned is not exhausted before the augmentation.
A matched edge is \emph{frozen} if an endpoint is exhausted; frozen matched edges cannot be removed by any feasible augmentation.
We view $\omp_K$ as the special case $\omp_K(\infty,\infty)$.
\end{definition}

\begin{definition}[Competitive ratio for $\omp_K(s,t)$]
An online deterministic algorithm $\mathcal A$ is \textbf{$\rho$-competitive} for
$\omp_K(s,t)$ if, for every instance $(G,\pi)$, its output size satisfies
$\mathcal A(G,\pi)\ge \rho\cdot \nu(G)$, where $\nu(G)$ is the size of a maximum
matching in $G$. Throughout the paper, we will often use $\opt=\nu(G)$ to denote the size of the optimal matching, and $\alg$ to denote the size of the matching returned by an algorithm.
\end{definition}

We write augmenting paths as ordered sequences of vertices. Thus, $\apath{j}{i}$ denotes a length-$1$ augmenting path, where $j\in R$ is the newly arrived online vertex and $i\in L$ is a free offline vertex. Similarly, $\apath{j}{x}{y}{i}$ denotes a length-$3$ augmenting path, where $j,y\in R$ and $x,i\in L$; here, the edge $xy$ belongs to the current matching, while $i$ remains free. Applying the augmentation $\apath{j}{x}{y}{i}$ removes $xy$ from the matching and adds the edges $jx$ and $yi$.

We use the following shorthand in the warmup and in the middle-edge
lemma. An algorithm is \emph{\spgreedy} for $\omp_3(s,t)$ if, upon the
arrival of $j\in R$, it augments along some feasible path of length $1$
or $3$ starting at $j$ whenever such a path exists. If no such path
exists, it leaves $j$ unmatched. The term \emph{short} refers to the
allowed path lengths $1$ and $3$; it does not require the algorithm to
prioritize length-$1$ augmentations over length-$3$ augmentations.
A \emph{direct-first} algorithm is a special \spgreedy algorithm that
prioritizes length-$1$ augmentations (direct matches) over length-$3$
augmentations.

We will repeatedly use a standard primal--dual template to analyze the
competitive ratio of an online matching algorithm, based on the LP relaxation
of maximum bipartite matching. This approach was used, e.g., by Devanur, Jain,
and Kleinberg~\cite{Devanur2013} to analyze $\ranking$, and has since been used
extensively in the literature. It is included here only to fix the normalization of the dual variables used in the different analyses.

\begin{definition}[Bipartite matching LP]\label{def:matching-lp}
Let $G=(L,R,E)$ be a bipartite graph, and let $N(v)$ denote the neighborhood of a vertex $v$. The standard LP relaxation and its dual are:
\[
\begin{array}{@{}c@{\qquad}c@{}}
\begin{aligned}
(P_G)\qquad
\max \ & \sum_{ij\in E} x_{ij} \\
\text{s.t.}\ &
\sum_{j\in N(i)} x_{ij} \le 1 \qquad \forall i\in L,\\
&
\sum_{i\in N(j)} x_{ij} \le 1 \qquad \forall j\in R,\\
& x_{ij} \ge 0 \qquad \forall ij\in E
\end{aligned}
&
\begin{aligned}
(D_G)\qquad
\min \ & \sum_{i\in L}\alpha_i + \sum_{j\in R}\beta_j \\
\text{s.t.}\ & \alpha_i+\beta_j \ge 1 \qquad \forall\, ij\in E,\\
& \alpha_i\ge 0 \qquad \forall i\in L,\\
& \beta_j\ge 0 \qquad \forall j\in R\\
&\\
\end{aligned}
\end{array}
\]
\end{definition}

\begin{lemma}[Primal--dual template]\label{lemma:pd-template}
Let $\mathcal A$ be an online algorithm for $\omp_K(s,t)$, and consider an
arbitrary instance $(G,\pi)$ with $G=(L,R,E)$.
Assume that, during the execution of $\mathcal A$, we construct nonnegative
dual variables $(\alpha_i)_{i\in L}$ and $(\beta_j)_{j\in R}$ such that:
\begin{enumerate}
  \item[(1)] All dual variables start at $0$.
  \item[(2)] Whenever the matching maintained by $\mathcal A$ increases in size by
        $1$, the dual objective
        $\sum_{i\in L}\alpha_i + \sum_{j\in R}\beta_j$
        increases by at most $1/\rho$ (and no other increase/decrease occurs during the execution).
  \item[(3)] At the end of the execution, $(\alpha,\beta)$ is feasible for $D_G$.
\end{enumerate}
Then $\mathcal A$ is $\rho$-competitive for $\omp_K(s,t)$, i.e.,
$\mathcal A(G,\pi)\ge \rho\cdot \nu(G)$, where $\nu(G)$ is the size of a maximum
matching in $G$.
\end{lemma}

\begin{proof}
Let $X$ be the size of the final matching output by $\mathcal A$ on $(G,\pi)$,
and let $Y = \sum_{i\in L}\alpha_i + \sum_{j\in R}\beta_j$ be the final value of
the constructed dual objective.
By (1) and (2), we have $Y\le X/\rho$.
By (3), $(\alpha,\beta)$ is feasible for $D_G$, hence $\opt(D_G)\le Y$.
By strong duality for linear programs, $\opt(P_G)=\opt(D_G)$.
Moreover, since $G$ is bipartite, the matching polytope is integral, and the LP
relaxation is exact; in particular, $\nu(G)=\opt(P_G)$. Therefore,
\[
\nu(G)\ =\ \opt(P_G)\ =\ \opt(D_G)\ \le\ Y\ \le\ X/\rho,
\]
which implies $X/\nu(G)\ge \rho$, as claimed.
\end{proof}

\subsection{The middle-edge lemma}\label{subsec:middle-edge}

We next state the middle-edge lemma. This lemma is used in the unbudgeted
warmup and later in the proof of Theorem~\ref{thm:lcp-ratio}. In the unbudgeted case, an
\spgreedy algorithm cannot end with an augmenting path of length $3$: when the
online endpoint arrived, the algorithm could have used that path. With
reassignment budgets, this may fail. Such a path may remain because one endpoint
of its middle matched edge is exhausted. The lemma states that this is the only
possible obstruction: every remaining augmenting $3$-path has a frozen middle
matched edge.

\begin{lemma}[Middle-edge lemma]\label{lemma:middle}
Fix $s,t\in \mathbb{N} \cup \{\infty\}$ and an instance $(G,\pi)$ with $G=(L,R,E)$.
Let $\mathcal A$ be any \spgreedy algorithm for $\omp_3(s,t)$ and let $M$ be the
final matching produced by $\mathcal A$ on $(G,\pi)$.
Then for every edge $xy\in M$, if neither endpoint is exhausted (equivalently,
if $xy$ is not frozen), the edge $xy$ cannot be the middle edge of an augmenting
$3$-path for $M$. Equivalently, at least one of $x$ or $y$ has no free neighbor
with respect to $M$.
\end{lemma}
\begin{remark}
In the special case $s=t=\infty$, no vertex is ever exhausted. Hence
no matched edge is frozen, and Lemma~\ref{lemma:middle} rules out all augmenting
$3$-paths in the final matching of any \spgreedy algorithm.
\end{remark}

\begin{proof}
Suppose for a contradiction that, at the end of the execution, $M$ admits an
augmenting $3$-path $\apath{j}{x}{y}{i}$, where
$j\in R$ and $i\in L$ are free in $M$, $xy\in M$, and
neither $x$ nor $y$ is exhausted.

    Let $M_j$ be the matching maintained by $\mathcal A$ when $j$ arrives. Since
    $j$ is free in the final matching, $\mathcal A$ left $j$ unmatched when it
    arrived. As $\mathcal A$ is \spgreedy, there is no feasible (i.e.,
    budget-respecting) augmenting path of length $1$ or $3$ starting at $j$ at
    that time. In particular:
    \begin{enumerate}
        \item $x\in N(j)$ is matched in $M_j$ (otherwise $\mathcal A$ would match $j$
        directly to $x$).
        \item $xy \notin M_j$. Otherwise $\apath{j}{x}{y}{i}$ would be a feasible
        augmenting $3$-path at the arrival of $j$: the vertex $i$ is free at that
        time because it is free at the end and offline vertices never become free
        again, and the edge $yi$ is available because $y$ arrived before $j$.
    \end{enumerate}

    Let $z\in R\setminus\{j\}$ be such that $xz\in M_j$. Since $j$ was not matched
    upon arrival via an augmenting path using $xz$, and $x$ is not exhausted,
    either $z$ is exhausted at that time or $z$ has no free neighbor in $L$ to
    complete a feasible length-$3$ augmentation.

In either case, the edge $xz$ can never be removed from the matching after
the arrival of $j$. If $z$ is exhausted, then $xz$ is frozen. If
$z$ has no free neighbor at that time, then it can never gain one later: the
full neighborhood of $z$ has already been revealed, and offline vertices never
become free after being matched.
This contradicts the assumption that, at the end of the execution, $x$ is
matched to $y\neq z$. Therefore, in any augmenting $3$-path, the middle edge
must be frozen.
\end{proof}

\subsection{Unbudgeted \texorpdfstring{$\omp_3$}{OMP3}}

We begin with the infinite-budget case. In this setting, the usual
short-augmentation argument gives the $2/3$ guarantee for any algorithm
that always uses an available augmenting path of length $1$ or $3$. As
noted in prior work (e.g., Shin, Kim, Lee, and An~\cite{ShinKLA2020}),
when reassignment budgets are infinite on both sides, this guarantee is
straightforward.

We give a primal--dual proof using Lemma~\ref{lemma:pd-template}.
This proof sets up the dual-variable structure used later for finite reassignment
budgets, where the short-augmenting-path certificate is no longer sufficient.

\begin{theorem}\label{thm:spgreedy}
Every \spgreedy algorithm for $\omp_3$ attains a competitive ratio of $\frac23$.
\end{theorem}

\begin{proof}[Primal--Dual proof of Theorem~\ref{thm:spgreedy}]
Fix an instance $(G,\pi)$ with $G=(L,R,E)$ and a \spgreedy algorithm $\mathcal A$ for
$\omp_3$, and let $M$ be its final matching. Alongside the execution, we
construct a dual solution $(\alpha,\beta)$ for $D_G$, starting from the all-zero
solution. The updates are given in Algorithm~\ref{alg:dual-warmup}. Each time
$\mathcal A$ increases the matching size by $1$, the dual objective increases by
exactly $3/2$. The postprocessing step only redistributes dual weight along
edges of $M$, preserving the value, to ensure dual feasibility.

\begin{algorithm}[ht!]
\SetAlgoLined
\DontPrintSemicolon
\KwIn{Instance $(G,\pi)$, algorithm $\mathcal A$}
\KwOut{Dual vector $(\alpha_i,\beta_j)_{i\in L,\,j\in R}$}
\lForAll{$i\in L, j\in R$}{set $\alpha_i\gets 0$, $\beta_j\gets 0$\quad}
\ForAll(\tcp*[f]{dual update steps}){$j\in R$ in arrival order, while running $\mathcal A$}{
  \lIf{$\mathcal A$ matches $j$ directly to some $i\in L$}{
    Set $\alpha_i\gets \frac12$ and $\beta_j\gets 1$.
  }
  \lElseIf{$\mathcal A$ matches $j$ using the $3$-path $\apath{j}{x}{y}{i}$}{
    Set $\alpha_x\gets \frac12$, $\alpha_i\gets \frac12$,
    $\beta_j\gets 1$, and $\beta_y\gets 1$.
  }
}
Let $M$ be the final matching constructed by $\mathcal A$\;

\ForAll(\tcp*[f]{postprocessing step}){$ij\in M$ with $j\in R$}{
  \lIf{$j$ has no free neighbor in $L$ at the end of the execution}{
    Set $\alpha_i\gets 1$ and $\beta_j\gets \frac12$.
  }
}
\caption{$\dual$ for $\omp_3(\infty,\infty)$}
\label{alg:dual-warmup}
\end{algorithm}

We first verify the change in the dual objective. Each time $\mathcal A$ matches
an arrival $j$, the dual objective increases by exactly $3/2$. In a direct
match, $j$ is new so $\beta_j$ increases from $0$ to $1$, and the chosen
neighbor $i$ is free so $\alpha_i$ increases from $0$ to $1/2$. In a $3$-path
augmentation, $j$ is new so $\beta_j$ again increases from $0$ to $1$, and the
free endpoint $i\in L$ gets $\alpha_i$ increased from $0$ to $1/2$; the other
assignments are only for clarity and do not increase any dual variable.

The postprocessing step only redistributes $1/2$ within a matched edge $ij\in M$,
changing $(\alpha_i,\beta_j)$ from $(1/2,1)$ to $(1,1/2)$ and thus preserving the
objective value. Therefore, by Lemma~\ref{lemma:pd-template} with $\rho=2/3$, it
remains to prove dual feasibility, i.e., that $\alpha_i+\beta_j\ge 1$ for every
edge $ij\in E$.

Fix an arbitrary edge $ij\in E$. Since $\mathcal A$ is \spgreedy, $i$ and $j$ cannot
both be free in the final matching $M$ (otherwise $ij$ would be an augmenting
$1$-path). We consider the remaining cases.

If both $i$ and $j$ are matched in $M$, then by construction each has value at
least $\frac12$, so $\alpha_i+\beta_j\ge 1$.

If $j$ is matched in $M$ and $i$ is free, then $\beta_j=1$. Indeed, $\beta_j$ can
become $\frac12$ only in the postprocessing step of
Algorithm~\ref{alg:dual-warmup}, which requires that $j$ has no free neighbor in
$L$ at the end (impossible since $i\in N(j)$ is free). Thus
$\alpha_i+\beta_j\ge 1$.

Finally, suppose that $i$ is matched in $M$ and $j$ is free, and let $ik\in M$
be the edge matching $i$.
Since this is an execution in $\omp_3(\infty,\infty)$, neither endpoint
of $ik$ is exhausted. By the middle--edge lemma, Lemma~\ref{lemma:middle}, the edge $ik$ cannot be the
middle edge of an augmenting $3$-path for $M$.
In particular, $k$ cannot have a free neighbor $k'\in L$; otherwise,
$\apath{j}{i}{k}{k'}$ would be an augmenting $3$-path. Therefore $k$ has no free neighbor in
$L$ at the end of the execution, and the postprocessing step sets $\alpha_i=1$
(and $\beta_k=\frac12$). Thus $\alpha_i+\beta_j\ge 1$.

We proved that $(\alpha,\beta)$ is feasible for $D_G$. By Lemma~\ref{lemma:pd-template}, $|M|\ge \frac23\,\nu(G)$, so $\mathcal A$ is $\frac23$-competitive.
\end{proof}

\section{The Lowest-Cost-Path Algorithm} \label{sec:lcp}
\subsection{Algorithm overview}

In this section, we propose and analyze a specific algorithm for $\omp_3(s,t)$,
the Lowest-Cost-Path (\lcp) algorithm (Algorithm~\ref{alg:LCP}).

\begin{definition}[Lowest-Cost-Path (\lcp)]
The Lowest-Cost-Path (\lcp) algorithm is an algorithm for $\omp_3(s,t)$ that,
upon the arrival of $j\in R$, assigns a cost to each \emph{feasible} augmentation
and performs one of minimum cost:
\begin{enumerate}
  \item A direct match (i.e., augmenting along a $1$-path $\apath{j}{i}$ where $i$ is free)
        has cost $-1$.
  \item Augmenting along a $3$-path $\apath{j}{x}{y}{i}$ (where $xy$ is an edge of the current
        matching) has cost equal to the reassignment count of the middle vertex
        $y\in R$ \emph{prior} to the move.
        Concretely, if $y$ has been reassigned $Q-1$ times so far, then using this path has cost $Q-1$.
\end{enumerate}
Since direct matches have cost $-1$ and $3$-paths have nonnegative cost, \lcp always prioritizes direct matches whenever one is available.
Among all feasible augmenting paths, \lcp chooses one of minimum cost, using an arbitrary deterministic tie-breaking rule fixed in advance. The analysis below does not depend on which such rule is used.
A key property of \lcp is that, for every such tie-breaking rule, it never reassigns any offline vertex more than once (Lemma~\ref{lemma:lcp-spgreedy}).
\end{definition}

We focus mainly on the regime $s=1$, where each offline vertex in $L$ can be reassigned at most once. In this regime, the competitive ratio of \lcp interpolates between $1/2$ for $t=0$ (when \lcp coincides with the deterministic greedy algorithm for online matching) and $2/3$ for $t=\infty$.

\begin{algorithm}[ht!]
  \SetAlgoLined
  \DontPrintSemicolon
  \KwIn{Instance $(G,\pi)$ with $G=(L,R,E)$; online budget $t$}
  \KwOut{A matching $M$}
  $M\gets \emptyset$\;
  \ForEach{arrival $j\in R$}{
  Initialize reassignment counter $c_j \gets 0$\;
    \eIf{$j$ has a free neighbor $i\in L$}{
      match $j$ to any such $i$ (Set $M\gets M\cup \{ij\}$)
    }{
      among all feasible augmenting $3$-paths $\apath{j}{x}{y}{i}$, choose one minimizing $c_y$\;
      \lIf{found}{
        Set $M\gets \bigl(M\setminus\{xy\}\bigr)\cup\{jx,yi\}$ and $c_y \gets c_y + 1$
      }
    }
  }
  \caption{Lowest-Cost-Path (\lcp)}
  \label{alg:LCP}
\end{algorithm}

Although \lcp is defined for $\omp_3(s,t)$, Algorithm~\ref{alg:LCP}
only tracks online reassignment counters. Lemma~\ref{lemma:lcp-spgreedy}
shows no offline counter is needed: \lcp never reassigns any offline
vertex twice.

The algorithm can be implemented in polynomial time per arrival. Direct matches
are found by scanning the neighbors of the arriving vertex. Feasible
length-$3$ augmenting paths are found by scanning each offline neighbor $x$ of
the arriving vertex, the online vertex $y$ currently matched to $x$, and the
free offline neighbors of $y$. The same enumeration applies to the weighted
algorithms in Section~\ref{sec:offline-weighted}.

\begin{lemma}\label{lemma:lcp-spgreedy}
For every offline budget $s \ge 1$ and online budget $t$, \lcp never reassigns an offline vertex more than once. Consequently, \lcp is an \spgreedy algorithm for $\omp_3(s,t)$.
\end{lemma}
\begin{proof}
Suppose, for a contradiction, that \lcp reassigns some offline vertex more than once. Let $j$ be the first arrival at which this happens, and let $\apath{j}{x}{y}{i}$ be the $3$-path chosen by \lcp. Thus $x$ has already been reassigned before the arrival of $j$, and \lcp reassigns $x$ again along this path.

Since $xy$ is the current matched edge just before $j$ is processed and $x$ has already been reassigned, the edge $xy$ must have been created at the arrival of $y$, by a previous $3$-path augmentation in which $x$ changed its partner to $y$. Indeed, an already matched offline vertex can acquire a new online partner only when that online vertex arrives.

At the arrival of $y$, \lcp chose a $3$-path augmentation rather than a direct match. Since direct matches have cost $-1$ and $3$-paths have nonnegative cost, $y$ had no free offline neighbor when it arrived. By the basic properties of the model (stated in the preliminaries), online neighborhoods are revealed once and remain fixed, and offline vertices never become free after being matched. Hence $y$ cannot have a free offline neighbor later.

This contradicts the fact that the path $\apath{j}{x}{y}{i}$ uses the edge $yi$ with $i$ free just before $j$ is processed. Therefore, \lcp never reassigns an offline vertex more than once. Since \lcp chooses a minimum-cost feasible augmentation whenever one exists, it is \spgreedy.
\end{proof}
\begin{theorem}\label{thm:lcp-infinity}
For $t=\infty$, \lcp attains a competitive ratio of $2/3$ in $\omp_3(1,\infty)$.
\end{theorem}

\begin{theorem}\label{thm:lcp-ratio}
For any finite $t\in \mathbb{N}$, \lcp has a competitive ratio of $(2\cdot 2^t-1)/(3\cdot 2^t-1)$ in $\omp_3(1,t)$.
\end{theorem}

The proofs of Theorems~\ref{thm:lcp-infinity} and~\ref{thm:lcp-ratio} are given in the remainder of this section.
Since \lcp never reassigns an offline vertex more than once, the same competitive ratio holds for \lcp in $\omp_3(s,t)$ for every $s\ge 1$.

The analysis of these guarantees uses a primal--dual approach. We use the
middle-edge lemma (Lemma~\ref{lemma:middle}) from the preliminaries, together
with a structural description of the matching history. Before proving the
primal--dual bounds, we record this structural description. It is used in the
proof of Theorem~\ref{thm:lcp-ratio} and later in the upper-bound construction.

The special case $t=1$ admits a shorter self-contained proof; we include it in Appendix~\ref{subsec:lcp-one-one}, while the main text proves Theorem~\ref{thm:lcp-ratio} directly for all finite $t$.

\subsection{History graph and component structure}\label{subsec:history-graphs}

We will use the same structural description of the matching history in the proof of Theorem~\ref{thm:lcp-ratio} and in our upper-bound construction. We state it here as a standalone tool.

Consider any execution of an algorithm for $\omp_3$ in which no offline
vertex is reassigned more than once (equivalently, each offline vertex in $L$
uses at most one unit of reassignment budget). Besides the current matching, we
track the graph formed by all edges that have ever belonged to the
matching. If $M_k$ denotes the matching after the first $k$ arrivals, we
define the history graph
\[
H_k:=\bigcup_{\ell=1}^k M_\ell.
\]
Note that $H_k$ only grows with $k$, and (since we only use length-$3$
augmenting paths and no offline vertex is reassigned more than once) an edge
that leaves the matching can never re-enter it later.

We will show that, provided no offline vertex uses more than one unit of reassignment budget, the connected components of $H_k$ belong to a specific family of rooted trees.
To describe them, we define a family of rooted trees $T^Q$ as follows. Let $T^0$ be a single edge, viewed as a rooted bipartite tree with the root on the online side. For $Q\ge 1$, let $T^Q$ be the rooted tree obtained by taking a root vertex $v$, attaching to $v$ exactly one path of length $1$, and additionally attaching $Q$ vertex-disjoint paths of length $2$ to $v$.

\begin{lemma}[Tree transition]\label{lem:tree-transition}
Let $Q\ge 1$. Let $C$ be a rooted tree isomorphic to $T^{Q-1}$, with root
$y\in R$. Suppose that $x\in L$ is the unique leaf of $C$ adjacent to $y$, and
that the edge $xy$ belongs to the current matching.

Let $j\in R$ be the newly arrived online vertex, and let $i\in L$ be an offline
vertex that is unmatched just before $j$ is processed. Suppose that $i$ and
$j$ do not belong to $C$. If the algorithm augments along the path $\apath{j}{x}{y}{i}$, then
\[
    C' := C\cup\{xj,yi\}
\]
is isomorphic to $T^Q$. Its root is still $y$ (i.e., $\rootof(C')=y$), its unique leaf in $L$ adjacent
to the root is $i$, and the edge $yi$ belongs to the new matching.
\end{lemma}

\begin{proof}
The augmentation replaces the matched edge $xy$ by the two matched edges
$xj$ and $yi$. The edge $xj$ extends the old length-$1$ branch $yx$ into the
length-$2$ branch $yxj$, and the edge $yi$ creates a new length-$1$ branch at
the same root $y$. All other branches of $C$ are unchanged. Hence $C'$ is
isomorphic to $T^Q$. The root remains $y$ (i.e., $\rootof(C')=y$), the unique leaf in $L$ adjacent to
the root is $i$, and the edge $yi$ belongs to the new matching.
\end{proof}

We now apply this transition property to the history graph. Fix an instance $(G,\pi)$ and consider an arbitrary execution of some algorithm where no offline vertex uses more than one unit of reassignment budget. For $k=0,1,\ldots,n$, where $n=|R|$, let $M_k$ and $H_k$ be defined as above.

\begin{lemma}\label{lem:history-components}
For every $k$, each connected component of $H_k$ is isomorphic to
$T^Q$ for some $Q\ge0$. If the online budget is finite and equal to
$t$, then necessarily $Q\le t$. If a component is isomorphic to $T^Q$,
we say it is of type $T^Q$. Moreover, if a component is of type $T^Q$,
the root has been reassigned exactly $Q$ times up to time $k$.

If $C$ is isomorphic to $T^Q$, then $\rootof(C)$ has degree $Q+1$ in $H_k$ and has
been reassigned exactly $Q$ times (equivalently, its matched neighbor changes
exactly $Q$ times over the execution). All
non-leaf vertices in $L\cap V(C)$ have degree $2$ and have been reassigned once. In particular, when the offline budget is $s=1$, they are exhausted. There is
exactly one leaf in $L$, and it is adjacent to the root. The edges of $M_k$ in
$C$ are exactly the edges incident to the leaves of $C$. Moreover, among the
matching edges in $C$, the only one whose offline endpoint has not been reassigned is
the edge incident to the unique leaf in $L$. In particular, when the offline budget is $s=1$, this is the only matching edge whose offline endpoint is not exhausted.
\end{lemma}

\begin{proof}
We proceed by induction on $k$. For $k=0$, the graph $H_0$ is empty. Assume the
lemma holds for $H_{k-1}$, and let $j\in R$ be the online vertex arriving at
step $k$. If $j$ remains unmatched, then $H_k=H_{k-1}$, and there is nothing to
prove.

If $\mathcal A$ matches $j$ directly to an unmatched neighbor $i\in L$, the
edge $ij$ creates a new component isomorphic to $T^0$. The vertex $j$ is the
only online vertex in this component, and hence $\rootof(C)=j$. The edge $ij$ is
the only matching edge in the component, and it is incident to the unique leaf
in $L$. Since $j$ is newly arrived, its incident matching edge has changed $0$
times so far. Thus the lemma holds in this case.

It remains to consider the case in which $\mathcal A$ matches $j$ by augmenting
along a $3$-path $\apath{j}{x}{y}{i}$. Let $C$ be the component of $H_{k-1}$
containing the matching edge $xy$.

Since the execution reassigns no offline vertex more than once, the
offline vertex $x$ cannot have been reassigned before this augmentation:
otherwise the augmentation $\apath{j}{x}{y}{i}$ would reassign $x$ a second
time. By the induction hypothesis, $C$ is isomorphic to $T^Q$ for some
$Q$, and the only matching edge in $C$ whose offline endpoint has not
been reassigned is the edge incident to the unique leaf in $L$.
Therefore $x$ is this leaf, $y$ is the root of $C$, and $xy$ is the
current matching edge incident to the unique leaf in $L$.

If the online budget is finite and equal to $t$, then $Q\le t-1$,
because the augmentation is feasible and reassigns the root $y$ one
additional time.

By Lemma~\ref{lem:tree-transition}, the new component
$C\cup\{xj,yi\}$ is isomorphic to $T^{Q+1}$. Its root is still $y$ (i.e., $\rootof(C\cup\{xj,yi\})=y$), which has
its $(Q+1)$-st online reassignment; its unique
leaf in $L$ adjacent to the root is $i$, and the edge $yi$ belongs to the new
matching.

The matching edges in the new component are exactly the edges incident to its
leaves: the old edge $xy$ is removed from the matching, and the new matching
edges are $xj$ and $yi$. All non-leaf vertices in $L$ have degree $2$ and have been reassigned once.
Thus, when the offline budget is $s=1$, they are exhausted. All other components are unchanged, so the lemma follows.
\end{proof}

In the remainder of the paper, we use this structure twice. First, in the primal--dual analysis of \lcp, the type $T^Q$ determines the dual values assigned to the active offline leaf and to the online vertices created at level $Q$. Second, in the upper bound analysis, the adversary forces the algorithm to build the same history components.

\subsection{\texorpdfstring{The case $\omp_3(1,\infty)$}{The case OMP3(1,infinity)}}

We start with the case $t=\infty$, which admits a short primal--dual proof. In this regime, no online vertex is ever exhausted, so frozen edges arise only from the offline budget $s=1$, and the middle-edge lemma (Lemma~\ref{lemma:middle}) is particularly convenient for establishing dual feasibility. The dual assignment is constructed online, together with a final postprocessing step on the edges of the resulting matching.

\begin{proof}[Proof of Theorem~\ref{thm:lcp-infinity}]
Fix an instance $(G,\pi)$ with $G=(L,R,E)$ and let $\mathcal A$ denote the \lcp
algorithm for $\omp_3(1,\infty)$. Let $M$ be the final matching produced by
$\mathcal A$. We construct a dual solution $(\alpha,\beta)$ for $D_G$, starting
from the all-zero solution. The online updates and the final postprocessing
step are given in Algorithm~\ref{alg:dual-lcp-1inf}.

\begin{algorithm}[ht!]
\SetAlgoLined
\DontPrintSemicolon
\KwIn{Instance $(G,\pi)$, algorithm $\mathcal A$ (\lcp for $\omp_3(1,\infty)$)}
\KwOut{Dual vector $(\alpha_i,\beta_j)_{i\in L,\,j\in R}$}
\lForAll{$i\in L, j\in R$}{Set $\alpha_i\gets 0$, $\beta_j\gets 0$}
\ForAll(\tcp*[f]{dual update steps}){$j\in R$ in arrival order, while running $\mathcal A$}{
  \lIf{$\mathcal A$ matches $j$ directly to some $i\in L$}{
    Set $\alpha_i\gets \frac12$ and $\beta_j\gets 1$
  }
  \lElseIf{$\mathcal A$ matches $j$ using a feasible $3$-path $\apath{j}{x}{y}{i}$}{
    Set $\alpha_i\gets \frac12$, $\alpha_x\gets 1$, $\beta_j\gets \frac12$ and $\beta_y\gets 1$
  }
}
Let $M$ be the final matching constructed by $\mathcal A$\;
\ForAll(\tcp*[f]{postprocessing step}){$ij\in M$ with $j\in R$}{
  \If{neither $i$ nor $j$ is exhausted, and $j$ has no free neighbor in $L$ at the end}{
    Set $\alpha_i\gets 1$ and $\beta_j\gets \frac12$\;
  }
}
\caption{$\dual$ for $\omp_3(1,\infty)$}
\label{alg:dual-lcp-1inf}
\end{algorithm}

We first verify the change in the dual objective at every step. Each time
$\mathcal A$ matches an arrival $j$ directly, the dual objective increases by
exactly $3/2$: $j$ is new, so $\beta_j$ increases from $0$ to $1$, and the
chosen neighbor $i$ is free, so $\alpha_i$ increases from $0$ to $1/2$.

Each time the algorithm matches an arrival $j$ using a $3$-path $\apath{j}{x}{y}{i}$, the augmenting path must be feasible, which implies that the middle edge $xy$ is not frozen. Since $x$ is an offline vertex and offline vertices have budget $1$, the vertex $x$ has not been reassigned before this update.

We claim that, immediately before the update,
\[
(\alpha_x,\beta_y)=\left(1/2,1\right).
\]
Indeed, since $x$ has not been reassigned before, the edge $xy$ was created either when $y$ was matched directly to $x$, or by an earlier $3$-path augmentation in which $x$ was the free endpoint. In both cases, Algorithm~\ref{alg:dual-lcp-1inf} assigned $(\alpha_x,\beta_y)=(1/2,1)$. Moreover, since $x$ has not been used as the middle offline vertex before the present update, $\alpha_x$ has not subsequently changed. This proves the claim.

After the update, the values are
\[
    (\alpha_i,\alpha_x,\beta_j,\beta_y)=\left(1/2,1,1/2,1\right).
\]
Thus the dual objective increases by exactly $3/2$: the new online vertex $j$ contributes $1/2$, the free offline endpoint $i$ contributes $1/2$, and $\alpha_x$ increases from $1/2$ to $1$.

The postprocessing step only redistributes $1/2$ within a matched edge $ij\in M$, changing $(\alpha_i,\beta_j)$ from $(1/2,1)$ to $(1,1/2)$, and hence it does not change the total dual objective value. Therefore the final dual objective is at most $\frac32\,|M|$.

To conclude via Lemma~\ref{lemma:pd-template} with $\rho=2/3$, it remains to show
that $(\alpha,\beta)$ is feasible for $D_G$, namely that
$\alpha_i+\beta_j\ge 1$ for every edge $ij\in E$.

Fix an arbitrary edge $ij\in E$. We rule out all possibilities with
$\alpha_i+\beta_j<1$. The only possible final values are
$\alpha_i\in\{0,1/2,1\}$ and $\beta_j\in\{0,1/2,1\}$. Hence the only
pairs to consider are
\[
    (0,0),\qquad
    \left(0,1/2\right),\qquad
    \left(1/2,0\right).
\]

\begin{itemize}
\item $(\alpha_i,\beta_j)=(0,0)$.
This would mean that both $i$ and $j$ are free in the final matching $M$.
Since matched vertices never become unmatched, $i$ was free when $j$ arrived.
Thus $\mathcal A$ would have matched $j$ directly to $i$, a contradiction.

\item $(\alpha_i,\beta_j)=(0,\frac12)$.
If $\beta_j=\frac12$ was assigned during the online phase, then $j$ was matched
using a $3$-path. This is impossible, since $\alpha_i=0$ implies that $i$ is
free in the final matching, hence was free when $j$ arrived, and \lcp prioritizes
direct matches. If instead $\beta_j=\frac12$ was assigned during postprocessing,
then $j$ had no free neighbor in $L$ at the end, again contradicting that
$i\in N(j)$ is free.

\item $(\alpha_i,\beta_j)=(\frac12,0)$.
Since $\alpha_i\ne 0$, the vertex $i$ is matched in $M$; let $y\in R$ be such
that $iy\in M$. Since $\alpha_i=\frac12$ after postprocessing, the edge $iy$ was
not modified during the postprocessing step. Moreover, $i$ is not exhausted, and
$y$ is not exhausted because $t=\infty$. Hence $iy$ is not frozen. Therefore,
the fact that the postprocessing step did not act on $iy$ implies that $y$ has
a free neighbor $k\in L$ at the end.

Since $\beta_j=0$, the vertex $j$ is free in $M$. Thus $\apath{j}{i}{y}{k}$ is an
$M$-augmenting $3$-path whose middle edge $iy$ is not frozen. This contradicts
the middle-edge lemma, Lemma~\ref{lemma:middle}.
\end{itemize}

Thus $\alpha_i+\beta_j\ge 1$ for every edge $ij\in E$, so $(\alpha,\beta)$ is
feasible for $D_G$. Lemma~\ref{lemma:pd-template} yields $|M|\ge \frac23\,\nu(G)$,
proving that \lcp is $\frac23$-competitive for $\omp_3(1,\infty)$.\qedhere
\end{proof}

\subsection{\texorpdfstring{The case $\omp_3(1,t)$ for finite $t$}{The case OMP3(1,t) for finite t}}
\label{subsec:lcp-finite-t}

We now analyze \lcp for $\omp_3(1,t)$ with finite $t$. By
Lemma~\ref{lem:history-components}, every history component created by
\lcp has type $T^Q$ for some $Q\le t$. A direct match creates a
component of type $T^0$, and a feasible $3$-path augmentation increases
its type from $T^{Q-1}$ to $T^Q$. When \lcp uses a feasible $3$-path,
the minimum-cost rule chooses one whose middle online vertex has the
smallest possible reassignment count. Our dual updates are defined by two sequences
\[
    a=(a_0,\ldots,a_t), \qquad b=(b_0,\ldots,b_t),
\]
with $b_0=1$. These sequences specify the values assigned when a
history component grows. When a component of type $T^Q$ is formed, its
unique offline leaf receives value $a_Q$, and the new online leaf
receives value $b_Q$. The other offline vertices in the component have
already been used as middle vertices and receive value $1$.

Figure~\ref{fig:dual-in-components} illustrates the resulting dual values
in the first few component types. We prove Theorem~\ref{thm:lcp-ratio} with these updates.

Note that the formula gives $1/2$ for $t=0$, gives $3/5$ for $t=1$, and converges to the $t=\infty$ guarantee of Theorem~\ref{thm:lcp-infinity}. Appendix~\ref{subsec:lcp-one-one} gives a self-contained proof of the special case $t=1$, which may be useful to follow the general case.

\begin{proof}[Proof of Theorem~\ref{thm:lcp-ratio}]
If $t=0$, no online vertex can be reassigned, hence no length-$3$ augmentation is feasible. In this case, \lcp reduces to the
standard greedy matching algorithm without recourse, which guarantees a
competitive ratio of
\[
    \frac12=\frac{2\cdot 2^0-1}{3\cdot 2^0-1}.
\]

Let $t\ge 1$. Fix an instance $(G,\pi)$ with $G=(L,R,E)$ and let
$\mathcal A$ be the \lcp algorithm for $\omp_3(1,t)$. Let $M$ be the final
matching produced by $\mathcal A$ on $(G,\pi)$. We construct, alongside the
matching built by $\mathcal A$, a dual solution $(\alpha,\beta)$ for the dual
LP $D_G$, starting from the zero solution. The construction is given in
Algorithm~\ref{alg:dual-lcp-1t}. The algorithm depends on two sequences
$a=(a_0,a_1,\dots,a_t)$ and $b=(b_0,b_1,\dots,b_t)$, with $b_0=1$. These
sequences will be chosen later to satisfy the primal--dual template,
Lemma~\ref{lemma:pd-template}.

\begin{algorithm}[ht!]
  \SetAlgoLined
  \DontPrintSemicolon
  \KwIn{Instance $G=(L,R,E)$, algorithm $\mathcal A$}
  \KwOut{Dual vector $(\alpha_i,\beta_j)_{i\in L,j\in R}$}
  \lForAll{$i\in L$, $j\in R$}{Set $\alpha_i\gets 0$, $\beta_j\gets 0$}
  \ForAll(\tcp*[f]{dual update steps}){$j\in R$ in arrival order, while running $\mathcal A$}{
    \lIf{$\mathcal A$ matches $j$ directly to $i\in L$}{
      Set $(\alpha_i,\beta_j)\gets(a_0,b_0)$
    }
    \ElseIf{$\mathcal A$ matches $j$ using the $3$-path $\apath{j}{x}{y}{i}$}{
      Let $Q$ be the number of times that $y$ has been reassigned after
      augmenting along $\apath{j}{x}{y}{i}$\;
      Set $(\alpha_i,\alpha_x,\beta_j,\beta_y)\gets(a_Q,1,b_Q,1)$
    }
  }
  Let $M$ be the final matching constructed by $\mathcal A$\;
  \ForAll(\tcp*[f]{postprocessing step}){$ij\in M$ with $j\in R$}{
    \If{$i$ and $j$ are not exhausted and $j$ has no free neighbor in $L$}{
      Flip $i$'s and $j$'s values, i.e.,
      $(\alpha_i,\beta_j)\gets(1,\alpha_i)$\;
    }
  }
  \caption{$\dual(a,b)$}
  \label{alg:dual-lcp-1t}
\end{algorithm}

By Lemma~\ref{lem:history-components}, if an edge $ij\in M$ is considered in the
postprocessing step, then $i$ is the unique offline leaf of its component and
$j$ is the root of that component. Indeed, all other matched offline vertices
in the component are exhausted. Since the root of every component is formed by
a direct match, the algorithm assigned it $\beta_j=b_0=1$. Thus the assignment
$(1,\alpha_i)$ only redistributes the dual values on $ij$ and does not change
their sum.

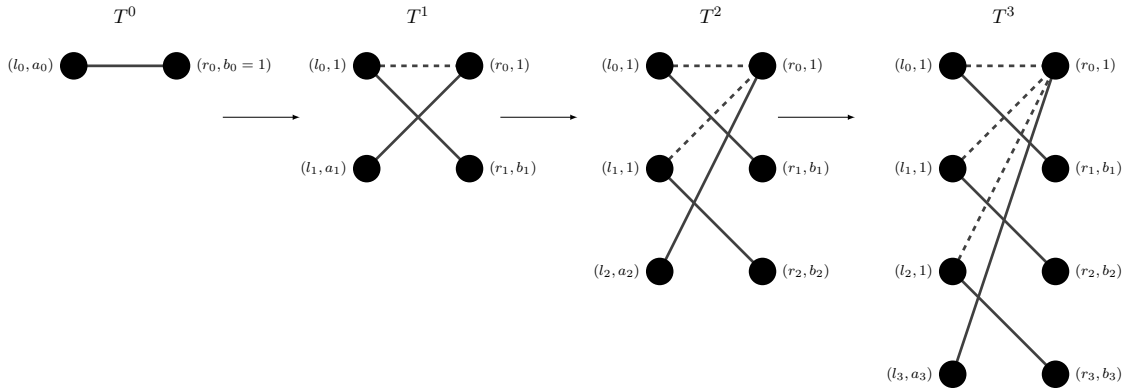
\begin{figure}[ht]
\centering
\scalebox{0.68}{
\begin{tikzpicture}

\begin{scope}[xshift=0cm]
\node at (1,3) {$T^0$};
\Vertex[x=0,y=2,size=.5,color=black,style=solid,label={$(l_0,a_0)$},position=left]{T1l1}
\Vertex[x=2,y=2,size=.5,color=black,style=solid,label={$(r_0,b_0=1)$},position=right]{T1r1}
\Edge[style=solid](T1l1)(T1r1)
\end{scope}

\draw[->,>=latex] (2.9,1) -- (4.4,1);

\begin{scope}[xshift=5.7cm]
\node at (1,3) {$T^1$};
\Vertex[x=0,y=2,size=.5,color=black,style=solid,label={$(l_0,1)$},position=left]{T2l1}
\Vertex[x=0,y=0,size=.5,color=black,style=solid,label={$(l_1,a_1)$},position=left]{T2l2}
\Vertex[x=2,y=2,size=.5,color=black,style=solid,label={$(r_0,1)$},position=right]{T2r1}
\Vertex[x=2,y=0,size=.5,color=black,style=solid,label={$(r_1,b_1)$},position=right]{T2r2}
\Edge[style=dashed](T2l1)(T2r1)
\Edge[style=solid](T2l1)(T2r2)
\Edge[style=solid](T2l2)(T2r1)
\end{scope}

\draw[->,>=latex] (8.3,1) -- (9.8,1);

\begin{scope}[xshift=11.4cm]
\node at (1,3) {$T^2$};
\Vertex[x=0,y=2,size=.5,color=black,style=solid,label={$(l_0,1)$},position=left]{T3l1}
\Vertex[x=0,y=0,size=.5,color=black,style=solid,label={$(l_1,1)$},position=left]{T3l2}
\Vertex[x=0,y=-2,size=.5,color=black,style=solid,label={$(l_2,a_2)$},position=left]{T3l3}
\Vertex[x=2,y=2,size=.5,color=black,style=solid,label={$(r_0,1)$},position=right]{T3r1}
\Vertex[x=2,y=0,size=.5,color=black,style=solid,label={$(r_1,b_1)$},position=right]{T3r2}
\Vertex[x=2,y=-2,size=.5,color=black,style=solid,label={$(r_2,b_2)$},position=right]{T3r3}
\Edge[style=dashed](T3l1)(T3r1)
\Edge[style=solid](T3l1)(T3r2)
\Edge[style=dashed](T3l2)(T3r1)
\Edge[style=solid](T3l2)(T3r3)
\Edge[style=solid](T3l3)(T3r1)
\end{scope}

\draw[->,>=latex] (13.7,1) -- (15.2,1);

\begin{scope}[xshift=17.1cm]
\node at (1,3) {$T^3$};
\Vertex[x=0,y=2,size=.5,color=black,style=solid,label={$(l_0,1)$},position=left]{T4l1}
\Vertex[x=0,y=0,size=.5,color=black,style=solid,label={$(l_1,1)$},position=left]{T4l2}
\Vertex[x=0,y=-2,size=.5,color=black,style=solid,label={$(l_2,1)$},position=left]{T4l3}
\Vertex[x=0,y=-4,size=.5,color=black,style=solid,label={$(l_3,a_3)$},position=left]{T4l4}
\Vertex[x=2,y=2,size=.5,color=black,style=solid,label={$(r_0,1)$},position=right]{T4r1}
\Vertex[x=2,y=0,size=.5,color=black,style=solid,label={$(r_1,b_1)$},position=right]{T4r2}
\Vertex[x=2,y=-2,size=.5,color=black,style=solid,label={$(r_2,b_2)$},position=right]{T4r3}
\Vertex[x=2,y=-4,size=.5,color=black,style=solid,label={$(r_3,b_3)$},position=right]{T4r4}
\Edge[style=dashed](T4l1)(T4r1)
\Edge[style=solid](T4l1)(T4r2)
\Edge[style=dashed](T4l2)(T4r1)
\Edge[style=solid](T4l2)(T4r3)
\Edge[style=dashed](T4l3)(T4r1)
\Edge[style=solid](T4l3)(T4r4)
\Edge[style=solid](T4l4)(T4r1)
\end{scope}

\end{tikzpicture}
}
\caption{Components $T^0, T^1, T^2, T^3$ obtained by augmenting along the paths
$\apath{r_0}{l_0}$, $\apath{r_1}{l_0}{r_0}{l_1}$, $\apath{r_2}{l_1}{r_0}{l_2}$, and $\apath{r_3}{l_2}{r_0}{l_3}$, respectively. Each
vertex is labeled with its name and the dual values immediately after the
component is formed.}
\label{fig:dual-in-components}
\end{figure}

The sequences $a$ and $b$ specify the dual values assigned to vertices at each level of the history components. Requiring each successful update to increase the dual objective by at most $1/\rho$, and enforcing dual feasibility on the edges of the history components, yields the following recurrences.

Let $0<\rho\le 1$ be the competitive ratio we aim \lcp to achieve, and define
$\lambda:=1/\rho-1$. To use Lemma~\ref{lemma:pd-template}, we require that each
dual update increases the dual objective by exactly $1/\rho=1+\lambda$, while the
postprocessing step preserves the total dual value.

A direct-matching update increases the dual objective by $1+a_0$. A $3$-path
update that creates a component of type $T^Q$ increases the dual objective by
$a_Q+b_Q+1-a_{Q-1}$. Thus we impose
\begin{equation}
\begin{aligned}
a_0 &= \lambda,\\
a_Q+b_Q-a_{Q-1} &= \lambda \qquad \forall\, Q\in\{1,\ldots,t\}.
\label{eq:recurrence1}
\end{aligned}
\end{equation}

It remains to enforce dual feasibility. Fix an edge $ij\in E$, with $i\in L$ and $j\in R$. We must show $\alpha_i+\beta_j\ge 1$.

We first list the possible final dual values.
An offline vertex has value in $\{0,1,a_0,\ldots,a_t\}$: it has value $0$ if it is unmatched, value $1$ if it is exhausted or postprocessed, and value $a_Q$ if it is the unique offline leaf of a component of type $T^Q$ that is not postprocessed.
An online vertex has value in $\{0,1,b_1,\ldots,b_t,a_0,\ldots,a_{t-1}\}$: it has value $0$ if it is unmatched, value $b_Q$ if it is an online leaf created in a $3$-path update forming $T^Q$, value $a_Q$ if it receives this value in the postprocessing step, and value $1$ if it is the root of a component before postprocessing.

We now rule out all pairs of values that could violate dual feasibility.

\begin{enumerate}[leftmargin=*, itemsep=2pt]
\item If $\alpha_i=1$ or $\beta_j=1$, then $\alpha_i+\beta_j\ge 1$.

\item The pair $(\alpha_i,\beta_j)=(0,0)$ is impossible. Otherwise both
$i$ and $j$ are unmatched in the final matching, so $ij$ is an augmenting
path of length $1$, contradicting the fact that \lcp leaves no direct
augmenting path.

\item The pair $(\alpha_i,\beta_j)=(0,b_Q)$, with $Q\ge 1$, is impossible. If
$\beta_j=b_Q$, then $j$ was matched by a $3$-path when it arrived, so it had no
free neighbor at that time. But $\alpha_i=0$ means that $i$ is unmatched in the
final matching, hence was free when $j$ arrived.

\item Consider $(\alpha_i,\beta_j)=(a_Q,0)$. Let $y$ be the vertex matched to
$i$ in $M$. Then $i$ is the unique offline leaf of a component of type $T^Q$,
and $\rootof(C)=y$. If $Q<t$, then neither endpoint of $iy$ is exhausted. Since
$i$ has the free neighbor $j$, the path structure would give an
augmenting 3-path through the middle edge $iy$ unless $y$ has no free
neighbor. Because neither endpoint of $iy$ is exhausted, Lemma~\ref{lemma:middle}
implies that $y$ has no free neighbor at the end. Hence the edge $iy$ should have been
postprocessed, contradicting $\alpha_i=a_Q$. Therefore the only remaining case
is $Q=t$, and we impose
\begin{equation}
a_t=1.
\label{eq:atone}
\end{equation}

\item Consider $\beta_j=a_Q$ for some $Q\in\{0,\ldots,t-1\}$. This value can
only arise from the postprocessing step. Thus $j$ is matched and all neighbors
of $j$ are matched at the end, so any neighbor $i$ of $j$ satisfies
$\alpha_i\ne 0$. It is therefore enough to impose
\begin{equation}
a_Q\ge \frac12 \qquad \forall\, Q\in\{0,\ldots,t\}.
\label{eq:ahalf}
\end{equation}

\item Finally, consider $(\alpha_i,\beta_j)=(a_Q,b_{Q'})$ with $Q<t$ and $1\le Q'\le t$. We will use the following claim.

\begin{claim}\label{clm:Qprime-bound}
Suppose that $ij\in E$, $\alpha_i=a_Q$ with $Q<t$, and $\beta_j=b_{Q'}$ with $Q'\ge 1$. Then $Q'\le Q+1$.
\end{claim}

\begin{proof}
Let $y$ be the vertex matched to $i$ in the final matching. Since $\alpha_i=a_Q$ and $Q<t$, the vertex $i$ is the unique offline leaf of a component of type $T^Q$, whose root is $y$. In particular, $i$ has not been reassigned and $y$ is not exhausted.

Moreover, the matched edge $iy$ was not postprocessed. Hence, by the definition of the postprocessing step, $y$ has a free neighbor $f\in L$ at the end of the execution. Since offline vertices never become free after being matched, $f$ was already free when $j$ arrived.

We first show that $iy$ was already present when $j$ arrived. Otherwise, since $i$ is never reassigned and is matched to $y$ at the end, the vertex $i$ was free when $j$ arrived. Since $ij\in E$, \lcp would have matched $j$ directly to $i$, contradicting $\beta_j=b_{Q'}$.

Thus, when $j$ arrived, the path $\apath{j}{i}{y}{f}$ was a feasible augmenting $3$-path. At the end of the algorithm, the root $y$ has been reassigned exactly $Q$ times, so at the time $j$ arrived it had been reassigned at most $Q$ times. Since $\beta_j=b_{Q'}$, \lcp chose a path whose middle online vertex had already been reassigned $Q'-1$ times. By the minimum-cost choice of \lcp, $Q'-1\le Q$, and hence $Q'\le Q+1$.
\end{proof}

We will choose the sequence $(a_Q)_{Q=0}^t$ to be nondecreasing. Under this
condition, the smallest possible offline value in this case is $a_{Q'-1}$. It is enough to impose
\begin{equation}
a_{r-1}+b_r=1 \qquad \forall\, r\in\{1,\ldots,t\}.
\label{eq:recurrence2}
\end{equation}

\end{enumerate}

The preceding case analysis shows that it is enough to choose nonnegative sequences $a=(a_0,\ldots,a_t)$ and $b=(b_0,\ldots,b_t)$ satisfying the following conditions:
\begin{align}
a_0 &= \lambda, \label{eq:ab-direct}\\
a_Q+b_Q-a_{Q-1} &= \lambda
        &&\forall Q\in\{1,\ldots,t\}, \label{eq:ab-update}\\
a_t &= 1, \label{eq:ab-last}\\
a_Q &\ge 1/2
        &&\forall Q\in\{0,\ldots,t\}, \label{eq:ab-half}\\
a_{r-1}+b_r &\ge 1
        &&\forall r\in\{1,\ldots,t\}. \label{eq:ab-feas}
\end{align}
We will choose the parameters so that \eqref{eq:ab-feas} holds with equality.

We now solve explicitly for $a$ and $b$. Combining \eqref{eq:ab-update} and \eqref{eq:ab-feas} gives, for all $Q\in\{1,\ldots,t\}$,
\[
    a_Q=2a_{Q-1}+(\lambda-1),
    \qquad a_0=\lambda.
\]
Thus
\[
    a_Q=1-\lambda+2^Q(2\lambda-1)\qquad \forall\, Q\in\{0,\ldots,t\},
\]
and
\[
    b_Q=1-a_{Q-1}
        =\lambda-2^{Q-1}(2\lambda-1)\qquad \forall\, Q\in\{1,\ldots,t\}.
\]

We choose $\lambda$ so that $a_t=1$. This gives
\[
    1=1-\lambda+2^t(2\lambda-1)
    \quad\Longleftrightarrow\quad
    \lambda=\frac{2^t}{2\cdot 2^t-1}.
\]
For this value, $\lambda\ge 1/2$. Hence
$a_Q-a_{Q-1}=2^{Q-1}(2\lambda-1)\ge 0$ for every $Q\in\{1,\ldots,t\}$,
so $(a_Q)_{Q=0}^t$ is nondecreasing. Moreover, $a_Q\ge 1/2$ for all $Q$.
Since $a_t=1$, monotonicity gives $a_{Q-1}\le 1$, and therefore
$b_Q=1-a_{Q-1}\ge 0$ for every $Q$. Thus the dual solution is nonnegative and
feasible.

Since $1+\lambda=1/\rho$, we obtain
\[
    \rho=\frac{2\cdot 2^t-1}{3\cdot 2^t-1}.
\]
With this choice, each dual update increases the dual objective by exactly
$1/\rho$, and the postprocessing step preserves the total dual value. By
Lemma~\ref{lemma:pd-template}, \lcp is
$(2\cdot 2^t-1)/(3\cdot 2^t-1)$-competitive in $\omp_3(1,t)$.
\end{proof}

\section{Upper Bounds}

We now prove matching upper bounds. We first give two small adaptive constructions showing that no deterministic algorithm can beat $2/3$ in the unrestricted model $\omp_3$, and that no deterministic algorithm can beat $3/5$ when the online budget is $1$, even with infinite offline budget. We then prove the exact upper bound $\gamma_t$ for finite online budgets via an adaptive adversarial construction.

\subsection{Global barriers for \texorpdfstring{$\omp_3(\infty,\infty)$}{OMP3(infty,infty)} and \texorpdfstring{$\omp_3(\infty,1)$}{OMP3(infty,1)}: \texorpdfstring{$2/3$}{2/3} and \texorpdfstring{$3/5$}{3/5}}
\label{sec:simpleupperbounds}

In this section we show that no deterministic algorithm can beat the ratios achieved by \lcp in the endpoint regimes $t=\infty$ and $t=1$. These bounds apply even if the offline side has infinite reassignment budget.

\begin{theorem}\label{thm:upperboundgeneral}
\begin{enumerate}[label=(\alph*), leftmargin=*, itemsep=0pt]
\item The competitive ratio of any deterministic algorithm for
$\omp_3$ is at most $\frac23$.
\item The competitive ratio of any deterministic algorithm for
$\omp_3(\infty,1)$ is at most $\frac35$.
\end{enumerate}
\end{theorem}

We prove the two parts of Theorem~\ref{thm:upperboundgeneral} separately. Both
proofs use small adaptive adversarial instances. The adversary reveals each
neighborhood after observing the current matching maintained by the algorithm.

\begin{proof}[Proof of Theorem~\ref{thm:upperboundgeneral}(a)]
Fix a deterministic algorithm $\mathcal A$ for $\omp_3$. Let
$L=\{a,b,c\}$. This set is known to the algorithm from the beginning. We reveal
the online vertices adaptively.

\begin{enumerate}[leftmargin=*, itemsep=2pt]
    \item The first arrival is $r_1$, with $N(r_1)=\{a,b\}$. If $\mathcal A$
    leaves $r_1$ unmatched, we stop. Then the optimum is $1$, while
    $\mathcal A$ matches no vertex. Hence the ratio is $0$. Thus we may assume
    that $\mathcal A$ matches $r_1$. By renaming $a$ and $b$, assume that
    $r_1$ is matched to $b$.

    \item The second arrival is $r_2$, with $N(r_2)=\{b,c\}$. If $\mathcal A$
    leaves $r_2$ unmatched, we stop. The graph revealed so far has a matching
    of size $2$, while $\mathcal A$ matches only one vertex. Hence the ratio is
    at most $\frac12$. Thus we may assume that $\mathcal A$ matches $r_2$.

    Let $u$ be the offline vertex that becomes matched for the first time when
    $r_2$ is processed. If $r_2$ is matched directly to $c$, then $u=c$. If
    $r_2$ is matched by the length-$3$ augmenting path $\apath{r_2}{b}{r_1}{a}$, then $u=a$.
    These two cases are shown in Figure~\ref{fig:upper-23}.

    \item The third arrival is $r_3$, with $N(r_3)=\{u\}$. The algorithm cannot
    match $r_3$. Its only neighbor $u$ is already matched. Also, the online
    vertex matched to $u$ has no free neighbor. Hence there is no length-$3$
    augmenting path starting at $r_3$.
\end{enumerate}

The final graph has a matching of size $3$. If $u=a$, take $ar_3$, $br_1$, and
$cr_2$. If $u=c$, take $ar_1$, $br_2$, and $cr_3$. On the other hand,
$\mathcal A$ matches at most two vertices. Therefore
\[
    \frac{|\mathcal A(G,\pi)|}{\nu(G)}
    \le
    \frac{2}{3}.
\]
This gives the upper bound $\frac23$.
\end{proof}

\begin{figure}[ht]
\centering
\scalebox{0.8}{
    \begin{tikzpicture}
        \Vertex[x=0,y=0, size=.5, color=white, style=dashed, label=$a$, position=above]{1}
        \Vertex[x=0,y=-1.5, size=.5, color=black, style=solid, label=$b$, position=above]{2}
        \Vertex[x=0,y=-3, size=.5, color=white, style=dashed, label=$c$, position=above]{3}

        \Vertex[x=2,y=0, size=.5, color=black, label=$r_1$, position=above]{r1}
        \Vertex[x=2,y=-1.5, size=.5, color=white, style=dashed, label=$r_2$, position=above]{r2}

        \Edge[style=solid](r1)(2)
        \Edge[style=dashed](r1)(1)
        \Edge[style=dashed](r2)(2)
        \Edge[style=dashed](r2)(3)

        \Vertex[x=3,y=-1.5,style={color=white}]{flecha1}
        \Vertex[x=5,y=-1.5,style={color=white}]{flecha2}
        \Edge[Direct](flecha1)(flecha2)

        \Vertex[x=6,y=0, size=.5, color=black, style=solid, label={$a=u$}, position=above]{1}
        \Vertex[x=6,y=-1.5, size=.5, color=black, style=solid, label=$b$, position=above]{2}
        \Vertex[x=6,y=-3, size=.5, color=white, style=dashed, label=$c$, position=above]{3}

        \Vertex[x=8,y=0, size=.5, color=black, label=$r_1$, position=above]{r1}
        \Vertex[x=8,y=-1.5, size=.5, color=black, style=solid, label=$r_2$, position=above]{r2}
        \Vertex[x=8,y=-3, size=.5, color=white, style=dashed, label=$r_3$, position=above]{r3}

        \Edge[style=dashed](r1)(2)
        \Edge[style=solid](r1)(1)
        \Edge[style=solid](r2)(2)
        \Edge[style=dashed](r2)(3)
        \Edge[style=dashed](r3)(1)

        \Vertex[x=6.5, y=1, style={color=white}, label={If $r_2$ is matched by a $3$-path}]{r}
        \Vertex[x=9.5, y=-1.5, style={color=white}, label={or}]{r}

        \Vertex[x=11,y=0, size=.5, color=white, style=dashed, label=$a$, position=above]{1}
        \Vertex[x=11,y=-1.5, size=.5, color=black, style=solid, label=$b$, position=above]{2}
        \Vertex[x=11,y=-3, size=.5, color=black, style=solid, label={$c=u$}, position=above]{3}

        \Vertex[x=13,y=0, size=.5, color=black, label=$r_1$, position=above]{r1}
        \Vertex[x=13,y=-1.5, size=.5, color=black, style=solid, label=$r_2$, position=above]{r2}
        \Vertex[x=13,y=-3, size=.5, color=white, style=dashed, label=$r_3$, position=above]{r3}

        \Edge[style=solid](r1)(2)
        \Edge[style=dashed](r1)(1)
        \Edge[style=dashed](r2)(2)
        \Edge[style=solid](r2)(3)
        \Edge[style=dashed](r3)(3)

        \Vertex[x=12.5, y=1, style={color=white}, label={If $r_2$ is matched directly to $c$}]{r}
    \end{tikzpicture}
}
\caption{The two possible cases after processing $r_2$ in the proof of
Theorem~\ref{thm:upperboundgeneral}(a). The vertex $u$ is the offline vertex
that becomes matched for the first time when $r_2$ is processed. The adversary
then reveals $r_3$ with $N(r_3)=\{u\}$.}
\label{fig:upper-23}
\end{figure}
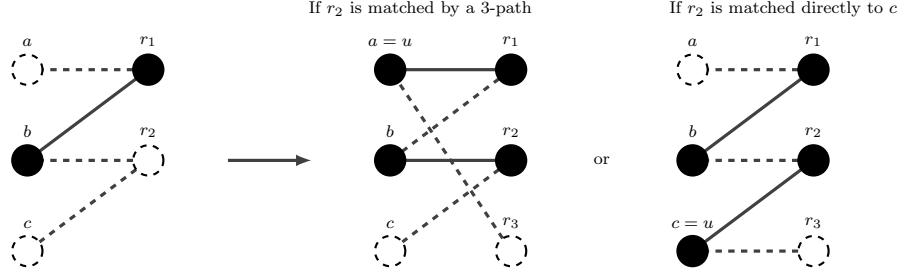

We now prove the upper bound for $\omp_3(\infty,1)$.

\begin{proof}[Proof of Theorem~\ref{thm:upperboundgeneral}(b)]
Fix a deterministic algorithm $\mathcal A$ for $\omp_3(\infty,1)$. Let
$L=\{a_1,a_2,a_3,a_4,a_5\}$. This set is known to the algorithm from the beginning. We
reveal the online vertices adaptively.

\begin{enumerate}[leftmargin=*, itemsep=2pt]
    \item The first arrival is $r_1$, with $N(r_1)=L$. If $\mathcal A$ leaves
    $r_1$ unmatched, we stop. Then the optimum is $1$, while $\mathcal A$
    matches no vertex. Hence the ratio is $0$. Thus we may assume that
    $\mathcal A$ matches $r_1$. By renaming vertices in $L$, assume that $r_1$
    is matched to $a_1$.

    \item The second arrival is $r_2$, with $N(r_2)=L\setminus\{a_1\}$. If
    $\mathcal A$ leaves $r_2$ unmatched, we stop. The graph revealed so far has
    a matching of size $2$, while $\mathcal A$ matches only one vertex. Hence
    the ratio is at most $\frac12$. Thus we may assume that $\mathcal A$
    matches $r_2$. By renaming vertices in $L\setminus\{a_1\}$, assume that
    $r_2$ is matched to $a_2$.

    \item The third arrival is $r_3$, with $N(r_3)=\{a_1,a_2\}$. Both neighbors of
    $r_3$ are already matched. Thus, if $\mathcal A$ matches $r_3$, it must use
    a length-$3$ augmenting path.
\end{enumerate}

We now split into two cases.

\smallskip
\noindent\emph{Case 1: $\mathcal A$ matches $r_3$.}
Let the augmenting path be $\apath{r_3}{v}{x}{u}$, where $v\in\{a_1,a_2\}$,
$x\in\{r_1,r_2\}$, and $u$ is free before the augmentation. Then $x$ is
reassigned once. Since $t=1$, the online vertex $x$ is exhausted.
Figure~\ref{fig:upper35} shows the branch $v=a_2$, $x=r_2$, and $u=a_3$.

We reveal two more vertices.

\begin{itemize}[leftmargin=*, itemsep=2pt]
    \item The fourth arrival is $r_4$, with $N(r_4)=\{u\}$. The algorithm
    cannot match $r_4$. Its only neighbor $u$ is already matched to the
    exhausted vertex $x$. Hence any length-$3$ augmentation starting at $r_4$
    would have to reassign $x$, which is infeasible.

    \item The fifth arrival is $r_5$, with $N(r_5)=\{v\}$. The algorithm cannot
    match $r_5$ either. Its only neighbor $v$ is matched to $r_3$. Also, $r_3$
    has no free neighbor, since $N(r_3)=\{a_1,a_2\}$ and both vertices are matched.
    Hence there is no length-$3$ augmenting path starting at $r_5$.
\end{itemize}

The final graph has a matching of size $5$. Match $r_4$ to $u$ and $r_5$ to
$v$. Match $r_3$ to the vertex of $\{a_1,a_2\}$ different from $v$. Finally, match
$r_1$ and $r_2$ to the two vertices in $L\setminus\{a_1,a_2,u\}$. These edges
exist by the definitions of $N(r_1)$ and $N(r_2)$. Thus $\nu(G)=5$, while
$\mathcal A$ matches only $r_1,r_2,r_3$. Therefore
\[
    \frac{|\mathcal A(G,\pi)|}{\nu(G)}
    \le
    \frac35 .
\]

\smallskip
\noindent\emph{Case 2: $\mathcal A$ leaves $r_3$ unmatched.}
We reveal two more vertices.

\begin{itemize}[leftmargin=*, itemsep=2pt]
    \item The fourth arrival is $r_4$, with $N(r_4)=\{a_1,a_2\}$. If $\mathcal A$
    leaves $r_4$ unmatched, we stop. The graph revealed so far has a matching
    of size $4$, while $\mathcal A$ matches only two vertices. Hence the ratio
    is at most $\frac12$.

    Thus we may assume that $\mathcal A$ matches $r_4$. Since both neighbors of
    $r_4$ are already matched, this match uses a length-$3$ augmenting path.
    Let this path be $\apath{r_4}{v}{x}{u}$, where $v\in\{a_1,a_2\}$, $x\in\{r_1,r_2\}$, and
    $u$ is free before the augmentation. Then $x$ is reassigned once, and hence
    exhausted.

    \item The fifth arrival is $r_5$, with $N(r_5)=\{u\}$. The algorithm cannot
    match $r_5$. Its only neighbor $u$ is already matched to the exhausted
    vertex $x$. Hence any length-$3$ augmentation starting at $r_5$ would have
    to reassign $x$, which is infeasible.
\end{itemize}

The final graph has a matching of size $5$. Match $r_5$ to $u$. Match $r_3$
and $r_4$ to $a_1$ and $a_2$. Finally, match $r_1$ and $r_2$ to the two vertices in
$L\setminus\{a_1,a_2,u\}$. These edges exist by the definitions of $N(r_1)$ and
$N(r_2)$. Thus $\nu(G)=5$, while $\mathcal A$ matches only $r_1,r_2,r_4$.
Therefore
\[
    \frac{|\mathcal A(G,\pi)|}{\nu(G)}
    \le
    \frac35 .
\]

In all cases, the competitive ratio of $\mathcal A$ is at most $\frac35$.
\end{proof}

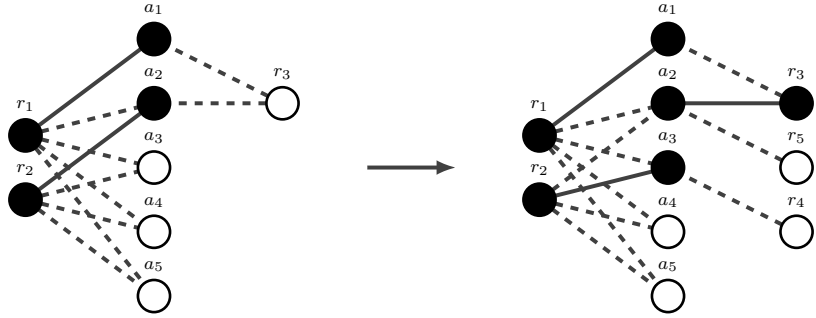
\begin{figure}[ht]
\centering
\begin{tikzpicture}[scale=0.85, every node/.style={transform shape}]
\def\dy{1.0}

\Vertex[x=0,y=0*\dy,    size=.5, color=black, style=solid, label=$a_1$, position=above]{l1a}
\Vertex[x=0,y=-1*\dy,   size=.5, color=black, style=solid, label=$a_2$, position=above]{l2a}
\Vertex[x=0,y=-2*\dy,   size=.5, color=white, style=solid, label=$a_3$, position=above]{l3a}
\Vertex[x=0,y=-3*\dy,   size=.5, color=white, style=solid, label=$a_4$, position=above]{l4a}
\Vertex[x=0,y=-4*\dy,   size=.5, color=white, style=solid, label=$a_5$, position=above]{l5a}

\Vertex[x=-2,y=-1.5*\dy,size=.5, color=black, style=solid, label=$r_1$, position=above]{r1a}
\Vertex[x=-2,y=-2.5*\dy,size=.5, color=black, style=solid, label=$r_2$, position=above]{r2a}
\Vertex[x=2,y=-1*\dy,   size=.5, color=white, style=solid, label=$r_3$, position=above]{r3a}

\Edge[style=solid](r1a)(l1a)
\Edge[style=dashed](r1a)(l2a)
\Edge[style=dashed](r1a)(l3a)
\Edge[style=dashed](r1a)(l4a)
\Edge[style=dashed](r1a)(l5a)
\Edge[style=solid](r2a)(l2a)
\Edge[style=dashed](r2a)(l3a)
\Edge[style=dashed](r2a)(l4a)
\Edge[style=dashed](r2a)(l5a)
\Edge[style=dashed](r3a)(l1a)
\Edge[style=dashed](r3a)(l2a)

\Vertex[x=3,y=-2*\dy,style={color=white}]{flecha1}
\Vertex[x=5,y=-2*\dy,style={color=white}]{flecha2}
\Edge[Direct](flecha1)(flecha2)

\Vertex[x=8,y=0*\dy,    size=.5, color=black, style=solid, label=$a_1$, position=above]{l1b}
\Vertex[x=8,y=-1*\dy,   size=.5, color=black, style=solid, label=$a_2$, position=above]{l2b}
\Vertex[x=8,y=-2*\dy,   size=.5, color=black, style=solid, label=$a_3$, position=above]{l3b}
\Vertex[x=8,y=-3*\dy,   size=.5, color=white, style=solid, label=$a_4$, position=above]{l4b}
\Vertex[x=8,y=-4*\dy,   size=.5, color=white, style=solid, label=$a_5$, position=above]{l5b}

\Vertex[x=6,y=-1.5*\dy,size=.5, color=black, style=solid, label=$r_1$, position=above]{r1b}
\Vertex[x=6,y=-2.5*\dy,size=.5, color=black, style=solid, label=$r_2$, position=above]{r2b}
\Vertex[x=10,y=-1*\dy,  size=.5, color=black, style=solid, label=$r_3$, position=above]{r3b}
\Vertex[x=10,y=-2*\dy,  size=.5, color=white, style=solid, label=$r_5$, position=above]{r4b}
\Vertex[x=10,y=-3*\dy,  size=.5, color=white, style=solid, label=$r_4$, position=above]{r5b}

\Edge[style=solid](r1b)(l1b)
\Edge[style=dashed](r1b)(l2b)
\Edge[style=dashed](r1b)(l3b)
\Edge[style=dashed](r1b)(l4b)
\Edge[style=dashed](r1b)(l5b)
\Edge[style=dashed](r2b)(l2b)
\Edge[style=solid](r2b)(l3b)
\Edge[style=dashed](r2b)(l4b)
\Edge[style=dashed](r2b)(l5b)
\Edge[style=dashed](r3b)(l1b)
\Edge[style=solid](r3b)(l2b)
\Edge[style=dashed](r4b)(l2b)
\Edge[style=dashed](r5b)(l3b)

\end{tikzpicture}
\caption{One branch in the proof of Theorem~\ref{thm:upperboundgeneral}(b).
Here $r_3$ is matched using the path $\apath{r_3}{a_2}{r_2}{a_3}$. The adversary then reveals
$r_4$ with neighborhood $\{a_3\}$ and $r_5$ with neighborhood $\{a_2\}$.}
\label{fig:upper35}
\end{figure}

\subsection{An adaptive instance for finite online budgets}
\label{subsec:general-upper-bound}

The upper bound for finite online budget $t$ is more complex than the endpoint bounds for $t=1$ and $t=\infty$. In those endpoint regimes, a small fixed instance already suffices. For intermediate $t$, however, a deterministic algorithm may deliberately leave an online vertex unmatched even when a feasible augmentation exists, in order to conserve online reassignment budget for future arrivals.

We therefore describe an \emph{adaptive instance} (i.e., neighborhoods are revealed online depending on the algorithm's actions) that forces the matching to evolve in phases. In each phase, a batch of online vertices arrives; the algorithm must either match them using length-$3$ augmentations or leave them unmatched. These augmentations drive a pool of rooted components to grow into deeper trees, all of type $T^Q$ as in Lemma~\ref{lem:history-components}. Matched phase vertices also create entries in $\mathcal B$. At the end, we add ``blocker'' vertices: blockers for $\mathcal B$ are blocked because their matched phase vertices have no free offline neighbor, while blockers for the remaining pool components are blocked by exhausted online roots. Importantly, the argument does not rely on exhausting the offline budget.

\begin{theorem}\label{thm:upper-bound-general-t}
For every finite $t\ge1$ and every $s\ge1$, no deterministic online algorithm for $\omp_3(s,t)$ has competitive ratio larger than
\[
    \gamma_t:=\frac{2\cdot 2^t-1}{3\cdot 2^t-1}.
\]
\end{theorem}

\begin{proof}
The case $s=0$ is not part of this construction: then no length-$3$ augmentation is feasible, and the model reduces to classical online matching. Hence we assume $s\ge1$.

Fix $s\ge 1$, a deterministic algorithm $\mathcal A$ for $\omp_3(s,t)$, and an integer $N$. We construct an adaptive instance.
For fixed $N$ and $t$, we fix in advance a sufficiently large offline set $L$ containing all vertices that may be used in the construction. The adaptive choices only determine which incident edges are revealed to arriving online vertices.
All offline vertices used below are distinct unless explicitly identified.

\paragraph{Initial phase}
We create a set $X$ of online vertices one by one. Each vertex $r\in X$ is
revealed with a private set $U_r\subseteq L$ of $t+2$ fresh offline neighbors,
disjoint from all other such sets. Thus $N(r)=U_r$ and $|U_r|=t+2$.

We reveal vertices of $X$ one by one, until one of the following two events occurs: 
\[ \mathcal E_1: \mathcal A \text{ has matched } N \text{ vertices of } X, \quad \text{ or } \quad \mathcal E_2: \lceil N/\gamma_t\rceil \text{ vertices of } X \text{ have been revealed.} \]

If $\mathcal E_2$ occurs strictly before $\mathcal E_1$, we stop the instance. In this
case, $\mathcal A$ has matched at most $N-1$ vertices of $X$, because
$\mathcal E_1$ has not occurred.
On the other hand, the optimum can match all $\lceil N/\gamma_t\rceil$ revealed vertices of $X$, since each has a private set of offline neighbors.
Hence the ratio of $\mathcal A$ is at most
\[
    \frac{N-1}{\lceil N/\gamma_t\rceil}
    <
    \gamma_t .
\]
Therefore this branch already gives the desired upper bound. From now on, we assume that $\mathcal E_1$ occurs before or at the same time
as $\mathcal E_2$.

\paragraph{Preparation and definitions.}
Let $R_0\subseteq X$ be the set of the $N$ online vertices matched by
$\mathcal A$. Together with their matches, the vertices in $R_0$ define
$N$ components of type $T^0$.

We put these components in the initial pool $\mathcal P_0$. We do not use the
vertices in $X\setminus R_0$ again. We also initialize two empty sets
$\mathcal R$ and $\mathcal B$.

For every component $C$ that appears in one of the pools, we write
$\rootof(C)$ for its root. Thus $\rootof(C)\in R_0$. We call the current
matching edge of $\mathcal A$ incident to $\rootof(C)$ the active edge of $C$.
We write $\act(C)$ for its offline endpoint, so the active edge of $C$ is
\[
    \act(C)\rootof(C).
\]
Initially, if $C\in\mathcal P_0$, then $\act(C)$ is the offline vertex matched
to $\rootof(C)$ by $\mathcal A$.

The pools $\mathcal P_q$ used below are dynamic. Intuitively, $\mathcal P_q$ is the pool of ``live'' components that can still be advanced to the next phase, $\mathcal R$ collects components that have been removed from the pools as witnesses (and never return), and $\mathcal B$ collects old active offline endpoints that will later define blockers.
During the phases, a component can leave the
current pool in only two ways: it is either inserted into the next pool, or it
is inserted into $\mathcal R$. The set $\mathcal B$ collects certain offline
vertices that will be used later in the construction.

Throughout the construction we maintain the following invariants.

(i) Every component in the current pool $\mathcal P_q$ has type $T^q$.

(ii) For every $C\in\mathcal P_q$, the edge
$\act(C)\rootof(C)$ is the current matching edge incident to the root.

(iii) For every $C\in\mathcal P_q$, the vertex $\act(C)$ is the unique leaf
in $L$ adjacent to $\rootof(C)$.

(iv) For every $C\in\mathcal P_q$, the vertex $\act(C)$ has not been
reassigned. Hence the edge $\act(C)\rootof(C)$ is feasible as the middle
edge of a length-$3$ augmentation with respect to the offline budget,
for every $s\ge 1$.

(v) Components in the current pool are pairwise vertex-disjoint.

(vi) Once a component is inserted into $\mathcal R$, it is never inserted
into any later pool and is never used again as a witness component.

(vii) Before phase $q$, for every component $C\in\mathcal P_{q-1}$, if
$r=\rootof(C)$, then at most $q$ vertices of the private set $U_r$ have
been matched by $\mathcal A$.

\paragraph{The phases.}
Set $a_0:=N$. We run $t$ phases. Phase $q$, for $q\in\{1,\ldots,t\}$, is
defined as follows.

At the start of the phase, we take the pool $\mathcal P_{q-1}$ produced at the
end of the previous phase; for $q=1$, this is the initial pool
$\mathcal P_0$. By induction, this pool has size $a_{q-1}$, and every
component in it has type $T^{q-1}$. We also initialize an empty next pool
$\mathcal P_q$.

The phase proceeds while $|\mathcal P_{q-1}|\ge 2$. In each iteration, we
reveal a new online vertex $j$ adjacent exactly to the vertices $\act(C)$ with
$C\in\mathcal P_{q-1}$. We call the vertices revealed in this loop the
phase-$q$ vertices.

The vertex $j$ has no free neighbor, since every vertex $\act(C)$ is matched
by $\mathcal A$. Thus, if $\mathcal A$ matches $j$, it cannot use an augmenting
path of length $1$; it must use an augmenting path of length~$3$.

For each $C\in\mathcal P_{q-1}$, there is a feasible augmenting path of
length $3$ whose first edge is $j\act(C)$ and whose middle edge is the current
matching edge $\act(C)\rootof(C)$. Indeed, let $u$ be a vertex of the private
set $U_{\rootof(C)}$ that has never been matched by $\mathcal A$. Such a vertex exists because, before phase $q$, the root $\rootof(C)$ has
used its initial private endpoint and at most one new private endpoint in
each of the previous $q-1$ phases. Thus it has been matched to at most
$q$ vertices of $U_{\rootof(C)}$, while $|U_{\rootof(C)}|=t+2$. Then
\[
    \apath{j}{\act(C)}{\rootof(C)}{u}
\]
is such a path: the edge $j\act(C)$ is present by the definition of $j$, the
edge $\act(C)\rootof(C)$ is the current matching edge of $C$, and
$\rootof(C)u$ is present because $u\in U_{\rootof(C)}$. Moreover, $\act(C)$ has never been reassigned, so the path is feasible for every offline budget $s\ge 1$, and $\rootof(C)$ has been reassigned only $q-1<t$ times.

There are two cases.

First, suppose that $\mathcal A$ matches $j$. Since the only neighbors of $j$
are the vertices $\act(C)$ with $C\in\mathcal P_{q-1}$, the augmenting path
used by $\mathcal A$ starts with $j\act(C)$ for some
$C\in\mathcal P_{q-1}$. Since $\act(C)$ is currently matched to $\rootof(C)$,
the path then uses the matching edge $\act(C)\rootof(C)$. Before updating the
active endpoint of $C$, we insert $\act(C)$ into the set $\mathcal B$.

By the tree-transition lemma (Lemma~\ref{lem:tree-transition}), the component $C$
changes from type $T^{q-1}$ to type $T^q$. We remove $C$ from
$\mathcal P_{q-1}$, update $\act(C)$ to the new offline endpoint adjacent to
$\rootof(C)$, and insert $C$ into $\mathcal P_q$. The old active endpoint has
just been reassigned and is recorded in $\mathcal B$. The new active endpoint
is the fresh offline endpoint used at the end of the augmentation, so it has
not been reassigned. Hence the invariants are preserved.

Since the loop runs only while $|\mathcal P_{q-1}|\ge 2$, after removing $C$
there is still another component in $\mathcal P_{q-1}$. Choose one such
component $S$ and set $\wit(j)=S$.

For the second case, suppose that $\mathcal A$ does not match $j$. Then choose
any component $S\in\mathcal P_{q-1}$ and set $\wit(j)=S$. In both cases, we remove $\wit(j)$ from the current pool $\mathcal P_{q-1}$ and insert it into $\mathcal R$. This component will never return to a pool and will never be used again as a witness. We refer to $\wit(j)$ as the \emph{witness component} for $j$. We also set
\[
    \mate(j):=\act(\wit(j)).
\]
In the offline matching constructed at the end of this proof, the vertex $j$
will be matched to $\mate(j)$.

This update preserves the invariants. If a component is advanced, it moves from
$\mathcal P_{q-1}$ to $\mathcal P_q$ by Lemma~\ref{lem:tree-transition}. If a
component is used as a witness, it is removed from the pools and is never used
again. No later phase vertex is adjacent to an endpoint removed from a pool.

When the loop stops, phase $q$ ends. Let $p_q$ be the number of phase-$q$
vertices, and let
$
    e_q:=|\mathcal P_{q-1}|
$
be the number of components left in the current pool \emph{at the end of the
phase}. Since the loop stops when fewer than two components remain,
$e_q\in\{0,1\}$. Finally, set
$
    a_q:=|\mathcal P_q|.
$
Equivalently, $a_q$ is the number of phase-$q$ vertices matched by
$\mathcal A$, since $\mathcal P_q$ consists exactly of the components advanced
by $\mathcal A$ during phase $q$.

We have
\[
    p_q+a_q=a_{q-1}-e_q .
\]
Indeed, the quantity $a_{q-1}-e_q$ is exactly the number of components that
leave $\mathcal P_{q-1}$ during phase $q$. Each such component goes to exactly
one of two places. It either enters $\mathcal R$, in which case it is
$\wit(j)$ for a unique phase-$q$ vertex $j$, or it enters $\mathcal P_q$, in
which case it is the component advanced when $\mathcal A$ matches a phase-$q$
vertex. The first group has size $p_q$, and the second group has size $a_q$.

Since every phase-$q$ vertex matched by $\mathcal A$ is one of the $p_q$
phase-$q$ vertices, we have $a_q\le p_q$. Therefore
\[
    a_q\le \frac{a_{q-1}}2 .
\]
As $a_0=N$, it follows that
\[
    a_q\le \frac{N}{2^q}
    \qquad\text{for }q\in \{0,\ldots,t\}.
\]

\paragraph{Final blockers.}
After the $t$ phases, we use the sets $\mathcal B$ and $\mathcal P_t$ to define
additional online vertices, called blockers, that $\mathcal A$ cannot match.

For every offline vertex $x\in\mathcal B$, we reveal one online vertex $b_x$
adjacent only to $x$. Let $j$ be the phase vertex that became matched to $x$
when $x$ was inserted into $\mathcal B$. At the arrival of $j$, all neighbors
of $j$ were active endpoints of components in the current pool. Thus all of
them were matched. The full neighborhood of $j$ was revealed at that time,
and offline vertices never become free after being matched. Hence $j$ has no
free offline neighbor at the end.

The vertex $b_x$ has no feasible direct match, since $x$ is matched to $j$.
Any length-$3$ augmenting path starting at $b_x$ would have to be of the form
$\apath{b_x}{x}{j}{i}$, with $i$ a free offline neighbor of $j$. No such $i$
exists. Hence $\mathcal A$ cannot match $b_x$.

Second, for every component $C\in\mathcal P_t$, we reveal one online vertex
adjacent only to $\act(C)$. This vertex cannot be matched by an augmenting path
of length $1$, since $\act(C)$ is matched. Any augmenting path of length $3$
would have to use the current matching edge $\act(C)\rootof(C)$, but
$\rootof(C)$ has already been reassigned $t$ times. Thus, the online budget of
$\rootof(C)$ blocks the path, and $\mathcal A$ cannot match this vertex either.

There are
\[
    |\mathcal B|=\sum_{q=1}^t a_q
\]
blockers of the first kind. Indeed, one offline vertex is inserted into
$\mathcal B$ for each phase vertex matched by $\mathcal A$, and these inserted
vertices are distinct because an old active endpoint is never active again. There are also
$a_t$ blockers of the second kind, one for each component in $\mathcal P_t$.

\paragraph{Counting.}
We now compare the size of $\mathcal A$'s matching with the size of an offline matching.
The algorithm $\mathcal A$ matches the $N$ initial vertices in $R_0$ and exactly $a_q$ phase-$q$ vertices for each $q\in\{1,\ldots,t\}$. It matches no final blocker. Hence
\[
    \alg = N+\sum_{q=1}^t a_q .
\]

We now construct a feasible offline matching.

First, match each of the $N$ vertices $r\in R_0$ to a vertex of its private set $U_r$ never matched by $\mathcal A$. Such a vertex exists because the component rooted at $r$ uses its initial
private endpoint and at most one new private endpoint in each of the $t$
phases, hence at most $t+1$ vertices of $U_r$ are matched by $\mathcal A$.
Since $|U_r|=t+2$, at least one vertex of $U_r$ has never been matched by
$\mathcal A$.

Second, for every vertex $j$ revealed in phase $q$, for some $q\in\{1,\ldots,t\}$, match $j$ to $\mate(j)$. This is feasible because, when $j$ was revealed, it was adjacent to $\act(\wit(j))$, and $\mate(j)$ was defined to be this vertex. Moreover, each witness component is inserted into $\mathcal R$ at most once, so these offline vertices are all distinct. The total number of these vertices is $\sum_{q=1}^t p_q$.

Finally, match every final blocker to its unique offline neighbor.

\begin{claim}
The offline vertices used in the offline matching described above are
pairwise distinct.
\end{claim}

\begin{proof}
For each $r\in R_0$, we choose a vertex of $U_r$ that has never been matched
by $\mathcal A$. Every active endpoint used in the construction has been matched by
$\mathcal A$ at some time. Hence the chosen vertices from the sets $U_r$ are not
active endpoints ever used by the construction. Also, the sets $U_r$ are
pairwise disjoint.

Each phase vertex $j$ is matched to
$\mate(j)=\act(\wit(j))$. Once $\wit(j)$ is inserted into $\mathcal R$, it
never returns to a pool and is never used again as a witness. Therefore no
two phase vertices use the same offline vertex.

Blockers of the first kind use vertices inserted into $\mathcal B$, i.e.,
old active endpoints of components that were advanced by $\mathcal A$.
Such vertices are never current active endpoints again. Blockers of the
second kind use the final active endpoints of components in $\mathcal P_t$.
Finally, a component inserted into $\mathcal R$ never advances later, so the
vertex used as $\mate(j)$ for a phase vertex is never later inserted into
$\mathcal B$ and is not a final active endpoint in $\mathcal P_t$.

Thus the offline vertices used for $R_0$, for phase vertices, and for final
blockers are pairwise distinct.
\end{proof}

Using $p_q+a_q=a_{q-1}-e_q$ and $a_0=N$, we have
$\sum_{q=1}^t p_q=N-a_t-\sum_{q=1}^t e_q$. Therefore,
\begin{align*}
    \opt
    &\ge N+\sum_{q=1}^t p_q+\sum_{q=1}^t a_q+a_t  \\
    &= 2N+\sum_{q=1}^t a_q-\sum_{q=1}^t e_q
     \ge 2N+\sum_{q=1}^t a_q-t ,
\end{align*}
where the last inequality uses $e_q\in\{0,1\}$ for every $q$.

Let $\alpha:=\sum_{q=1}^t a_q$. Then
$\alg/\opt\le (N+\alpha)/(2N+\alpha-t)$. For fixed $N>t$, the right-hand
side is increasing in $\alpha$. Since
$\alpha\le N\sum_{q=1}^t 2^{-q}=N(1-2^{-t})$, we get
\begin{align*}
    \frac{\alg}{\opt}
    &\le
    \frac{N+N(1-2^{-t})}{2N+N(1-2^{-t})-t}  =
    \frac{N(2\cdot 2^t-1)}
    {N(3\cdot 2^t-1)-t2^t}.
\end{align*}
Letting $N\to\infty$ gives
\[
    \limsup_{N\to\infty}\frac{\alg}{\opt}
    \le \frac{2\cdot 2^t-1}{3\cdot 2^t-1}
    = \gamma_t .
\]

Thus, for every deterministic algorithm $\mathcal A$ and every $\varepsilon>0$, there is a finite instance on which $\mathcal A$ obtains ratio at most $\gamma_t+\varepsilon$. Since $s\ge 1$ was arbitrary, no deterministic algorithm for $\omp_3(s,t)$, for any $s\ge 1$, has competitive ratio larger than $\gamma_t$.\qedhere
\end{proof}

\paragraph{Consequences.}
Together with Theorem~\ref{thm:lcp-ratio}, Theorem~\ref{thm:upper-bound-general-t}, and Lemma~\ref{lemma:lcp-spgreedy}, this shows that $\gamma_t$ is the optimal deterministic competitive ratio for every finite $t\ge1$ and every $s\ge1$. If $s=0$ or $t=0$, the model reduces to classical deterministic online matching, and the optimal ratio is $1/2$. If $s\ge1$ and $t=\infty$, Theorem~\ref{thm:lcp-infinity} and Lemma~\ref{lemma:lcp-spgreedy} give ratio $2/3$, while Theorem~\ref{thm:upperboundgeneral}(a) gives the matching bound.

\section{Offline-vertex-weighted Setting}
\label{sec:offline-weighted}

We consider the variant in which each offline vertex $i\in L$ has weight
$w(i)\ge 0$, and the objective is to maximize the total weight of the matched
offline vertices. The arrival model and the allowed operations are the same as
before. We refer to this offline-vertex-weighted model as \emph{WOMP};
correspondingly, $\womp_K(s,t)$ denotes the weighted analogue of $\omp_K(s,t)$
with the same augmentation and reassignment-budget constraints.
Note that the weighted results do not follow as a formal corollary
of the unweighted analysis; in particular, the unweighted dual variables cannot
encode arbitrary offline weights.

We study two offline-vertex-weighted regimes. First, in $\womp_3(1,1)$, both
sides have unit reassignment budget; the optimal deterministic ratio is
$2-\sqrt2$. Second, in $\womp_3(1,\infty)$, the offline side still has unit budget
but the online side has unlimited budget; the optimal deterministic ratio is
$(\sqrt5-1)/2$. The two analyses use different greedy rules and different
primal--dual constructions.

\begin{theorem}\label{thm:weighted-optimal-ratio}
The optimal deterministic competitive ratio for offline-vertex-weighted
$\womp_3(1,1)$ is exactly
\[
    2-\sqrt2 \approx 0.585786.
\]
More precisely, there is a deterministic online algorithm with competitive
ratio $2-\sqrt2$, and no deterministic online algorithm can achieve a larger
competitive ratio.
\end{theorem}
Compare with the unweighted unit-budget case, where the optimal deterministic ratio is $3/5$.
Since $2-\sqrt2 < \frac35$, offline weights make $\womp_3(1,1)$ strictly harder.

We next consider $\womp_3(1,\infty)$, in which the offline side still has unit
reassignment budget but the online side has unlimited budget.

\begin{theorem}\label{thm:weighted-unbounded-optimal-ratio}
The optimal deterministic competitive ratio for offline-vertex-weighted
$\womp_3(1,\infty)$ is exactly the reciprocal of the golden ratio
\[
    (\sqrt5-1)/2 \approx 0.61803.
\]
More precisely, there is a deterministic online algorithm with competitive
ratio $(\sqrt5-1)/2$, and no deterministic online algorithm can achieve a larger
competitive ratio.
\end{theorem}
Compare with the unweighted unit-budget case, where the optimal deterministic ratio is $2/3$.
Since $(\sqrt5-1)/2 < \frac23$, offline weights make $\womp_3(1,\infty)$ strictly harder.

Both ratios are achieved by greedy-type algorithms and matched by adaptive
upper bounds against deterministic algorithms. Finally, we observe that if
weights are placed on online vertices instead, no positive deterministic
competitive ratio is possible.

We next prove these two exact ratios. Subsection~\ref{subsec:weighted-11}
presents the $\womp_3(1,1)$ algorithm and its analysis, and
Subsection~\ref{subsec:weighted-11-upper} gives the corresponding upper bound.
Subsection~\ref{subsec:weighted-1inf} presents the $\womp_3(1,\infty)$ algorithm
and its analysis, and Subsection~\ref{subsec:weighted-1inf-upper} gives the
corresponding upper bound. Finally, Subsection~\ref{subsec:weighted-online-side}
explains why restricting weights to the offline side is necessary.

We begin by stating the LP relaxation for maximum-weight bipartite matching
(with weights on $L$) and a weighted primal--dual template, which will be used
in both algorithmic analyses. The primal LP and its dual are:
\[
\begin{array}{@{}c@{\qquad}c@{}}
\begin{aligned}
(P_{G,w})\qquad
\max \ & \sum_{ij\in E} w(i)\,x_{ij}\\
\text{s.t.}\ &
\sum_{j\in N(i)} x_{ij} \le 1 \qquad \forall i\in L,\\
&
\sum_{i\in N(j)} x_{ij} \le 1 \qquad \forall j\in R,\\
& x_{ij}\ge 0 \qquad \forall ij\in E
\end{aligned}
&
\begin{aligned}
(D_{G,w})\qquad
\min \ & \sum_{i\in L}\alpha_i+\sum_{j\in R}\beta_j\\
\text{s.t.}\ &
\alpha_i+\beta_j \ge w(i) \qquad \forall ij\in E,\\
& \alpha_i\ge 0 \qquad \forall i\in L,\\
& \beta_j\ge 0 \qquad \forall j\in R
\end{aligned}
\end{array}
\]

\begin{lemma}[Weighted primal--dual template]\label{lemma:weighted-pd-template}
Let $\mathcal A$ be an online algorithm. Suppose that during the execution of
$\mathcal A$, we maintain dual variables $\alpha_i,\beta_j\ge 0$, initially all
zero, such that whenever $\mathcal A$ newly matches an offline vertex $i\in L$,
the dual objective
\[
    \sum_{x\in L}\alpha_x+\sum_{y\in R}\beta_y
\]
increases by at most $w(i)/\rho$. Suppose further that, after $\mathcal A$
terminates, an optional postprocessing step is applied that does not increase
the dual objective, and the final variables are feasible for $(D_{G,w})$. Then
$\mathcal A$ is $\rho$-competitive.
\end{lemma}

\begin{proof}
Let $M$ be the final matching produced by $\mathcal A$, and let
$W_{\mathcal A}=\sum_{ij\in M} w(i)$ be its total weight. In this model, once an
offline vertex is matched, it never becomes free again. Thus, the set of offline
vertices newly matched during the execution is exactly the set of vertices matched
in $M$.

By hypothesis, each newly matched offline vertex $i$ increases the
dual objective by at most $w(i)/\rho$. The total increase during the online phase is
therefore bounded by $\sum_{ij\in M} w(i)/\rho = W_{\mathcal A}/\rho$. The postprocessing
step does not increase the dual objective. Since the final variables are feasible
for $(D_{G,w})$, weak duality gives
\[
    \opt(P_{G,w})
    \le
    \sum_{i\in L}\alpha_i+\sum_{j\in R}\beta_j
    \le
    \frac{W_{\mathcal A}}{\rho}.
\]
We conclude that $W_{\mathcal A}\ge \rho\cdot \opt(P_{G,w})$, so $\mathcal A$ is 
$\rho$-competitive.
\end{proof}

\subsection{\texorpdfstring{$\womp_3(1,1)$}{WOMP3(1,1)}: Algorithm and Analysis}
\label{subsec:weighted-11}

Our threshold greedy algorithm $\greedy_{q,\delta}$
(Algorithm~\ref{alg:greedy-q-delta}) works as follows. When an online vertex $j\in R$ arrives, let
$D(j)$ be a heaviest free neighbor of $j$; if none exists, set $D(j)=\bot$ and use the convention $w(D(j))=0$.
A feasible augmenting $3$-path $\apath{j}{x}{y}{i}$ is \emph{eligible} if the matched edge $xy$ 
is not frozen, $i$ is free, and the weights satisfy
\[
    w(x)\ge q\,w(D(j))
    \qquad\text{and}\qquad
    w(i)\ge \delta\,w(x).
\]
Combinatorially, this ensures that an augmentation is performed only if displacing the current match $x$ yields a sufficient multiplicative gain over the best direct match $D(j)$, and the newly matched offline vertex $i$ is not too light compared to $x$.

If any eligible paths exist, the algorithm augments along one that maximizes the 
pair $(w(x),w(i))$ lexicographically (i.e., it maximizes $w(x)$ first, and then 
maximizes $w(i)$ among those achieving the maximum $w(x)$). Otherwise, it matches 
$j$ directly to $D(j)$ if possible.

\begin{algorithm}[ht]
  \SetAlgoLined
  \DontPrintSemicolon
  \KwIn{$L$ offline side, $w$ weight function on $L$, parameters $q\ge 1$ and $\delta\in[0,1]$}
  \KwOut{A matching $M$}

  Initialize $M\gets \emptyset$\;

  \ForAll{arrivals $j\in R$}{
    $D(j)\gets$ a free neighbor $i\in N(j)$ of maximum weight $w(i)$
    (if none exists, set $D(j)=\bot$ and $w(D(j))=0$)\;

        Let $\mathcal P(j)$ be the set of eligible augmenting $3$-paths $\apath{j}{x}{y}{i}$,
    where eligibility means that $xy\in M$, the matched edge $xy$ is not frozen,
    $i$ is free, $w(x)\ge q\cdot w(D(j))$, and $w(i)\ge \delta\cdot w(x)$\;

    \uIf{$\mathcal P(j)\neq\emptyset$}{
      Choose a path $P(j)=\apath{j}{x}{y}{i}\in\mathcal P(j)$ that maximizes
      $(w(x),w(i))$ lexicographically\;

      Augment along $P(j)$. That is, set
      $M\gets (M\setminus \{xy\}) \cup \{jx,yi\}$.
    }
    \lElseIf{$D(j)\neq\bot$}{ Augment along $\apath{j}{D(j)}$.
    }
  }
  \caption{$\greedy_{q,\delta}$ with parameters $q\ge 1$ and $\delta\in[0,1]$.}
  \label{alg:greedy-q-delta}
\end{algorithm}

When the algorithm augments along a path $\apath{j}{x}{y}{i}$, we define
$S(j)=x$ and $T(j)=i$. By definition, the pair $(w(S(j)), w(T(j)))$ satisfies
$w(S(j))\ge q\,w(D(j))$ and $w(T(j))\ge \delta\,w(S(j))$, and is lexicographically maximum
among all eligible paths (ordered by $(w(x),w(i))$). The proof uses this lexicographic rule only when comparing the chosen path with another eligible path available at the same arrival. In particular, if the two paths have the same middle offline vertex, then the chosen path is guaranteed to have a free endpoint weight at least as large as the alternative one.

\begin{theorem}\label{thm:weighted-lower-bound}
    Choosing $q=1+\sqrt2$ and $\delta=1/\sqrt2$, the
    $\greedy_{q,\delta}$ algorithm is $(2-\sqrt2)$-competitive for
    offline-vertex-weighted $\womp_3(1,1)$.
\end{theorem}

These values are chosen to balance the inequalities used in the primal--dual
analysis. This balance is used in Lemma~\ref{lem:weighted-dual-objective}.

Fix an instance $(G,\pi,w)$, where $G=(L,R,E)$, $\pi$ is the arrival order on $R$, and
$w\colon L\to\mathbb R_{\ge 0}$ assigns weights to offline vertices. Let $\mathcal A$
be the $\greedy_{q,\delta}$ algorithm for some parameter $q\ge 1$, to be determined later.
The dual updates and the postprocessing step are illustrated in
Figure~\ref{fig:weight11-updates}
(\subref{fig:weight11-direct}--\subref{fig:weight11-post}).

For the rest of the section, set
\[
    q=1+\sqrt2,
    \qquad
    \delta=\frac1{\sqrt2},
    \qquad
    \rho=\frac1{1+\delta}=2-\sqrt2.
\]

We will use the identities and inequalities
\[
    \delta-\frac1q=1-\delta,
    \qquad
    \frac{\delta}{q}=1-\delta,
    \qquad
    \frac1q\ge 1-\delta,
    \qquad
    \delta\ge 1-\delta.
\]
We next describe the online dual updates as a function of the moves performed
by $\mathcal A=\greedy_{q,\delta}$.

Unlike the unweighted dual constructions above, the weighted dual updates are not monotone in individual variables. A $3$-path update may decrease the value of the online dual variable associated with the middle matched vertex. This is harmless: the weighted primal--dual template only uses the net change in the dual objective and final dual feasibility.

\begin{algorithm}[ht]
  \SetAlgoLined
  \DontPrintSemicolon
  \KwIn{Instance $(G,\pi,w)$ and the execution of $\mathcal A$}
  \KwOut{Dual variables $(\alpha_i)_{i\in L}$ and $(\beta_j)_{j\in R}$}

  \lForAll{$i\in L,\, j\in R$}{Set $\alpha_i\gets 0$, $\beta_j\gets 0$}

  \ForAll(\tcp*[f]{dual update steps}){$j\in R$ in arrival order, while running $\mathcal A$}{
    \lIf{$\mathcal A$ matches $j$ directly to $i\in L$}{
      Set $\alpha_i \gets \delta\, w(i)$ and $\beta_j \gets w(i)$
    }
    \uElseIf{$\mathcal A$ matches $j$ via an augmenting $3$-path $\apath{j}{x}{y}{i}$}{
      Set $\alpha_x \gets w(x)$, $\alpha_i \gets w(i)$,
      $\beta_y \gets w(i)$, and $\beta_j \gets w(x)/q$
    }
  }

  Let $M$ be the final matching constructed by $\mathcal A$\;
  \ForAll(\tcp*[f]{postprocessing step}){$ij\in M$ with $i\in L$ and $j\in R$}{
    \If{neither $i$ nor $j$ is exhausted, and every free neighbor $h$ of $j$ at the end satisfies $w(h)\le \delta\,w(i)$}{
      Set $\alpha_i \gets w(i)$ and $\beta_j \gets \delta\, w(i)$
    }
  }
  \caption{Dual for offline-weighted $\womp_3(1,1)$.}
  \label{alg:dual-weight}
\end{algorithm}

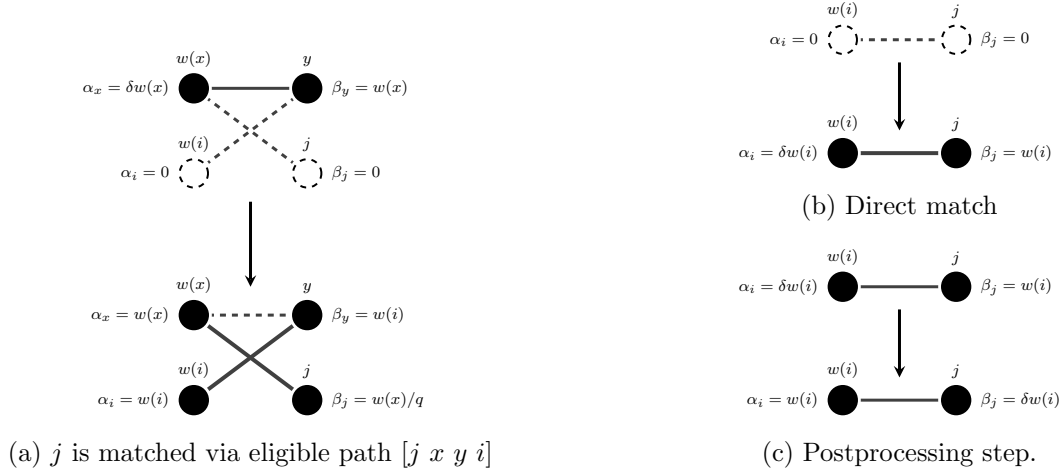
\begin{figure}[ht]
\centering

\begin{minipage}[b]{0.48\textwidth}
    \centering
    \begin{subfigure}[b]{\textwidth}
        \centering
        \scalebox{0.75}{
        \begin{tikzpicture}
            \useasboundingbox (-2, -0.4) rectangle (4, 5.9);
            
            \Vertex[x=0,y=4.0, style={color=white}, label={$\alpha_i=0$}, position=left]{yi_top} 
            \Vertex[x=2,y=5.5, style={color=white}, label={$\beta_y=w(x)$}, position=right]{yy_top}
            \Vertex[x=0,y=5.5, style={color=white}, label={$\alpha_x=\delta w(x)$}, position=left]{yx_top} 
            \Vertex[x=2,y=4.0, style={color=white}, label={$\beta_j=0$}, position=right]{yj_top}

            \Vertex[x=0,y=4.0, size=.5, color=white, style=dashed, label=$w(i)$, position=above]{i_top} 
            \Vertex[x=2,y=5.5, size=.5, color=black, label=$y$, position=above]{y_top}
            \Vertex[x=0,y=5.5, size=.5, color=black, label=$w(x)$, position=above]{x_top} 
            \Vertex[x=2,y=4.0, size=.5, color=white, style=dashed, label=$j$, position=above]{j_top}
            
            \Edge[style=dashed](i_top)(y_top)
            \Edge[style=dashed](x_top)(j_top)
            \Edge[style=solid](x_top)(y_top)

            \draw[->, >=stealth, line width=1.5pt] (1, 3.5) -- (1, 2.0);

            \Vertex[x=0,y=0, style={color=white}, label={$\alpha_i=w(i)$}, position=left]{yi_bot} 
            \Vertex[x=2,y=1.5, style={color=white}, label={$\beta_y=w(i)$}, position=right]{yy_bot}
            \Vertex[x=0,y=1.5, style={color=white}, label={$\alpha_x=w(x)$}, position=left]{yk_bot} 
            \Vertex[x=2,y=0, style={color=white}, label={$\beta_j=w(x)/q$}, position=right]{yj_bot}

            \Vertex[x=0,y=0, size=.5, color=black, label=$w(i)$, position=above]{i_bot} 
            \Vertex[x=2,y=1.5, size=.5, color=black, label=$y$, position=above]{y_bot}
            \Vertex[x=0,y=1.5, size=.5, color=black, label=$w(x)$, position=above]{x_bot} 
            \Vertex[x=2,y=0, size=.5, color=black, label=$j$, position=above]{j_bot}
            
            \Edge[style={line width=2pt}](i_bot)(y_bot)
            \Edge[style={line width=2pt}](x_bot)(j_bot)
            \Edge[style=dashed](x_bot)(y_bot)
        \end{tikzpicture}
        }
        \caption{$j$ is matched via eligible path $\apath{j}{x}{y}{i}$}
        \label{fig:weight11-aug3}
    \end{subfigure}
\end{minipage}%
\hfill
\begin{minipage}[b]{0.48\textwidth}
    \centering
    \begin{subfigure}[b]{\textwidth}
        \centering
        \scalebox{0.75}{
        \begin{tikzpicture}
            \useasboundingbox (-2, -0.4) rectangle (4, 2.4);
            
            \Vertex[x=0, y=2.0, style={color=white}, label={$\alpha_i=0$}, position=left]{yi_top} 
            \Vertex[x=2, y=2.0, style={color=white}, label={$\beta_j=0$}, position=right]{yj_top}
            \Vertex[x=0, y=2.0, size=.5, color=white, style=dashed, label=$w(i)$, position=above]{i_top} 
            \Vertex[x=2, y=2.0, size=.5, color=white, style=dashed, label=$j$, position=above]{j_top}
            \Edge[style=dashed](i_top)(j_top)

            \draw[->, >=stealth, line width=1.5pt] (1, 1.6) -- (1, 0.4);

            \Vertex[x=0, y=0, style={color=white}, label={$\alpha_i=\delta w(i)$}, position=left]{yi_bot} 
            \Vertex[x=2, y=0, style={color=white}, label={$\beta_j=w(i)$}, position=right]{yj_bot}
            \Vertex[x=0, y=0, size=.5, color=black, label=$w(i)$, position=above]{i_bot} 
            \Vertex[x=2, y=0, size=.5, color=black, label=$j$, position=above]{j_bot}            
            \Edge[style={line width=2pt}](i_bot)(j_bot)
        \end{tikzpicture}
        }
        \caption{Direct match}
        \label{fig:weight11-direct}
    \end{subfigure}

    \vspace{0.5cm}

    \begin{subfigure}[b]{\textwidth}
        \centering
        \scalebox{0.75}{
        \begin{tikzpicture}
            \useasboundingbox (-2, -0.4) rectangle (4, 2.4);
            
            \Vertex[x=0,y=2.0, style={color=white}, label={$\alpha_i=\delta w(i)$}, position=left]{yi_top} 
            \Vertex[x=2,y=2.0, style={color=white}, label={$\beta_j=w(i)$}, position=right]{yj_top}
            \Vertex[x=0,y=2.0, size=.5, color=black, label=$w(i)$, position=above]{i_top} 
            \Vertex[x=2,y=2.0, size=.5, color=black, label=$j$, position=above]{j_top}
            
            \Edge[style=solid](i_top)(j_top)

            \draw[->, >=stealth, line width=1.5pt] (1, 1.6) -- (1, 0.4);
            
            \Vertex[x=0,y=0, style={color=white}, label={$\alpha_i=w(i)$}, position=left]{yi_bot} 
            \Vertex[x=2,y=0, style={color=white}, label={$\beta_j=\delta w(i)$}, position=right]{yj_bot}

            \Vertex[x=0,y=0, size=.5, color=black, label=$w(i)$, position=above]{i_bot} 
            \Vertex[x=2,y=0, size=.5, color=black, label=$j$, position=above]{j_bot}

            \Edge[style=solid](i_bot)(j_bot)
        \end{tikzpicture}
        }
        \caption{Postprocessing step.}
        \label{fig:weight11-post}
    \end{subfigure}
\end{minipage}

\caption{Dual updates used in Algorithm~\ref{alg:dual-weight}.}
\label{fig:weight11-updates}
\end{figure}

The following lemma formalizes the dual increase.

\begin{lemma}\label{lem:weighted-dual-objective}
Each time $\mathcal A$ matches a new offline vertex $i\in L$, the dual objective $\displaystyle 
    \sum_{x\in L}\alpha_x+\sum_{y\in R}\beta_y$
increases by at most $\displaystyle
    (1+\delta)w(i)= w(i) / \rho.$
Furthermore, the postprocessing step preserves the total dual objective value.
\end{lemma}

\begin{proof}
We compute the net change in the dual objective; individual variables need not increase in this weighted construction.
If $j$ is matched directly to $i$, then $\alpha_i$ and $\beta_j$ are new
variables, and the update increases the dual objective by
$\delta w(i)+w(i)=(1+\delta)w(i)$.

Suppose now that $j$ is matched via an augmenting $3$-path $\apath{j}{x}{y}{i}$. Since the middle edge $xy$ is not frozen and both budgets are one, neither
$x$ nor $y$ has been reassigned before. Hence $xy$ could not have been
created by an earlier $3$-path: then either $x$ or $y$ would already be
exhausted. Therefore $xy$ was previously created by a direct match. Hence, before the update,
\[
    \alpha_x+\beta_y=\delta w(x)+w(x)=(1+\delta)w(x).
\]
The variables $\alpha_i$ and $\beta_j$ are new, while $\alpha_x$ and $\beta_y$
are updated to
\[
    \alpha_x=w(x),\qquad
    \alpha_i=w(i),\qquad
    \beta_y=w(i),\qquad
    \beta_j=\frac{w(x)}q.
\]
Thus the increase in the dual objective is
\[
    \left(w(x)+2w(i)+\frac{w(x)}q\right)-(1+\delta)w(x)
    =
    2w(i)-\left(\delta-\frac1q\right)w(x).
\]
For $q=1+\sqrt2$ and $\delta=1/\sqrt2$, we have
$\delta-1/q=1-\delta$. Also, $w(i)\le w(x)$: when $y$ was matched directly to
$x$, the vertex $i$ was a free neighbor of $y$, and $x$ was chosen as a
heaviest free neighbor of $y$. Therefore the increase is at most
\[
    2w(i)-(1-\delta)w(i)=(1+\delta)w(i).
\]

Finally, the postprocessing step replaces
$(\alpha_i,\beta_j)=(\delta w(i),w(i))$ with
$(\alpha_i,\beta_j)=(w(i),\delta w(i))$, and hence preserves the dual
objective.
\end{proof}

\begin{proof}[Proof of Theorem~\ref{thm:weighted-lower-bound}]
Let $q=1+\sqrt2$, $\delta=\frac1{\sqrt2}$, and $\rho=\frac1{1+\delta}=2-\sqrt2$.
By Lemmas~\ref{lemma:weighted-pd-template} and
\ref{lem:weighted-dual-objective}, it remains to prove that the final dual
variables produced by Algorithm~\ref{alg:dual-weight} are feasible. 

We use the following observation. If the algorithm matches an online vertex $r$
through an augmenting $3$-path $\apath{r}{x}{s}{z}$ (with $r,s\in R$ and $x,z\in L$), the new edge $rx$
becomes frozen. This holds because the offline vertex $x$ changes partner and immediately
exhausts its reassignment budget. Consequently, the algorithm never reassigns $r$ later in
the execution.

Fix an arbitrary edge $ij\in E$, with $i\in L$ and $j\in R$. We prove that $\alpha_i+\beta_j\ge w(i)$. The cases below are exhaustive. If $i$ is free at the end, then $\alpha_i=0$. If $i$ is matched at the end, then the update rules and the postprocessing step imply that either $\alpha_i=w(i)$ or $\alpha_i=\delta w(i)$. We first rule out the only problematic subcase with $\alpha_i=0$.

\noindent\textbf{Case 0: Both $i$ and $j$ are unmatched at the end.}
This case cannot occur. Because $i$ is unmatched at the end, $i$ was free when
$j$ arrived. Consequently, $D(j)\neq\bot$. The algorithm therefore either matched $j$ directly
or augmented along an eligible $3$-path, ensuring $j$ remains matched.

\noindent\textbf{Case 1: $i$ is free at the end.}
Then $i$ was free throughout the execution. By Case~0, $j$ is matched at the
end. We show that $\beta_j\ge w(i)$.

Suppose first that $j$ is matched directly to $D(j)$ when it arrives. Since
$i$ was free at that time, $w(D(j))\ge w(i)$.

If $j$ is never reassigned, then
either $\beta_j=w(D(j))$, or the postprocessing step is applied and its
condition gives $w(i)\le \delta w(D(j))=\beta_j$. 

If $j$ is later reassigned
through a path $\apath{r}{D(j)}{j}{k}$, then the alternative path $\apath{r}{D(j)}{j}{i}$ was feasible
because $i$ was still free. If this alternative path was eligible, the
lexicographic choice of the augmenting path gives $w(k)\ge w(i)$. The alternative path $\apath{r}{D(j)}{j}{i}$ and the chosen path $\apath{r}{D(j)}{j}{k}$ share the intermediate vertex $D(j)$. Since the chosen path is eligible, the inequality $w(D(j))\ge q\,w(D(r))$ holds. If $\apath{r}{D(j)}{j}{i}$ is not eligible, it must violate the remaining weight inequality, yielding $w(i)<\delta w(D(j))$. By contrast, the chosen path satisfies $w(k)\ge\delta w(D(j))$. In both subcases, the update gives $\beta_j=w(k)\ge w(i)$.

Suppose now that $j$ is matched through an augmenting $3$-path when it arrives.
By the observation above, $j$ is never reassigned later, so
$\beta_j=w(S(j))/q$. Since $i$ was free when $j$ arrived,
$w(D(j))\ge w(i)$. By eligibility, $w(S(j))\ge q\,w(D(j))$, and hence
\[
    \beta_j=\frac{w(S(j))}{q}\ge w(D(j))\ge w(i).
\]

\noindent\textbf{Case 2: $\alpha_i=w(i)$.}
This case is immediate.

\noindent\textbf{Case 3: $\alpha_i=\delta w(i)$.}
Then $i$ was matched directly to some online vertex $y$, the edge $iy$ remains
in the final matching, and the postprocessing step was not applied to $iy$.
Since the edge $iy$ remains in the final matching, neither $i$ nor $y$ is
exhausted. Therefore the postprocessing condition implies that $y$ has a free
neighbor $h$ at the end such that
\[
    w(h)>\delta w(i).
\]
We prove that $\beta_j\ge (1-\delta)w(i)$.

If $j=y$, then $\beta_j=w(i)$, so assume $j\neq y$.

First suppose that $j$ arrives before $y$. Then $i$ is free when $j$ arrives.
If $j$ is matched directly to $D(j)$, then $w(D(j))\ge w(i)$. If $j$ is never
reassigned, then $\beta_j\ge \delta w(D(j))\ge\delta w(i)$, where the factor
$\delta$ accounts for the possible postprocessing step. If $j$ is later
reassigned through a path $\apath{r}{D(j)}{j}{k}$, then the weight inequality for the chosen path's final vertex gives $\beta_j=w(k)\ge\delta w(D(j))\ge\delta w(i)$. In both subcases,
$\beta_j\ge\delta w(i)\ge(1-\delta)w(i)$.

If $j$ is matched through an augmenting $3$-path when it arrives, then
$\beta_j=w(S(j))/q$. Since $w(D(j))\ge w(i)$ and
$w(S(j))\ge q\,w(D(j))$, we get $\beta_j\ge w(i)\ge(1-\delta)w(i)$.

It remains to consider the case in which $j$ arrives after $y$. Since
$h$ is free at the end, it was also free when $j$ arrived. Thus
$\apath{j}{i}{y}{h}$ was a feasible augmenting $3$-path at the arrival
of $j$. Its endpoint condition holds because $w(h)>\delta w(i)$. Here
$D(j)$ denotes the heaviest free neighbor of the arriving vertex $j$ at
that time. Hence, if $\apath{j}{i}{y}{h}$ is not eligible, the only
possible failing condition is the middle condition: the inequality
$w(i)\ge q\,w(D(j))$ fails. Thus $w(i)<q\,w(D(j))$.

If $j$ is matched directly to $D(j)$, then the inequality above gives
$w(i)<q\,w(D(j))$.
If $j$ is never reassigned, then after a possible postprocessing step we still
have
\[
    \beta_j\ge \delta w(D(j))>\frac{\delta}{q}w(i)=(1-\delta)w(i).
\]
If $j$ is later reassigned through a path $\apath{r}{D(j)}{j}{k}$, then the chosen
path is eligible, so $w(k)\ge \delta w(D(j))$. Together with $w(i)<q\,w(D(j))$,
this gives
\[
    \beta_j=w(k)\ge \delta w(D(j))>\frac{\delta}{q}w(i)=(1-\delta)w(i).
\]

Finally suppose that $j$ is matched through an augmenting $3$-path when it
arrives. If $w(i)\ge q\,w(D(j))$, then $\apath{j}{i}{y}{h}$ is eligible, and the lexicographic
choice of the augmenting path gives $w(S(j))\ge w(i)$. If $w(i)<q\,w(D(j))$, then
eligibility of the chosen path gives $w(S(j))\ge q\,w(D(j))>w(i)$. Thus, in all
cases, $w(S(j))\ge w(i)$. Since $j$ is not reassigned later,
\[
    \beta_j=\frac{w(S(j))}{q}\ge \frac{w(i)}{q}\ge (1-\delta)w(i).
\]
Therefore, in Case~3,
\[
    \alpha_i+\beta_j
    \ge
    \delta w(i)+(1-\delta)w(i)
    =
    w(i).
\]

We have proved that every dual constraint is satisfied. By
Lemma~\ref{lem:weighted-dual-objective}, every newly matched offline vertex
increases the dual objective by at most $(1+\delta)$ times its weight, and the
postprocessing step preserves the dual objective. Lemma~\ref{lemma:weighted-pd-template}
therefore gives competitive ratio
\[
    \rho=\frac1{1+\delta}
    =
    \frac1{1+1/\sqrt2}
    =
    2-\sqrt2.\qedhere
\]
\end{proof}

\subsection{\texorpdfstring{$\womp_3(1,1)$}{WOMP3(1,1)}: Upper Bound}
\label{subsec:weighted-11-upper}

We next show that the ratio proved above is best possible for deterministic
algorithms.

\begin{theorem}\label{thm:weighted-upper-bound}
No deterministic online algorithm for offline-vertex-weighted $\womp_3(1,1)$
can have competitive ratio larger than
\[
    \rho:=2-\sqrt2.
\]
\end{theorem}

\begin{proof}
Let $\gamma:=\sqrt2-1$.  We use an adaptive instance.  The offline side is
$L=\{a,b,c,d\}$, with weights
\[
    w(a)=w(b)=w(c)=1,
    \qquad
    w(d)=\gamma.
\]

The first online vertex $r_1$ arrives with $N(r_1)=\{a,b,c\}$.  If the algorithm
leaves $r_1$ unmatched, we stop.  The algorithm gets value $0$, while the optimum
gets value $1$.  Thus we may assume that the algorithm matches $r_1$.  Since
$a,b,c$ have the same weight, we rename the matched offline vertex as $a$.
After the first arrival, the algorithm has the edge $r_1a$, and $b,c$ are free.

The second online vertex $r_2$ arrives with $N(r_2)=\{a,d\}$.  If the algorithm
leaves $r_2$ unmatched, we stop.  The algorithm has value $1$, while the optimum
uses $r_2a$ and $r_1b$, of value $2$.  Hence the ratio is at most $1/2<\rho$.

There are two remaining cases.

\smallskip
\noindent\textbf{Case 1: the algorithm matches $r_2$ directly to $d$.}
Then the algorithm has value $1+\gamma$.  We reveal one more online vertex $r_3$
with $N(r_3)=\{d\}$.  Thus the online side is $R=\{r_1,r_2,r_3\}$.

The algorithm cannot match $r_3$.  Indeed, $d$ is matched to $r_2$, and $r_2$ has no
free neighbor: its only neighbors are $a$ and $d$, and both are matched.  Hence
the final value of the algorithm is
\[
    \alg=1+\gamma.
\]
The optimum uses $r_3d$, $r_2a$, and $r_1b$.  Thus $\opt=2+\gamma$, and
\[
    \frac{\alg}{\opt}
    =
    \frac{1+\gamma}{2+\gamma}
    =
    2-\sqrt2
    =
    \rho.
\]

\smallskip
\noindent\textbf{Case 2: the algorithm augments along a $3$-path.}
Then the algorithm uses either $\apath{r_2}{a}{r_1}{b}$ or
$\apath{r_2}{a}{r_1}{c}$.  Renaming $b$ and $c$ if needed, assume that it uses
$\apath{r_2}{a}{r_1}{b}$.  After the augmentation, the algorithm has the edges
$r_2a$ and $r_1b$.  Since we are in $\womp_3(1,1)$, the offline vertex $a$ and the
online vertex $r_1$ are exhausted.

We reveal two more online vertices $r_3$ and $r_4$, with
\[
    N(r_3)=\{a\},
    \qquad
    N(r_4)=\{b\}.
\]
Thus the online side is $R=\{r_1,r_2,r_3,r_4\}$ in this branch.

The algorithm cannot match $r_3$, because $a$ is matched and exhausted.  It also
cannot match $r_4$.  The only possible $3$-path would use the matched edge
$br_1$ and end at the free neighbor $c$ of $r_1$, but $r_1$ is exhausted.  Thus the
algorithm finishes with value $\alg=2$.

The optimum uses $r_3a$, $r_4b$, $r_2d$, and $r_1c$.  Hence $\opt=3+\gamma$, and
\[
    \frac{\alg}{\opt}
    =
    \frac{2}{3+\gamma}
    =
    2-\sqrt2
    =
    \rho.
\]

In all branches, the algorithm obtains ratio at most $\rho=2-\sqrt2$.  This
proves the theorem.
\end{proof}

\begin{proof}[Proof of Theorem~\ref{thm:weighted-optimal-ratio}]
The algorithmic guarantee is Theorem~\ref{thm:weighted-lower-bound}. The
matching upper bound is Theorem~\ref{thm:weighted-upper-bound}. Hence the
optimal deterministic competitive ratio for offline-vertex-weighted
$\womp_3(1,1)$ is exactly $2-\sqrt2$.
\end{proof}
\subsection{\texorpdfstring{$\womp_3(1,\infty)$}{WOMP3(1,infinity)}: Algorithm and Analysis}
\label{subsec:weighted-1inf}

Our score-greedy algorithm (Algorithm~\ref{alg:score-greedy}) works as follows. Fix
$\lambda\in(0,1)$. Upon the arrival of an online vertex $j\in R$, the algorithm
considers all feasible augmenting paths of length $1$ or $3$ starting at $j$,
assigns each a score, and augments along a path of maximum nonnegative score, if
one exists.

A length-$1$ path $\apath{j}{i}$ has score
\[
    \score(\apath{j}{i}) := w(i),
\]
and a $3$-path $\apath{j}{x}{y}{i}$ has score
\[
    \score(\apath{j}{x}{y}{i}) := w(i)-\lambda w(x).
\]
At each arrival, the algorithm performs a feasible augmenting path of maximum
nonnegative score, if such a path exists. Ties can be broken by any deterministic rule; the analysis uses only maximality of the score.

\begin{algorithm}[ht]
  \SetAlgoLined
  \DontPrintSemicolon
  \KwIn{$L$ offline side and $w$ weight function on $L$}
  \KwOut{A matching $M$}

  Initialize $M\gets \emptyset$\;

  \ForAll{arrivals $j\in R$}{
    Let $\mathcal P(j)$ be the set of feasible augmenting paths of length $1$
    or $3$ starting at $j$\;

    \ForAll{$P\in\mathcal P(j)$}{
      \lIf{$P=\apath{j}{i}$}{
        Set $\score(P)\gets w(i)$
      }
      \lElseIf{$P=\apath{j}{x}{y}{i}$}{
        Set $\score(P)\gets w(i)-\lambda w(x)$
      }
    }
    \uIf{$\mathcal P(j)$ contains a path $P$ with $\score(P)\ge 0$}{
      Choose a path $P(j)\in\mathcal P(j)$ maximizing $\score(P(j))$.
      Augment along $P(j)$.
    }
  }

  \caption{Score-greedy for offline-weighted $\womp_3(1,\infty)$.}
  \label{alg:score-greedy}
\end{algorithm}

We prove that Algorithm~\ref{alg:score-greedy} is $\left(\frac{\sqrt 5 -1}{2}\right)$-competitive.

\begin{theorem}\label{thm:weighted-infinity-lower-bound}
Choosing $\lambda=\frac{3-\sqrt 5}{2}$, the score-greedy
Algorithm~\ref{alg:score-greedy} is
$1-\lambda= \frac{\sqrt 5 -1}{2}$-competitive for
offline-vertex-weighted $\womp_3(1,\infty)$.
\end{theorem}

\begin{proof}
Let $\lambda:=\frac{3-\sqrt5}{2}$ and
$\rho:=1-\lambda=\frac{\sqrt5-1}{2}$.  Then
$\lambda=\rho^2$, $1+\rho=1/\rho$, $\rho>\lambda$, and $2\rho>1$.
We prove the theorem with Lemma~\ref{lemma:weighted-pd-template}.

\medskip
\noindent
\textbf{Dual updates.}
Initially all dual variables are zero.  If the algorithm matches an online
vertex $j$, let $P_j$ be the path chosen when $j$ arrives, and set
$\beta_j:=\score(P_j)$.  If $j$ is not matched, keep $\beta_j=0$.

If $P_j=\apath{j}{i}$, then set $\beta_j=w(i)$ and
$\alpha_i:=\rho w(i)$.  This creates a new history component with root $j$
and active offline vertex $i$.  The dual objective increases by
\[
    w(i)+\rho w(i)=(1+\rho)w(i)=\frac{w(i)}{\rho}.
\]

If $P_j=\apath{j}{x}{r}{i}$, then $xr$ is the active edge of the component
rooted at $r$ just before the update.  Set
\[
    \alpha_x:=w(x),\qquad
    \beta_j:=w(i)-\lambda w(x),\qquad
    \alpha_i:=\rho w(i).
\]
The vertex $i$ becomes the new active offline vertex of that component.  The
variable $\beta_r$ is not changed.  Also $\beta_j\ge 0$, because the algorithm
uses only paths of nonnegative score.  The net increase in the dual objective is
\[
    (1-\rho)w(x)+\bigl(w(i)-\lambda w(x)\bigr)+\rho w(i)
    =(1+\rho)w(i)
    =
    \frac{w(i)}{\rho},
\]
where we use $1-\rho=\lambda$.

Thus, whenever a new offline vertex $i$ is matched, the dual objective
increases by exactly $w(i)/\rho$.

\medskip
\noindent
\textbf{Postprocessing.}
Let $C$ be a final history component.  Write
$r:=\rootof(C)$ and $a:=\act(C)$.  Let $i_0,\ldots,i_q$ be the offline
vertices of $C$, with $i_q=a$.  Thus $i_0,\ldots,i_{q-1}$ are exhausted.
For all $h\in\{0,\ldots,q-1\}$, let $j_h$ be the online vertex that created the
transition from $i_h$ to $i_{h+1}$.  The final matching inside $C$ contains
$r i_q$ and $j_h i_h$ for all $h\in\{0,\ldots,q-1\}$.

Set $W_h:=w(i_h)$,
$W_C^{\min}:=w(a)=W_q$, and $W_C^{\max}:=W_0$.  Before postprocessing,
the online updates give
\[
    \beta_r=W_C^{\max},
    \qquad
    \alpha_a=\rho W_C^{\min},
\]
and, for all $h\in\{0,\ldots,q-1\}$,
\[
    \alpha_{i_h}=W_h,
    \qquad
    \beta_{j_h}=W_{h+1}-\lambda W_h.
\]

Now define
\[
    F_C:=\max\{w(f): f \text{ is a final free neighbor of } r\},
\]
with $F_C=0$ if $r$ has no final free neighbor.

\paragraph{Monotonicity of weights in a component.}
For every final component $C$, we have
\begin{equation}\label{eq:weighted-component-monotonicity}
    W_C^{\max}=W_0\ge W_1\ge \cdots \ge W_q=W_C^{\min}\ge F_C.
\end{equation}

Indeed, first consider
$h=0$.  When $r$ arrived, the vertices $i_0,\ldots,i_q$ were free.  Every
final free neighbor of $r$ was also free.  Moreover, all edges from $r$ to
these vertices had already been revealed.  Hence all direct paths from $r$ to
these vertices were feasible.  The algorithm chose $\apath{r}{i_0}$, so
score-greediness gives $W_0\ge W_k$ for every $k\ge 0$ and $W_0\ge F_C$.

Now fix $h\in\{1,\ldots,q\}$.  When $j_{h-1}$ arrived, the edge
$i_{h-1}r$ was the active matched edge of the component rooted at $r$.  At that
time, the vertices $i_h,\ldots,i_q$ were free.  Every final free neighbor of
$r$ was also free.  Moreover, all edges from $r$ to these vertices had already
been revealed.

Let $z$ be either one of the vertices $i_h,\ldots,i_q$, or a final free
neighbor of $r$.  By the previous paragraph, the path
$\apath{j_{h-1}}{i_{h-1}}{r}{z}$ was feasible when $j_{h-1}$ arrived.  All
these paths have the same middle offline vertex $i_{h-1}$.  Hence their scores
differ only in the weight of the last endpoint.  The algorithm chose the path
ending at $i_h$.  Thus $W_h=w(i_h)\ge w(z)$ for every such vertex $z$.

Taking $z=i_{h+1}$ gives $W_h\ge W_{h+1}$ for all $h\in\{0,\ldots,q-1\}$.  Taking
$h=q$ and taking $z$ to be a maximum-weight final free neighbor gives
$W_q\ge F_C$.

\paragraph{Final dual lower bounds for the root.} After postprocessing, the root $r$ of every final component $C$ satisfies $\beta_r\ge F_C$ and $\beta_r\ge \rho W_C^{\max}$. We verify this below.

The purpose of the postprocessing step is to transfer dual mass from the root $r$ to the active offline vertex $a$. This transfer compensates for the potential presence of a heavy final free neighbor of $r$, which otherwise threatens the dual feasibility of edges arriving later at $a$. Define $\eta_C:=\min\{\lambda W_C^{\min},\, W_C^{\max}-F_C\}$.  For the
component $C$, perform the shift
\[
    \beta_r\gets \beta_r-\eta_C,
    \qquad
    \alpha_a\gets \alpha_a+\eta_C.
\]
All other dual variables remain unchanged.  This shift preserves the dual
objective.  It also keeps the variables nonnegative, because $\eta_C\ge 0$ and
$\beta_r=W_C^{\max}-\eta_C\ge F_C\ge 0$.

\paragraph{A bound for roots.}
We will use one simple bound for roots of final components.  Let $C$ be a final
component, and write $r:=\rootof(C)$.  After postprocessing,
$\beta_r=W_C^{\max}-\eta_C$.  Moreover,
\begin{equation}\label{eq:weighted-root-bound}
    \beta_r\ge F_C
    \qquad\text{and}\qquad
    \beta_r\ge \rho W_C^{\max}.
\end{equation}
The first inequality follows from the definition of $\eta_C$.  For the second
one, use \eqref{eq:weighted-component-monotonicity}:
\[
    \eta_C
    \le
    \lambda W_C^{\min}
    \le
    \lambda W_C^{\max}.
\]
Thus
\[
    \beta_r
    =
    W_C^{\max}-\eta_C
    \ge
    (1-\lambda)W_C^{\max}
    =
    \rho W_C^{\max}.
\]

\medskip
\noindent
\textbf{Dual feasibility.}
All that remains is to prove that the final dual variables are feasible. Fix an edge $ij\in E$, with $i\in L$ and $j\in R$.  We prove
$\alpha_i+\beta_j\ge w(i)$.

There are three cases for $i$.  At the end, it is either exhausted, final free,
or the active offline vertex of its final component.  These cases are
exhaustive.  An unmatched offline vertex is final free.  A matched offline
vertex belongs to a final component.  In that component, it is either active or
it has already been reassigned once.  Since the offline budget is $1$, the
second case means that it is exhausted.

\smallskip
\noindent
\textbf{Case 1: $i$ is exhausted.}
When $i$ was reassigned, the dual update set $\alpha_i=w(i)$.  Hence
$\alpha_i+\beta_j\ge w(i)$.

\smallskip
\noindent
\textbf{Case 2: $i$ is final free.}
Then $\alpha_i=0$.  We prove that $\beta_j\ge w(i)$.

\smallskip
\noindent
\textbf{Case 2a: $j$ is unmatched at the end.}
This case cannot occur.  Since $i$ is final free, it was free when $j$
arrived.  Hence the direct path $\apath{j}{i}$ was feasible and had score
$w(i)\ge 0$.  The algorithm would have matched $j$, a contradiction.

\smallskip
\noindent
\textbf{Case 2b: $j$ is matched but is not the root of a final component.}
Again, $i$ was free when $j$ arrived.  Hence the path $\apath{j}{i}$ was
feasible and had score $w(i)$.  Since $j$ is not the root of a final component,
postprocessing does not change $\beta_j$.  Thus the final value of $\beta_j$
is still $\score(P_j)$.  By score-greediness, $\score(P_j)\ge w(i)$.  Hence
$\beta_j\ge w(i)$.

\smallskip
\noindent
\textbf{Case 2c: $j$ is the root of a final component $C'$.}
Then $i$ is a final free neighbor of $j=\rootof(C')$.  Hence $i$ is counted in
$F_{C'}$, so $F_{C'}\ge w(i)$.  By \eqref{eq:weighted-root-bound} applied to
$C'$, we have $\beta_j\ge F_{C'}$.  Therefore $\beta_j\ge w(i)$.

This proves Case~2.
\smallskip
\noindent
\textbf{Case 3: $i$ is the active offline vertex of its final component.}
Let $C$ be this final component, and write $r:=\rootof(C)$.  Thus
$i=\act(C)$.  Recall that we fixed $ij\in E$.  By postprocessing,
\[
    \alpha_i=\rho W_C^{\min}+\eta_C.
\]

If $j=r$, then postprocessing gives $\beta_j=W_C^{\max}-\eta_C$.  Hence
\[
    \alpha_i+\beta_j
    =
    (\rho W_C^{\min}+\eta_C)+(W_C^{\max}-\eta_C)
    =
    \rho W_C^{\min}+W_C^{\max}
    \ge
    W_C^{\min}
    =
    w(i),
\]
where the inequality follows from
\eqref{eq:weighted-component-monotonicity}.  Thus $ij$ is covered.  Hence we
assume $j\ne r$.

When $j$ arrived, the offline vertex $i$ was either free or matched to $r$.
We consider these two cases.

\smallskip
\noindent
\textbf{Case 3a: $i$ was free when $j$ arrived.}
When $j$ arrived, $\apath{j}{i}$ was feasible: $j$ was new, $i$ was free, and
$ij\in E$.  Its score was $w(i)=W_C^{\min}$.

If $j$ is unmatched at the end, then the algorithm did not match $j$ when it
arrived.  This is impossible, because $\apath{j}{i}$ was feasible and had
nonnegative score.

Suppose next that $j$ is matched but is not the root of a final
component.  Then postprocessing does not change $\beta_j$.  Hence the final
value of $\beta_j$ is still $\score(P_j)$, where $P_j$ is the path chosen when
$j$ arrived.  Since $\apath{j}{i}$ was feasible then and had score
$W_C^{\min}$, score-greediness gives
\[
    \score(P_j)\ge W_C^{\min}.
\]
Also $\alpha_i\ge 0$.  Therefore
\[
    \alpha_i+\beta_j
    \ge
    \beta_j
    =
    \score(P_j)
    \ge
    W_C^{\min}
    =
    w(i).
\]
Thus $ij$ is covered.

It remains to consider the case where $j$ is the root of another final
component $C'$. Then the path chosen when $j$ arrived was direct. Its score
was $W_{C'}^{\max}$. Since $\apath{j}{i}$ was feasible at the same time and
had score $W_C^{\min}$, score-greediness gives
$W_{C'}^{\max}\ge W_C^{\min}$. By \eqref{eq:weighted-root-bound} applied to
$C'$, we have $\beta_j\ge \rho W_{C'}^{\max}\ge \rho W_C^{\min}$. Thus
$\alpha_i+\beta_j\ge 2\rho W_C^{\min}\ge W_C^{\min}=w(i)$, because
$2\rho>1$. Hence $ij$ is covered.

\smallskip
\noindent
\textbf{Case 3b: $i$ was matched to $r$ when $j$ arrived.}
Recall that $\eta_C=\min\{\lambda W_C^{\min},\, W_C^{\max}-F_C\}$. If
$\eta_C=\lambda W_C^{\min}$, then
$\alpha_i=\rho W_C^{\min}+\lambda W_C^{\min}=W_C^{\min}=w(i)$, so $ij$ is
covered. Hence we assume
$\eta_C=W_C^{\max}-F_C<\lambda W_C^{\min}$.

Since $r=\rootof(C)$, postprocessing gives
$\beta_r=W_C^{\max}-\eta_C=F_C$. By \eqref{eq:weighted-root-bound} applied to
$C$ and \eqref{eq:weighted-component-monotonicity},
\[
    F_C\ge \rho W_C^{\max}\ge \rho W_C^{\min}.
\]
Hence
$F_C-\lambda W_C^{\min}\ge(\rho-\lambda)W_C^{\min}\ge0$.
Also, the strict inequality
$\eta_C=W_C^{\max}-F_C<\lambda W_C^{\min}$ implies $F_C>0$. Let $f$ be a
final free neighbor of $r$ with $w(f)=F_C$.

We claim that $\apath{j}{i}{r}{f}$ was feasible when $j$ arrived.  We have
$ij\in E$.  Also, $i$ was matched to $r$ at that time.  Since $i=\act(C)$, the
offline vertex $i$ had not been used before as the middle offline vertex of a
length-$3$ augmentation; otherwise it would not be the active offline vertex
of the final component.  Thus using $i$ as the middle offline vertex was
allowed.  The vertex $f$ was free at that time, because $f$ is final free.  The
pair $rf$ had already been revealed when $r$ arrived.  Finally, $r$ had
infinite online budget.  Hence $\apath{j}{i}{r}{f}$ was feasible.  Its score
was $F_C-\lambda W_C^{\min}$, which is nonnegative by the inequality above. We now split depending on whether $j$ is matched or not.

If $j$ is unmatched at the end, then the algorithm did not match $j$ when it
arrived.  This is impossible, because $\apath{j}{i}{r}{f}$ was feasible and
had nonnegative score.

Suppose next that $j$ is matched at the end but is not the root of a final
component.  Then postprocessing does not change $\beta_j$.  Hence the final
value of $\beta_j$ is still $\score(P_j)$.  Since $\apath{j}{i}{r}{f}$ was
feasible when $j$ arrived and had score $F_C-\lambda W_C^{\min}$,
score-greediness gives
\[
    \score(P_j)\ge F_C-\lambda W_C^{\min}.
\]
Using $\eta_C=W_C^{\max}-F_C$, we get
\[
\begin{aligned}
    \alpha_i+\beta_j
    &\ge
    \rho W_C^{\min}
    +(W_C^{\max}-F_C)
    +(F_C-\lambda W_C^{\min})  \\
    &=
    W_C^{\max}+(\rho-\lambda)W_C^{\min}
    \ge
    W_C^{\min}
    =
    w(i),
\end{aligned}
\]
where the last inequality follows from
\eqref{eq:weighted-component-monotonicity}.  Thus $ij$ is covered.

Finally, suppose that $j$ is the root of another final
component $C'$.  Then the path chosen when $j$ arrived was direct.  Its score
was $W_{C'}^{\max}$.  Since $\apath{j}{i}{r}{f}$ was feasible at the same time
and had score $F_C-\lambda W_C^{\min}$, score-greediness gives
\[
    W_{C'}^{\max}\ge F_C-\lambda W_C^{\min}.
\]
By \eqref{eq:weighted-root-bound} applied to $C'$,
\[
    \beta_j
    \ge
    \rho W_{C'}^{\max}
    \ge
    \rho(F_C-\lambda W_C^{\min}).
\]
Combining this with $\alpha_i=\rho W_C^{\min}+\eta_C$ and
$\eta_C=W_C^{\max}-F_C$, we get
\[
\begin{aligned}
    \alpha_i+\beta_j
    &\ge
    \rho W_C^{\min}
    +(W_C^{\max}-F_C)
    +\rho F_C-\rho\lambda W_C^{\min}  \\
    &=
    W_C^{\max}
    -(1-\rho)F_C
    +\rho(1-\lambda)W_C^{\min} \\
    &=
    W_C^{\max}
    -\lambda F_C
    +\rho^2 W_C^{\min} \\
    &=
    W_C^{\max}
    -\lambda F_C
    +\lambda W_C^{\min}.
\end{aligned}
\]
Here we used $1-\rho=\lambda$, $1-\lambda=\rho$, and $\rho^2=\lambda$.
By \eqref{eq:weighted-component-monotonicity}, $F_C\le W_C^{\min}$.  Hence the
last expression is at least $W_C^{\max}$.  By the same equation,
$W_C^{\max}\ge W_C^{\min}=w(i)$.  Therefore $\alpha_i+\beta_j\ge w(i)$.

The cases above cover every edge $ij\in E$.  Hence the final dual variables are
feasible for $(D_{G,w})$.  The postprocessing step preserves the dual objective.
Each online update increases the dual objective by exactly $w(i)/\rho$ when a
new offline vertex $i$ is matched.  Lemma~\ref{lemma:weighted-pd-template}
therefore implies that Algorithm~\ref{alg:score-greedy} is
$\rho$-competitive.  Since $\rho=1-\lambda=(\sqrt5-1)/2$, the theorem follows.
\end{proof}

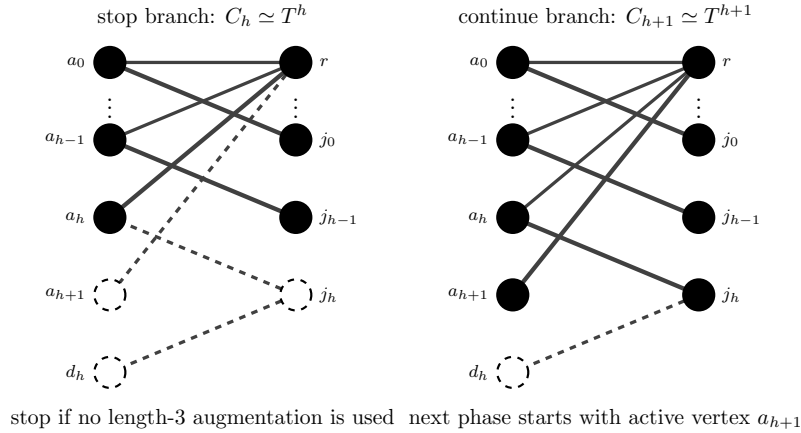
\begin{figure}[ht]
\centering
\scalebox{0.82}{
\begin{tikzpicture}
    \useasboundingbox (-1.0,-6.20) rectangle (11.20,1.20);

\node at (1.5,0.75) {\small stop branch: $C_h\simeq T^h$};

\Vertex[x=0,y=0, size=.5, color=black, style=solid,
        label=$a_0$, position=left]{L-a0}
\Vertex[x=0,y=-1.25, size=.5, color=black, style=solid,
        label=$a_{h-1}$, position=left]{L-ahm}
\Vertex[x=0,y=-2.5, size=.5, color=black, style=solid,
        label=$a_h$, position=left]{L-ah}
\Vertex[x=0,y=-3.75, size=.5, color=white, style=dashed,
        label=$a_{h+1}$, position=left]{L-ahp}
\Vertex[x=0,y=-5.0, size=.5, color=white, style=dashed,
        label=$d_h$, position=left]{L-dh}

\Vertex[x=3,y=0, size=.5, color=black, style=solid,
        label=$r$, position=right]{L-r}
\Vertex[x=3,y=-1.25, size=.5, color=black, style=solid,
        label=$j_0$, position=right]{L-j0}
\Vertex[x=3,y=-2.5, size=.5, color=black, style=solid,
        label=$j_{h-1}$, position=right]{L-jhm}
\Vertex[x=3,y=-3.75, size=.5, color=white, style=dashed,
        label=$j_h$, position=right]{L-jh}

\node at (0,-0.63) {$\vdots$};
\node at (3,-0.63) {$\vdots$};

\Edge[style=solid](L-a0)(L-r)
\Edge[style=solid](L-ahm)(L-r)

\Edge[style={line width=2pt}](L-a0)(L-j0)
\Edge[style={line width=2pt}](L-ahm)(L-jhm)
\Edge[style={line width=2pt}](L-ah)(L-r)

\Edge[style=dashed](L-jh)(L-ah)
\Edge[style=dashed](L-r)(L-ahp)
\Edge[style=dashed](L-jh)(L-dh)

\node at (1.5,-5.75) {\small stop if no length-$3$ augmentation is used};

\node at (8.0,0.75) {\small continue branch: $C_{h+1}\simeq T^{h+1}$};

\Vertex[x=6.5,y=0, size=.5, color=black, style=solid,
        label=$a_0$, position=left]{R-a0}
\Vertex[x=6.5,y=-1.25, size=.5, color=black, style=solid,
        label=$a_{h-1}$, position=left]{R-ahm}
\Vertex[x=6.5,y=-2.5, size=.5, color=black, style=solid,
        label=$a_h$, position=left]{R-ah}
\Vertex[x=6.5,y=-3.75, size=.5, color=black, style=solid,
        label=$a_{h+1}$, position=left]{R-ahp}
\Vertex[x=6.5,y=-5.0, size=.5, color=white, style=dashed,
        label=$d_h$, position=left]{R-dh}

\Vertex[x=9.5,y=0, size=.5, color=black, style=solid,
        label=$r$, position=right]{R-r}
\Vertex[x=9.5,y=-1.25, size=.5, color=black, style=solid,
        label=$j_0$, position=right]{R-j0}
\Vertex[x=9.5,y=-2.5, size=.5, color=black, style=solid,
        label=$j_{h-1}$, position=right]{R-jhm}
\Vertex[x=9.5,y=-3.75, size=.5, color=black, style=solid,
        label=$j_h$, position=right]{R-jh}

\node at (6.5,-0.63) {$\vdots$};
\node at (9.5,-0.63) {$\vdots$};

\Edge[style=solid](R-a0)(R-r)
\Edge[style=solid](R-ahm)(R-r)
\Edge[style=solid](R-ah)(R-r)

\Edge[style={line width=2pt}](R-a0)(R-j0)
\Edge[style={line width=2pt}](R-ahm)(R-jhm)
\Edge[style={line width=2pt}](R-ah)(R-jh)
\Edge[style={line width=2pt}](R-ahp)(R-r)

\Edge[style=dashed](R-jh)(R-dh)

\node at (8.0,-5.75) {\small next phase starts with active vertex $a_{h+1}$};

\end{tikzpicture}
}
\caption{The history-graph branch at phase $h$ of the weighted upper-bound
construction.  Offline vertices are drawn on the left and online vertices on
the right.  Thin solid edges belong to the history graph but are no longer in
the current matching.  Thick solid edges are current matching edges.  Dashed
edges are revealed edges that are not in the history graph.  If the algorithm
does not augment along $\apath{j_h}{a_h}{r}{a_{h+1}}$, the construction stops.
If it does augment, the component grows from type $T^h$ to type $T^{h+1}$.}
\label{fig:weighted-upper-history}
\end{figure}

\subsection{\texorpdfstring{$\womp_3(1,\infty)$}{WOMP3(1,infinity)}: Upper Bound}
\label{subsec:weighted-1inf-upper}

We now show that the ratio $\tau=(\sqrt5-1)/2$ is best possible.  The hard
instance is parameterized by an integer $K$, and its analysis naturally leads to
a one-parameter family of sequences $\{p_h(\theta)\}_{h\ge 0}$ and to a specific
choice of $\theta=\theta_K$.

\begin{lemma}\label{lem:weighted-infinity-sequence}
Let $\tau=(\sqrt5-1)/2$.  For a parameter $\theta\in(\tau,1)$ define a sequence
$\{p_h(\theta)\}_{h\ge 0}$ by
\[
    p_0(\theta):=\frac{2\theta-1}{1-\theta},
    \qquad
    p_{h+1}(\theta):=\frac{p_h(\theta)}{1-\theta}-1 .
\]
Then, for every integer $K\ge 0$, there exists a unique value
$\theta_K\in(\tau,1)$ such that $p_K(\theta_K)=1$.  Moreover,
$\theta_K\downarrow \tau$ as $K\to\infty$.

Fix $K$ and write $p_h:=p_h(\theta_K)$.  Then, for every $h\in\{0,\ldots,K\}$,
\[
    h+1+p_h
    =
    \theta_K\left(h+2+\sum_{s=0}^h p_s\right).
\]
\end{lemma}

\begin{proof}
We begin by checking the boundary value $\theta=\tau$.  Since $1-\tau=\tau^2$, we
have $p_0(\tau)=\tau$, and if $p_h(\tau)=\tau$ then
\[
    p_{h+1}(\tau)=\frac{\tau}{1-\tau}-1=\tau.
\]
Hence $p_K(\tau)=\tau<1$ for all $K$.

Next fix $\theta\in(\tau,1)$.  Each function $p_h(\theta)$ is continuous and
strictly increasing in $\theta$: this is immediate for $p_0$, and the
recurrence preserves these properties.  Moreover, $p_K(\theta)\to\infty$ as
$\theta\to 1$.  Therefore, by the intermediate value theorem, for each $K$ there
exists a unique $\theta_K\in(\tau,1)$ such that $p_K(\theta_K)=1$.

We now prove the stated identity.  For $h=0$ it is exactly the definition of
$p_0$, namely $1+p_0=\theta_K(2+p_0)$.  Assume it holds for some $h$ and write
$S_h:=\sum_{s=0}^h p_s$.  Then
\[
    h+1+p_h=\theta_K(h+2+S_h).
\]
Using the recurrence $p_{h+1}=p_h/(1-\theta_K)-1$ and rearranging yields
\[
    h+2+p_{h+1}=\theta_K(h+3+S_h+p_{h+1}),
\]
which is the desired identity for $h+1$.

Finally, we show that $\theta_K\downarrow\tau$.  Since $p_K(\theta_K)=1$, we get
\[
    p_{K+1}(\theta_K)=\frac{1}{1-\theta_K}-1=\frac{\theta_K}{1-\theta_K}>1.
\]
Because $p_{K+1}(\cdot)$ is strictly increasing and $p_{K+1}(\tau)=\tau<1$, the
(unique) solution to $p_{K+1}(\theta)=1$ satisfies $\theta_{K+1}<\theta_K$.
Thus $(\theta_K)$ is decreasing and bounded below by $\tau$, so it converges to
some limit $\bar\theta\ge\tau$.

If $\bar\theta>\tau$, choose $\theta'\in(\tau,\bar\theta)$.  Then $\theta_K\ge\theta'$
for all $K$, and by monotonicity $1=p_K(\theta_K)\ge p_K(\theta')$.  But for any
fixed $\theta'>\tau$, the recurrence $p_{h+1}=p_h/(1-\theta')-1$ makes
$p_h(\theta')\to\infty$ as $h\to\infty$, a contradiction. Therefore, $\bar\theta=\tau$.
\end{proof}

\begin{theorem}\label{thm:weighted-infinity-upper-bound}
No deterministic online algorithm for offline-vertex-weighted $\womp_3(1,\infty)$
can have competitive ratio larger than
\[
    \tau=\frac{\sqrt5-1}{2}.
\]
\end{theorem}

\begin{proof}
Fix $K\ge 0$, and let $\theta_K$ be as in
Lemma~\ref{lem:weighted-infinity-sequence}.  In this proof, set
\[
    q_h:=p_h(\theta_K),
    \quad \forall\, h\in\{0,\ldots,K\}.
\]
Thus $q_K=1$, and Lemma~\ref{lem:weighted-infinity-sequence} gives
\begin{equation}\label{eq:weighted-upper-identity}
    h+1+q_h
    =
    \theta_K\left(h+2+\sum_{s=0}^h q_s\right),
    \quad \forall\, h\in\{0,\ldots,K\}.
\end{equation}

We construct an adaptive instance where every deterministic algorithm gets
ratio at most $\theta_K$.

\paragraph{Initial construction.}
The offline vertices are $
    a_0,a_1,\ldots,a_{K+1}$ and $d_0,d_1,\ldots,d_K$.
Set
\[
    w(a_h)=1, \quad \forall\, h\in\{0,\ldots,K+1\},
    \qquad
    w(d_h)=q_h, \quad \forall\, h\in\{0,\ldots,K\}.
\]
The labels of the vertices $a_h$ will be assigned adaptively.  This is harmless,
because all these vertices have weight $1$ and have the same neighborhood until
one of them is selected by the algorithm.

The first online vertex $r$ arrives with
\[
    N(r)=\{a_0,a_1,\ldots,a_{K+1}\}.
\]
If the algorithm leaves $r$ unmatched, we stop.  The algorithm gets value $0$,
while the optimum gets value $1$.  Thus the ratio is $0$.  Hence we may assume
that the algorithm matches $r$.  We rename the offline vertex matched to $r$ as
$a_0$.

\paragraph{The phases.}
We now describe the adaptive part of the instance in terms of the history graph.
The construction tries to keep the history graph as one main component for as
long as possible.  At the beginning of phase $h$, for all $h\in\{0,\ldots,K\}$, the main
component is of type $T^h$.  Its root is $r$ (i.e., $\rootof(C)=r$), and its active offline vertex is
$a_h$.  Equivalently, $r$ is currently matched to $a_h$.  The vertices
$a_0,\ldots,a_{h-1}$ are exhausted, and the vertices
$a_{h+1},\ldots,a_{K+1}$ are still free.

The online vertex $j_h$ arrives with
\[
    N(j_h)=\{a_h,d_h\}.
\]
Since $r$ is adjacent to the still-free vertices
$a_{h+1},\ldots,a_{K+1}$, every length-$3$ augmentation that keeps the main
component alive has the form
\[
    \apath{j_h}{a_h}{r}{a}
\]
for some free vertex $a\in\{a_{h+1},\ldots,a_{K+1}\}$.  If the algorithm uses
such an augmentation, we rename that endpoint as $a_{h+1}$.  Then the next
phase starts with root $r$ and active offline vertex $a_{h+1}$.
Figure~\ref{fig:weighted-upper-history} shows this history-graph branch.

\paragraph{The two cases.}

\textbf{Case 1: The construction stops at some phase $h$.}
Fix the first phase $h$ at which the algorithm does not keep the main component
alive.  Thus, when $j_h$ arrives, the algorithm does not use a length-$3$
augmentation of the form
\[
    \apath{j_h}{a_h}{r}{a}
\]
with $a\in\{a_{h+1},\ldots,a_{K+1}\}$ free.  There are two subcases.

\smallskip
\noindent
\textbf{Case 1a: The algorithm leaves $j_h$ unmatched.}
We reveal one more online vertex $u_h^d$ adjacent only to $d_h$.  Then we stop
the main construction.  The algorithm can match at most one online vertex to
$d_h$.  Hence the total contribution of $d_h$ to the algorithm is at most
$q_h$.

\smallskip
\noindent
\textbf{Case 1b: The algorithm matches $j_h$ directly to $d_h$.}
Then $j_h$ starts a new component instead of extending the main component.  We
reveal one more online vertex $u_h^d$ adjacent only to $d_h$.  Then we stop the
main construction.  This new vertex cannot increase the algorithm's value.
Indeed, $d_h$ is already matched to $j_h$.  Also, $u_h^d$ has no neighbor other
than $d_h$.  The only possible alternating walk of length $3$ starting from
$u_h^d$ would have to go through $d_h$ and $j_h$, and then to $a_h$.  But
$a_h$ is not free.  Thus no feasible length-$3$ augmentation starts at
$u_h^d$.  Hence the total contribution of $d_h$ to the algorithm is exactly
$q_h$.

\smallskip
\noindent

In both subcases, the contribution of $d_h$ to the algorithm is at most $q_h$.
We now finish this instance.  For each $s<h$, we reveal one online vertex
$u_s^a$ adjacent only to $a_s$, and then no more online vertices arrive.

The algorithm cannot match any of the vertices $u_s^a$ with $s<h$.  Indeed,
each $a_s$ with $s<h$ was matched to $r$ at the beginning of phase $s$, and
was later used as the middle offline vertex of the length-$3$ augmentation in
that phase.  Since the offline budget is $1$, all vertices
$a_0,\ldots,a_{h-1}$ are exhausted.

At this point the algorithm has matched the offline vertices
$a_0,\ldots,a_h$ in the main component.  These vertices have total weight
$h+1$.  The only possible additional contribution is from $d_h$, and this
contribution is at most $q_h$.  Therefore, $\alg\le h+1+q_h$.

We now lower bound the optimum on the same revealed instance.  The optimum
matches
\[
    u_s^a\text{ to }a_s \quad (s<h),
    \qquad
    j_s\text{ to }d_s \quad (s<h),
\]
and also matches
\[
    j_h\text{ to }a_h,
    \qquad
    u_h^d\text{ to }d_h,
    \qquad
    r\text{ to }a_{h+1}.
\]
All these edges have been revealed.  The chosen online vertices are distinct,
and the chosen offline vertices are distinct.  Hence $
    \opt
    \ge
    h+2+\sum_{s=0}^h q_s $. 
Combining this with \eqref{eq:weighted-upper-identity}, we get
\[
    \frac{\alg}{\opt}
    \le
    \frac{h+1+q_h}{h+2+\sum_{s=0}^h q_s}
    =
    \theta_K .
\]

\smallskip
\noindent
\textbf{Case 2: The construction reaches the end of phase $K$.}
It remains to consider the case in which Case~1 never occurs.  Thus, at every
phase $h\in\{0,\ldots,K\}$, the algorithm uses a length-$3$ augmentation that keeps
the main component alive.  After each such augmentation, we rename the new free
endpoint as $a_{h+1}$.  Hence, after phase $K$, the main component has type
$T^{K+1}$.  The algorithm has matched the offline vertices
$a_0,a_1,\ldots,a_{K+1}$, and its value is $
    \alg=K+2.$

We now reveal one online vertex $u_s^a$ adjacent only to $a_s$, for each
$s\in\{0,\ldots,K\}$, and then no more online vertices arrive.  The algorithm cannot
match any of these vertices.  Indeed, each $a_s$ with $s\le K$ was used as the
middle offline vertex of the length-$3$ augmentation in phase $s$.  Since the
offline budget is $1$, all vertices $a_0,\ldots,a_K$ are exhausted.

The optimum matches
\[
    u_s^a\text{ to }a_s, \quad \forall\, s\in\{0,\ldots,K\},
    \qquad
    j_s\text{ to }d_s, \quad \forall\, s\in\{0,\ldots,K\},
\]
and also matches $r$ to $a_{K+1}$.  All these edges have been revealed.  The
chosen online vertices are distinct, and the chosen offline vertices are
distinct.  Therefore, $\opt
    \ge
    K+2+\sum_{s=0}^K q_s$.
    
Since $q_K=1$, \eqref{eq:weighted-upper-identity} with $h=K$ gives
\[
    K+2
    =
    \theta_K\left(K+2+\sum_{s=0}^K q_s\right).
\]
Thus
\[
    \frac{\alg}{\opt}
    \le
    \frac{K+2}{K+2+\sum_{s=0}^K q_s}
    =
    \theta_K .
\]

For every $K$, no deterministic algorithm can guarantee ratio larger than
$\theta_K$.  Since $\theta_K\downarrow\tau$, no deterministic algorithm can guarantee ratio larger than $\tau$.
\end{proof}

\begin{proof}[Proof of Theorem~\ref{thm:weighted-unbounded-optimal-ratio}]
The algorithmic guarantee is Theorem~\ref{thm:weighted-infinity-lower-bound}.
The matching upper bound is Theorem~\ref{thm:weighted-infinity-upper-bound}.
Hence the optimal deterministic competitive ratio for offline-vertex-weighted
$\womp_3(1,\infty)$ is exactly $(\sqrt5-1)/2$.
\end{proof}
\subsection{Online-side Vertex Weights}
\label{subsec:weighted-online-side}

The weighted variant considered above assigns weights only to offline
vertices. This restriction is necessary for any positive deterministic
guarantee under adversarial arrivals if the algorithm is only allowed to
augment.

Indeed, consider a single offline vertex $i$. The first online vertex
$r_1$ has weight $1$ and is adjacent only to $i$. If a deterministic
algorithm leaves $r_1$ unmatched, the adversary stops, and the ratio is
zero. Otherwise, the algorithm matches $r_1$ to $i$. The adversary then
reveals a second online vertex $r_2$, also adjacent only to $i$, with
weight $W$. Since there is no free offline neighbor to which $r_1$ can
be moved, no length-$3$ augmenting path is available. The algorithm
keeps value $1$, while the optimum has value $W$. Letting $W$ tend to
infinity gives competitive ratio zero.

Thus, with weights on online vertices, an algorithm would need
replacement operations that do not increase the cardinality of the
matching. Such operations are outside $\omp_3$.

\medskip
\paragraph*{Acknowledgements} This work was partially supported by ANID (Agencia Nacional de Investigaci\'on y Desarrollo, Chile) through FONDECYT Grant No.~1231669 and the BASAL Center for Mathematical Modeling (FB210005).
\bibliographystyle{plain}
\bibliography{bibliography}

\appendix
\section{Additional examples and a special case}\label{app:bad-examples}

This appendix contains auxiliary material used to separate the roles of the
assumptions in the main results. The first three subsections give small,
explicit instances cited in Section~\ref{sec:intro}. The last subsection gives
a self-contained proof of the case $t=1$ of Theorem~\ref{thm:lcp-ratio}.

\noindent Specifically:
\begin{enumerate}
    \item The direct-first property by itself does not yield the competitive-ratio
    guarantee of \lcp: the tie-breaking among feasible length-$3$ augmentations
    matters (Subsection~\ref{subsec:app-direct-first}).

    \item The ``collapse'' of the offline budget is a phenomenon tied to
    length-$3$ augmentations; it can fail when length-$5$ augmentations are
    allowed (Subsection~\ref{subsec:app-omp5}).

    \item Prioritizing length-$1$ augmentations is necessary: without it, even
    natural deterministic rules can have strictly worse performance
    (Subsection~\ref{subsec:app-prioritize-direct}).

    \item The case $t=1$ of Theorem~\ref{thm:lcp-ratio} admits a short
    self-contained proof (Subsection~\ref{subsec:lcp-one-one}).
\end{enumerate}

\subsection{Insufficiency of the direct-first property}\label{subsec:app-direct-first}

This example shows that the direct-first property alone does not imply the \lcp guarantee.

\begin{proposition}\label{prop:max-type-bad}
There is a direct-first deterministic algorithm for $\omp_3(1,2)$ whose
competitive ratio is at most $5/8$. In particular, the direct-first property
alone does not imply the guarantee $\gamma_2=7/11$ achieved by \lcp.
\end{proposition}

\begin{proof}
Consider the following direct-first rule: if a length-$1$ augmentation is
available, the algorithm takes one; otherwise, it chooses a feasible length-$3$
augmentation whose middle matched edge lies in a component of maximum type
$T^q$ in the current history graph.

We give an adaptive instance where all unlisted neighborhoods are empty. Let the
offline vertices be $a_0,a_1$, $b_0,b_1,b_2,b_3$, and $c_0,c_1,c_2,c_3$.
The first three online vertices arrive with $N(r_1)=\{a_0,a_1\}$,
$N(r_2)=\{b_0,b_1,b_2,b_3\}$, and $N(r_3)=\{c_0,c_1,c_2,c_3\}$. Without loss of
generality, assume that the algorithm matches $r_1$ to $a_0$, $r_2$ to $b_0$,
and $r_3$ to $c_0$.

Vertex $r_4$ arrives with $N(r_4)=\{b_0,c_0\}$. There are two symmetric length-$3$
augmentations. Without loss of generality, assume that the algorithm uses
$\apath{r_4}{c_0}{r_3}{c_1}$. The current matching is $\{r_1a_0, r_2b_0, r_4c_0, r_3c_1\}$.
Thus the $c$-component has type $T^1$.

Vertex $r_5$ arrives with $N(r_5)=\{a_0,c_1\}$. There are two feasible length-$3$
augmentations: $\apath{r_5}{a_0}{r_1}{a_1}$ and $\apath{r_5}{c_1}{r_3}{c_2}$. The first one uses a component
of type $T^0$, while the second one uses the $c$-component, which has type $T^1$.
Therefore the algorithm chooses the second path. The matching becomes
$M=\{r_1a_0, r_2b_0, r_4c_0, r_5c_1, r_3c_2\}$. At this point the online vertex $r_3$ has
been reassigned twice, so it has exhausted its budget $t=2$.

Finally, vertices $r_6,r_7,r_8$ arrive with $N(r_6)=\{c_0\}$, $N(r_7)=\{c_1\}$, and
$N(r_8)=\{c_2\}$. Vertices $r_6$ and $r_7$ have no length-$1$ or length-$3$ augmenting
path. For vertex $r_8$, the only possible length-$3$ augmentation is $\apath{r_8}{c_2}{r_3}{c_3}$,
but this path is infeasible because it would reassign the online vertex $r_3$ for
a third time. Hence the algorithm finishes with matching size $5$.

On the same instance, an optimum matching consists of the edges $r_1a_1$, $r_2b_1$,
$r_3c_3$, $r_4b_0$, $r_5a_0$, $r_6c_0$, $r_7c_1$, and $r_8c_2$.
Thus the ratio of the algorithm on this instance is $5/8<7/11$.
\end{proof}

\begin{figure}[ht]
\centering
\scalebox{0.58}{
\begin{minipage}{0.49\textwidth}
\centering
\begin{tikzpicture}
    \Vertex[x=0,y=0, size=.5, color=black, style=solid, label=$a_0$, position=left]{A-p1a0}
    \Vertex[x=0,y=-0.8, size=.5, color=white, style=dashed, label=$a_1$, position=left]{A-p1a1}

    \Vertex[x=0,y=-1.6, size=.5, color=black, style=solid, label=$b_0$, position=left]{A-p1b0}
    \Vertex[x=0,y=-2.4, size=.5, color=white, style=dashed, label=$b_1$, position=left]{A-p1b1}
    \Vertex[x=0,y=-3.2, size=.5, color=white, style=dashed, label=$b_2$, position=left]{A-p1b2}
    \Vertex[x=0,y=-4.0, size=.5, color=white, style=dashed, label=$b_3$, position=left]{A-p1b3}

    \Vertex[x=0,y=-4.8, size=.5, color=black, style=solid, label=$c_0$, position=left]{A-p1c0}
    \Vertex[x=0,y=-5.6, size=.5, color=black, style=solid, label=$c_1$, position=left]{A-p1c1}
    \Vertex[x=0,y=-6.4, size=.5, color=black, style=solid, label=$c_2$, position=left]{A-p1c2}
    \Vertex[x=0,y=-7.2, size=.5, color=white, style=dashed, label=$c_3$, position=left]{A-p1c3}

    \Vertex[x=3,y=0, size=.5, color=black, style=solid, label=$r_1$, position=right]{A-p1u1}
    \Vertex[x=3,y=-1.6, size=.5, color=black, style=solid, label=$r_2$, position=right]{A-p1u2}
    \Vertex[x=3,y=-4.8, size=.5, color=black, style=solid, label=$r_3$, position=right]{A-p1u3}
    \Vertex[x=3,y=-5.6, size=.5, color=black, style=solid, label=$r_4$, position=right]{A-p1u4}
    \Vertex[x=3,y=-6.4, size=.5, color=black, style=solid, label=$r_5$, position=right]{A-p1u5}
    \Vertex[x=3,y=-7.2, size=.5, color=white, style=dashed, label=$r_6$, position=right]{A-p1u6}
    \Vertex[x=3,y=-8.0, size=.5, color=white, style=dashed, label=$r_7$, position=right]{A-p1u7}
    \Vertex[x=3,y=-8.8, size=.5, color=white, style=dashed, label=$r_8$, position=right]{A-p1u8}

    \Edge[style=dashed](A-p1a1)(A-p1u1)

    \Edge[style=dashed](A-p1b1)(A-p1u2)
    \Edge[style=dashed](A-p1b2)(A-p1u2)
    \Edge[style=dashed](A-p1b3)(A-p1u2)

    \Edge[style=dashed](A-p1c3)(A-p1u3)

    \Edge[style=dashed](A-p1b0)(A-p1u4)
    \Edge[style=dashed](A-p1a0)(A-p1u5)

    \Edge[style=dashed](A-p1c0)(A-p1u6)
    \Edge[style=dashed](A-p1c1)(A-p1u7)
    \Edge[style=dashed](A-p1c2)(A-p1u8)

    \Edge[style=solid](A-p1c0)(A-p1u3)
    \Edge[style=solid](A-p1c1)(A-p1u3)

    \Edge[style={line width=2pt}](A-p1a0)(A-p1u1)
    \Edge[style={line width=2pt}](A-p1b0)(A-p1u2)
    \Edge[style={line width=2pt}](A-p1c2)(A-p1u3)
    \Edge[style={line width=2pt}](A-p1c0)(A-p1u4)
    \Edge[style={line width=2pt}](A-p1c1)(A-p1u5)
\end{tikzpicture}

{\small Algorithm's final matching}
\end{minipage}
\hfill
\begin{minipage}{0.49\textwidth}
\centering
\begin{tikzpicture}
    \Vertex[x=0,y=0, size=.5, color=black, style=solid, label=$a_0$, position=left]{B-p1a0}
    \Vertex[x=0,y=-0.8, size=.5, color=black, style=solid, label=$a_1$, position=left]{B-p1a1}

    \Vertex[x=0,y=-1.6, size=.5, color=black, style=solid, label=$b_0$, position=left]{B-p1b0}
    \Vertex[x=0,y=-2.4, size=.5, color=black, style=solid, label=$b_1$, position=left]{B-p1b1}
    \Vertex[x=0,y=-3.2, size=.5, color=white, style=dashed, label=$b_2$, position=left]{B-p1b2}
    \Vertex[x=0,y=-4.0, size=.5, color=white, style=dashed, label=$b_3$, position=left]{B-p1b3}

    \Vertex[x=0,y=-4.8, size=.5, color=black, style=solid, label=$c_0$, position=left]{B-p1c0}
    \Vertex[x=0,y=-5.6, size=.5, color=black, style=solid, label=$c_1$, position=left]{B-p1c1}
    \Vertex[x=0,y=-6.4, size=.5, color=black, style=solid, label=$c_2$, position=left]{B-p1c2}
    \Vertex[x=0,y=-7.2, size=.5, color=black, style=solid, label=$c_3$, position=left]{B-p1c3}

    \Vertex[x=3,y=0, size=.5, color=black, style=solid, label=$r_1$, position=right]{B-p1u1}
    \Vertex[x=3,y=-1.6, size=.5, color=black, style=solid, label=$r_2$, position=right]{B-p1u2}
    \Vertex[x=3,y=-4.8, size=.5, color=black, style=solid, label=$r_3$, position=right]{B-p1u3}
    \Vertex[x=3,y=-5.6, size=.5, color=black, style=solid, label=$r_4$, position=right]{B-p1u4}
    \Vertex[x=3,y=-6.4, size=.5, color=black, style=solid, label=$r_5$, position=right]{B-p1u5}
    \Vertex[x=3,y=-7.2, size=.5, color=black, style=solid, label=$r_6$, position=right]{B-p1u6}
    \Vertex[x=3,y=-8.0, size=.5, color=black, style=solid, label=$r_7$, position=right]{B-p1u7}
    \Vertex[x=3,y=-8.8, size=.5, color=black, style=solid, label=$r_8$, position=right]{B-p1u8}

    \Edge[style=dashed](B-p1a0)(B-p1u1)

    \Edge[style=dashed](B-p1b0)(B-p1u2)
    \Edge[style=dashed](B-p1b2)(B-p1u2)
    \Edge[style=dashed](B-p1b3)(B-p1u2)

    \Edge[style=dashed](B-p1c0)(B-p1u3)
    \Edge[style=dashed](B-p1c1)(B-p1u3)
    \Edge[style=dashed](B-p1c2)(B-p1u3)

    \Edge[style=dashed](B-p1c0)(B-p1u4)
    \Edge[style=dashed](B-p1a0)(B-p1u5)

    \Edge[style={line width=2pt}](B-p1a1)(B-p1u1)
    \Edge[style={line width=2pt}](B-p1b1)(B-p1u2)
    \Edge[style={line width=2pt}](B-p1c3)(B-p1u3)
    \Edge[style={line width=2pt}](B-p1b0)(B-p1u4)
    \Edge[style={line width=2pt}](B-p1a0)(B-p1u5)
    \Edge[style={line width=2pt}](B-p1c0)(B-p1u6)
    \Edge[style={line width=2pt}](B-p1c1)(B-p1u7)
    \Edge[style={line width=2pt}](B-p1c2)(B-p1u8)
\end{tikzpicture}

{\small An optimum matching}
\end{minipage}
}
\caption{The final graph in Proposition~\ref{prop:max-type-bad}. In the left
panel, thick solid edges are the algorithm's final matching, thin solid edges
entered the history graph but are not in the final matching, and dashed edges
were not used by the algorithm. In the right panel, thick solid edges form an
optimum matching, and all other displayed edges are dashed.}
\label{fig:max-type-bad}
\end{figure}
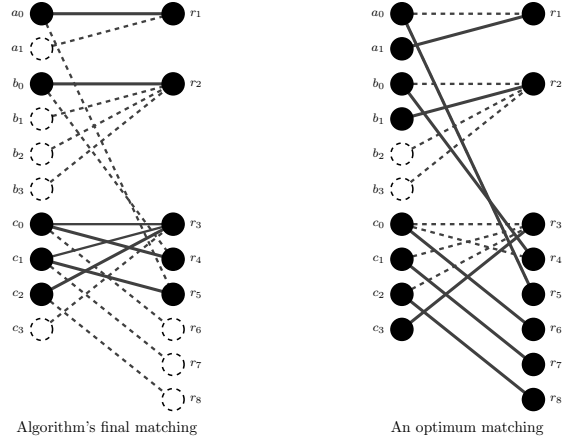

\subsection{Limits of length-\texorpdfstring{$5$}{5} augmentations}\label{subsec:app-omp5}

This example shows that allowing length-$5$ augmentations does not yield, for
$\omp_5(1,\infty)$, the same competitive ratio as the unbudgeted $3/4$
guarantee for $\omp_5=\omp_5(\infty,\infty)$.

\begin{proposition}\label{prop:omp5-unit-offline-bad}
No deterministic algorithm for $\omp_5(1,\infty)$ can guarantee a competitive
ratio larger than $5/7$.
\end{proposition}

\begin{proof}
We give an adaptive instance. Let the offline vertices be $a_0,a_1,a_2$,
$b_0,b_1,b_2$, and $c_0,c_1$. The first two online vertices arrive with
$N(r_1)=\{a_0,a_1,a_2\}$ and $N(r_2)=\{b_0,b_1,b_2\}$. If the algorithm leaves
one of them unmatched, the adversary stops and the ratio is at most $1/2$.
Thus, after relabeling, we may assume that the algorithm matches $r_1$ to
$a_0$ and $r_2$ to $b_0$.

Vertex $r_3$ arrives with $N(r_3)=\{a_0,b_0\}$. If the algorithm does not
augment, then it has matching size $2$, while the optimum has size $3$, and the
ratio is already below $5/7$. Hence we may assume that the algorithm augments.
By symmetry, assume that it uses $\apath{r_3}{a_0}{r_1}{a_1}$. The matching is
now $\{r_3a_0,r_1a_1,r_2b_0\}$, and $a_0$ has exhausted its offline budget.

Vertex $r_4$ arrives with $N(r_4)=\{c_0,c_1\}$.

\textbf{Case 1: the algorithm leaves $r_4$ unmatched.}
The adversary reveals $N(r_5)=\{a_1,c_0\}$.

\textbf{Case 1a: the algorithm leaves $r_5$ unmatched.}
The adversary reveals $N(r_6)=\{a_0\}$ and $N(r_7)=\{a_1\}$.
Vertex $r_6$ is blocked by the exhausted vertex $a_0$. Vertex $r_7$ can increase
the matching size by at most one, through $\apath{r_7}{a_1}{r_1}{a_2}$. Thus
the algorithm finishes with matching size at most $4$.

\textbf{Case 1b: the algorithm matches $r_5$ directly to $c_0$.}
The adversary reveals $N(r_6)=\{a_0\}$ and $N(r_7)=\{a_1\}$. Vertex $r_6$ is
blocked by $a_0$, and $r_7$ can increase the matching size by at most one.
Thus the algorithm finishes with matching size at most $5$.

\textbf{Case 1c: the algorithm matches $r_5$ with a 3-path.}
Suppose the algorithm uses $\apath{r_5}{a_1}{r_1}{a_2}$. Then $a_1$ exhausts its
offline budget. The adversary reveals $N(r_6)=\{a_0\}$ and $N(r_7)=\{a_1\}$.
Vertex $r_6$ is blocked by $a_0$, and vertex $r_7$ is blocked by $a_1$. Thus
the algorithm finishes with matching size at most $4$.

In all subcases of Case 1, the final graph has a matching of size $7$, namely
\[
    \{r_6a_0,r_7a_1,r_1a_2,r_3b_0,r_2b_1,r_5c_0,r_4c_1\}.
\]
Thus Case 1 gives final ratio at most $5/7$.

\textbf{Case 2: the algorithm matches $r_4$.}
By symmetry, assume that it matches $r_4$ to $c_0$. Vertex $r_5$ arrives with
$N(r_5)=\{a_1,c_0\}$.

\textbf{Case 2a: the algorithm augments through $a_1$.}
Suppose the algorithm uses $\apath{r_5}{a_1}{r_1}{a_2}$. Then $a_1$ exhausts its
offline budget. The adversary reveals $N(r_6)=\{a_0\}$ and $N(r_7)=\{a_1\}$.
The path $\apath{r_6}{a_0}{r_3}{b_0}{r_2}{b_1}$ has length $5$, but is
infeasible because it would reassign $a_0$ again. The path
$\apath{r_7}{a_1}{r_5}{c_0}{r_4}{c_1}$ has length $5$, but is infeasible because
it would reassign $a_1$ again. Thus the algorithm finishes with matching size
at most $5$.

\textbf{Case 2b: the algorithm augments through $c_0$.}
Suppose the algorithm uses $\apath{r_5}{c_0}{r_4}{c_1}$. Then $c_0$ exhausts its
offline budget. The adversary reveals $N(r_6)=\{a_0\}$ and $N(r_7)=\{c_0\}$.
The path $\apath{r_6}{a_0}{r_3}{b_0}{r_2}{b_1}$ has length $5$, but is
infeasible because it would reassign $a_0$ again. The path
$\apath{r_7}{c_0}{r_5}{a_1}{r_1}{a_2}$ has length $5$, but is infeasible because
it would reassign $c_0$ again. Thus the algorithm finishes with matching size
at most $5$.

\textbf{Case 2c: the algorithm does not augment when $r_5$ arrives.}
The adversary reveals $N(r_6)=\{a_0\}$ and $N(r_7)=\{a_1\}$. Vertex $r_6$ is
blocked by $a_0$. Vertex $r_7$ can increase the matching size by at most one,
through $\apath{r_7}{a_1}{r_1}{a_2}$. Thus the algorithm finishes with matching
size at most $5$.

In all subcases of Case 2, the final graph has a matching of size $7$. In Cases
2a and 2c, one such matching consists of the edges $r_6a_0$, $r_7a_1$,
$r_1a_2$, $r_3b_0$, $r_2b_1$, $r_5c_0$, and $r_4c_1$. In Case 2b, one such
matching consists of the edges $r_6a_0$, $r_7c_0$, $r_5a_1$, $r_1a_2$,
$r_3b_0$, $r_2b_1$, and $r_4c_1$. Thus Case 2 also gives final ratio at most
$5/7$.

This proves the claimed upper bound.
\end{proof}

\begin{figure}[ht]
\centering
\scalebox{0.68}{
\begin{tabular}{cccc}
\begin{minipage}{0.24\textwidth}
\centering
\begin{tikzpicture}
    \Vertex[x=0,y=0, size=.5, color=black, style=solid, label=$a_0$, position=left]{A1a0}
    \Vertex[x=0,y=-1.0, size=.5, color=black, style=solid, label=$a_1$, position=left]{A1a1}
    \Vertex[x=0,y=-2.0, size=.5, color=black, style=solid, label=$a_2$, position=left]{A1a2}
    \Vertex[x=0,y=-3.0, size=.5, color=black, style=solid, label=$b_0$, position=left]{A1b0}
    \Vertex[x=0,y=-4.0, size=.5, color=white, style=dashed, label=$b_1$, position=left]{A1b1}
    \Vertex[x=0,y=-5.0, size=.5, color=white, style=dashed, label=$b_2$, position=left]{A1b2}
    \Vertex[x=0,y=-6.0, size=.5, color=black, style=solid, label=$c_0$, position=left]{A1c0}
    \Vertex[x=0,y=-7.0, size=.5, color=white, style=dashed, label=$c_1$, position=left]{A1c1}

    \Vertex[x=3,y=0, size=.5, color=black, style=solid, label=$r_1$, position=right]{A1r1}
    \Vertex[x=3,y=-1.0, size=.5, color=black, style=solid, label=$r_2$, position=right]{A1r2}
    \Vertex[x=3,y=-2.0, size=.5, color=black, style=solid, label=$r_3$, position=right]{A1r3}
    \Vertex[x=3,y=-3.0, size=.5, color=white, style=dashed, label=$r_4$, position=right]{A1r4}
    \Vertex[x=3,y=-4.0, size=.5, color=black, style=solid, label=$r_5$, position=right]{A1r5}
    \Vertex[x=3,y=-5.0, size=.5, color=white, style=dashed, label=$r_6$, position=right]{A1r6}
    \Vertex[x=3,y=-6.0, size=.5, color=black, style=solid, label=$r_7$, position=right]{A1r7}

    \Edge[style=dashed](A1b1)(A1r2)
    \Edge[style=dashed](A1b2)(A1r2)
    \Edge[style=dashed](A1b0)(A1r3)
    \Edge[style=dashed](A1c0)(A1r4)
    \Edge[style=dashed](A1c1)(A1r4)
    \Edge[style=dashed](A1a1)(A1r5)
    \Edge[style=dashed](A1a0)(A1r6)

    \Edge[style=solid](A1a0)(A1r1)
    \Edge[style=solid](A1a1)(A1r1)

    \Edge[style={line width=2pt}](A1a2)(A1r1)
    \Edge[style={line width=2pt}](A1b0)(A1r2)
    \Edge[style={line width=2pt}](A1a0)(A1r3)
    \Edge[style={line width=2pt}](A1c0)(A1r5)
    \Edge[style={line width=2pt}](A1a1)(A1r7)
\end{tikzpicture}

{\small Case 1: algorithm}
\end{minipage}
&
\quad \begin{minipage}{0.24\textwidth}
\centering
\begin{tikzpicture}
    \Vertex[x=0,y=0, size=.5, color=black, style=solid, label=$a_0$, position=left]{O1a0}
    \Vertex[x=0,y=-1.0, size=.5, color=black, style=solid, label=$a_1$, position=left]{O1a1}
    \Vertex[x=0,y=-2.0, size=.5, color=black, style=solid, label=$a_2$, position=left]{O1a2}
    \Vertex[x=0,y=-3.0, size=.5, color=black, style=solid, label=$b_0$, position=left]{O1b0}
    \Vertex[x=0,y=-4.0, size=.5, color=black, style=solid, label=$b_1$, position=left]{O1b1}
    \Vertex[x=0,y=-5.0, size=.5, color=white, style=dashed, label=$b_2$, position=left]{O1b2}
    \Vertex[x=0,y=-6.0, size=.5, color=black, style=solid, label=$c_0$, position=left]{O1c0}
    \Vertex[x=0,y=-7.0, size=.5, color=black, style=solid, label=$c_1$, position=left]{O1c1}

    \Vertex[x=3,y=0, size=.5, color=black, style=solid, label=$r_1$, position=right]{O1r1}
    \Vertex[x=3,y=-1.0, size=.5, color=black, style=solid, label=$r_2$, position=right]{O1r2}
    \Vertex[x=3,y=-2.0, size=.5, color=black, style=solid, label=$r_3$, position=right]{O1r3}
    \Vertex[x=3,y=-3.0, size=.5, color=black, style=solid, label=$r_4$, position=right]{O1r4}
    \Vertex[x=3,y=-4.0, size=.5, color=black, style=solid, label=$r_5$, position=right]{O1r5}
    \Vertex[x=3,y=-5.0, size=.5, color=black, style=solid, label=$r_6$, position=right]{O1r6}
    \Vertex[x=3,y=-6.0, size=.5, color=black, style=solid, label=$r_7$, position=right]{O1r7}

    \Edge[style=dashed](O1a0)(O1r1)
    \Edge[style=dashed](O1a1)(O1r1)
    \Edge[style=dashed](O1b0)(O1r2)
    \Edge[style=dashed](O1b2)(O1r2)
    \Edge[style=dashed](O1a0)(O1r3)
    \Edge[style=dashed](O1c0)(O1r4)
    \Edge[style=dashed](O1a1)(O1r5)

    \Edge[style={line width=2pt}](O1a2)(O1r1)
    \Edge[style={line width=2pt}](O1b1)(O1r2)
    \Edge[style={line width=2pt}](O1b0)(O1r3)
    \Edge[style={line width=2pt}](O1c1)(O1r4)
    \Edge[style={line width=2pt}](O1c0)(O1r5)
    \Edge[style={line width=2pt}](O1a0)(O1r6)
    \Edge[style={line width=2pt}](O1a1)(O1r7)
\end{tikzpicture}

{\small Case 1: optimum}
\end{minipage}
&\qquad \quad 
\begin{minipage}{0.24\textwidth}
\centering
\begin{tikzpicture}
    \Vertex[x=0,y=0, size=.5, color=black, style=solid, label=$a_0$, position=left]{A2a0}
    \Vertex[x=0,y=-1.0, size=.5, color=black, style=solid, label=$a_1$, position=left]{A2a1}
    \Vertex[x=0,y=-2.0, size=.5, color=white, style=dashed, label=$a_2$, position=left]{A2a2}
    \Vertex[x=0,y=-3.0, size=.5, color=black, style=solid, label=$b_0$, position=left]{A2b0}
    \Vertex[x=0,y=-4.0, size=.5, color=white, style=dashed, label=$b_1$, position=left]{A2b1}
    \Vertex[x=0,y=-5.0, size=.5, color=white, style=dashed, label=$b_2$, position=left]{A2b2}
    \Vertex[x=0,y=-6.0, size=.5, color=black, style=solid, label=$c_0$, position=left]{A2c0}
    \Vertex[x=0,y=-7.0, size=.5, color=black, style=solid, label=$c_1$, position=left]{A2c1}

    \Vertex[x=3,y=0, size=.5, color=black, style=solid, label=$r_1$, position=right]{A2r1}
    \Vertex[x=3,y=-1.0, size=.5, color=black, style=solid, label=$r_2$, position=right]{A2r2}
    \Vertex[x=3,y=-2.0, size=.5, color=black, style=solid, label=$r_3$, position=right]{A2r3}
    \Vertex[x=3,y=-3.0, size=.5, color=black, style=solid, label=$r_4$, position=right]{A2r4}
    \Vertex[x=3,y=-4.0, size=.5, color=black, style=solid, label=$r_5$, position=right]{A2r5}
    \Vertex[x=3,y=-5.0, size=.5, color=white, style=dashed, label=$r_6$, position=right]{A2r6}
    \Vertex[x=3,y=-6.0, size=.5, color=white, style=dashed, label=$r_7$, position=right]{A2r7}

    \Edge[style=dashed](A2a2)(A2r1)
    \Edge[style=dashed](A2b1)(A2r2)
    \Edge[style=dashed](A2b2)(A2r2)
    \Edge[style=dashed](A2b0)(A2r3)
    \Edge[style=dashed](A2c0)(A2r4)
    \Edge[style=dashed](A2a1)(A2r5)
    \Edge[style=dashed](A2a0)(A2r6)
    \Edge[style=dashed](A2c0)(A2r7)

    \Edge[style=solid](A2a0)(A2r1)
    \Edge[style=solid](A2c0)(A2r4)

    \Edge[style={line width=2pt}](A2a1)(A2r1)
    \Edge[style={line width=2pt}](A2b0)(A2r2)
    \Edge[style={line width=2pt}](A2a0)(A2r3)
    \Edge[style={line width=2pt}](A2c1)(A2r4)
    \Edge[style={line width=2pt}](A2c0)(A2r5)
\end{tikzpicture}

{\small Case 2b: algorithm}
\end{minipage}
&\quad
\begin{minipage}{0.24\textwidth}
\centering
\begin{tikzpicture}
    \Vertex[x=0,y=0, size=.5, color=black, style=solid, label=$a_0$, position=left]{O2a0}
    \Vertex[x=0,y=-1.0, size=.5, color=black, style=solid, label=$a_1$, position=left]{O2a1}
    \Vertex[x=0,y=-2.0, size=.5, color=black, style=solid, label=$a_2$, position=left]{O2a2}
    \Vertex[x=0,y=-3.0, size=.5, color=black, style=solid, label=$b_0$, position=left]{O2b0}
    \Vertex[x=0,y=-4.0, size=.5, color=black, style=solid, label=$b_1$, position=left]{O2b1}
    \Vertex[x=0,y=-5.0, size=.5, color=white, style=dashed, label=$b_2$, position=left]{O2b2}
    \Vertex[x=0,y=-6.0, size=.5, color=black, style=solid, label=$c_0$, position=left]{O2c0}
    \Vertex[x=0,y=-7.0, size=.5, color=black, style=solid, label=$c_1$, position=left]{O2c1}

    \Vertex[x=3,y=0, size=.5, color=black, style=solid, label=$r_1$, position=right]{O2r1}
    \Vertex[x=3,y=-1.0, size=.5, color=black, style=solid, label=$r_2$, position=right]{O2r2}
    \Vertex[x=3,y=-2.0, size=.5, color=black, style=solid, label=$r_3$, position=right]{O2r3}
    \Vertex[x=3,y=-3.0, size=.5, color=black, style=solid, label=$r_4$, position=right]{O2r4}
    \Vertex[x=3,y=-4.0, size=.5, color=black, style=solid, label=$r_5$, position=right]{O2r5}
    \Vertex[x=3,y=-5.0, size=.5, color=black, style=solid, label=$r_6$, position=right]{O2r6}
    \Vertex[x=3,y=-6.0, size=.5, color=black, style=solid, label=$r_7$, position=right]{O2r7}

    \Edge[style=dashed](O2a0)(O2r1)
    \Edge[style=dashed](O2a1)(O2r1)
    \Edge[style=dashed](O2b0)(O2r2)
    \Edge[style=dashed](O2b2)(O2r2)
    \Edge[style=dashed](O2a0)(O2r3)
    \Edge[style=dashed](O2c0)(O2r4)
    \Edge[style=dashed](O2c0)(O2r5)

    \Edge[style={line width=2pt}](O2a2)(O2r1)
    \Edge[style={line width=2pt}](O2b1)(O2r2)
    \Edge[style={line width=2pt}](O2b0)(O2r3)
    \Edge[style={line width=2pt}](O2c1)(O2r4)
    \Edge[style={line width=2pt}](O2a1)(O2r5)
    \Edge[style={line width=2pt}](O2a0)(O2r6)
    \Edge[style={line width=2pt}](O2c0)(O2r7)
\end{tikzpicture}

{\small Case 2b: optimum}
\end{minipage}
\end{tabular}
}
\caption{Two representative branches in Proposition~\ref{prop:omp5-unit-offline-bad}.
Each pair shows the algorithm's final matching and an optimum matching of size
$7$.}
\label{fig:omp5-unit-offline-bad}
\end{figure}
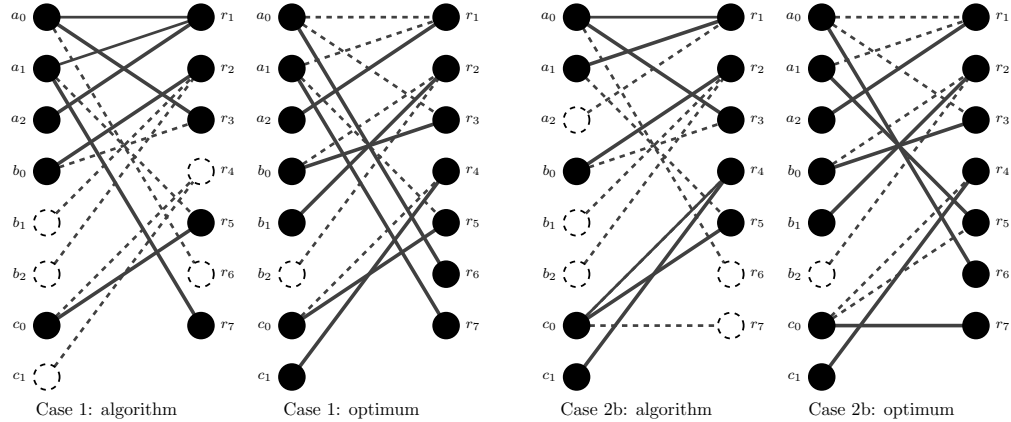 
\subsection{Necessity of prioritizing direct matches}\label{subsec:app-prioritize-direct}

This example shows that prioritizing direct matches is necessary in our discussion of variants of \spgreedy.

\begin{proposition}\label{prop:direct-first-needed}
There is an $\spgreedy$ algorithm for $\omp_3(1,1)$ that does not prioritize
length-$1$ augmentations and has competitive ratio at most $1/2$. In particular,
the direct-first rule cannot be omitted from a guarantee of $3/5$.
\end{proposition}

\begin{proof}
Consider a deterministic $\spgreedy$ algorithm whose fixed tie-breaking rule
chooses a length-$3$ augmentation whenever the instance below offers both a
length-$1$ and a length-$3$ augmentation. We show that this deterministic rule
can give ratio $1/2$.

Let the offline vertices be $a_0,a_1,a_2,c$. Vertex $r_1$ arrives with
$N(r_1)=\{a_0,a_1,a_2\}$. Without loss of generality, assume that the algorithm
matches $r_1$ to $a_0$.

Vertex $r_2$ arrives with $N(r_2)=\{a_0,c\}$. There is a direct match to the free
vertex $c$. There is also a length-$3$ augmentation $\apath{r_2}{a_0}{r_1}{a_1}$.
By the chosen tie-breaking rule, the algorithm uses this length-$3$ augmentation. The matching becomes
$\{r_2a_0, r_1a_1\}$. Now $a_0$ has exhausted its offline budget, and the online
vertex $r_1$ has exhausted its online budget $t=1$.

The adversary reveals vertices $r_3$ and $r_4$ with $N(r_3)=\{a_0\}$ and $N(r_4)=\{a_1\}$.
Vertex $r_3$ cannot be matched: the path $\apath{r_3}{a_0}{r_2}{c}$ would reassign $a_0$ for a
second time. Vertex $r_4$ cannot be matched either: the path $\apath{r_4}{a_1}{r_1}{a_2}$ would
reassign the online vertex $r_1$ for a second time. Thus the algorithm finishes
with matching size $2$.

The optimum matching has size $4$: $\{r_3a_0, r_4a_1, r_1a_2, r_2c\}$.
Hence the ratio is $1/2<3/5$.
\end{proof}

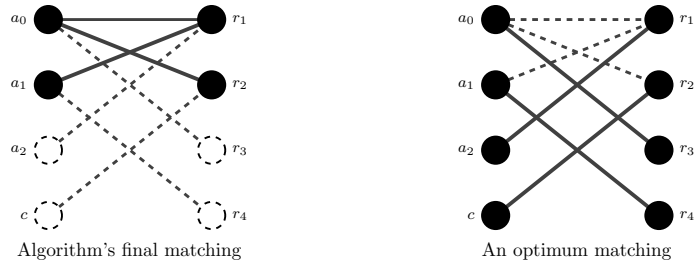
\begin{figure}[ht]
\centering
\scalebox{0.72}{
\begin{minipage}{0.49\textwidth}
\centering
\begin{tikzpicture}
    \Vertex[x=0,y=0, size=.5, color=black, style=solid, label=$a_0$, position=left]{A-p3a0}
    \Vertex[x=0,y=-1.2, size=.5, color=black, style=solid, label=$a_1$, position=left]{A-p3a1}
    \Vertex[x=0,y=-2.4, size=.5, color=white, style=dashed, label=$a_2$, position=left]{A-p3a2}
    \Vertex[x=0,y=-3.6, size=.5, color=white, style=dashed, label=$c$, position=left]{A-p3c}

    \Vertex[x=3,y=0, size=.5, color=black, style=solid, label=$r_1$, position=right]{A-p3u1}
    \Vertex[x=3,y=-1.2, size=.5, color=black, style=solid, label=$r_2$, position=right]{A-p3u2}
    \Vertex[x=3,y=-2.4, size=.5, color=white, style=dashed, label=$r_3$, position=right]{A-p3u3}
    \Vertex[x=3,y=-3.6, size=.5, color=white, style=dashed, label=$r_4$, position=right]{A-p3u4}

    \Edge[style=dashed](A-p3a2)(A-p3u1)
    \Edge[style=dashed](A-p3c)(A-p3u2)
    \Edge[style=dashed](A-p3a0)(A-p3u3)
    \Edge[style=dashed](A-p3a1)(A-p3u4)

    \Edge[style=solid](A-p3a0)(A-p3u1)

    \Edge[style={line width=2pt}](A-p3a1)(A-p3u1)
    \Edge[style={line width=2pt}](A-p3a0)(A-p3u2)
\end{tikzpicture}

{\small Algorithm's final matching}
\end{minipage}
\hfill
\begin{minipage}{0.49\textwidth}
\centering
\begin{tikzpicture}
    \Vertex[x=0,y=0, size=.5, color=black, style=solid, label=$a_0$, position=left]{B-p3a0}
    \Vertex[x=0,y=-1.2, size=.5, color=black, style=solid, label=$a_1$, position=left]{B-p3a1}
    \Vertex[x=0,y=-2.4, size=.5, color=black, style=solid, label=$a_2$, position=left]{B-p3a2}
    \Vertex[x=0,y=-3.6, size=.5, color=black, style=solid, label=$c$, position=left]{B-p3c}

    \Vertex[x=3,y=0, size=.5, color=black, style=solid, label=$r_1$, position=right]{B-p3u1}
    \Vertex[x=3,y=-1.2, size=.5, color=black, style=solid, label=$r_2$, position=right]{B-p3u2}
    \Vertex[x=3,y=-2.4, size=.5, color=black, style=solid, label=$r_3$, position=right]{B-p3u3}
    \Vertex[x=3,y=-3.6, size=.5, color=black, style=solid, label=$r_4$, position=right]{B-p3u4}

    \Edge[style=dashed](B-p3a0)(B-p3u1)
    \Edge[style=dashed](B-p3a1)(B-p3u1)
    \Edge[style=dashed](B-p3a0)(B-p3u2)

    \Edge[style={line width=2pt}](B-p3a2)(B-p3u1)
    \Edge[style={line width=2pt}](B-p3c)(B-p3u2)
    \Edge[style={line width=2pt}](B-p3a0)(B-p3u3)
    \Edge[style={line width=2pt}](B-p3a1)(B-p3u4)
\end{tikzpicture}

{\small An optimum matching}
\end{minipage}
}
\caption{The final graph in Proposition~\ref{prop:direct-first-needed}. In the
left panel, thick solid edges are the algorithm's final matching, the thin solid
edge entered the history graph but is not in the final matching, and dashed
edges were not used by the algorithm. In the right panel, thick solid edges form
an optimum matching, and all other displayed edges are dashed.}
\label{fig:direct-first-needed}
\end{figure}

\subsection{\texorpdfstring{The case $\omp_3(1,1)$}{The case OMP3(1,1)}}\label{subsec:lcp-one-one}

This subsection gives a self-contained proof of the case $t=1$ of Theorem~\ref{thm:lcp-ratio}.
The proof is not used in the main text; it is included as a concrete specialization
of the argument used to prove Theorem~\ref{thm:lcp-ratio}.

For $t=1$, the parameters used in the proof of Theorem~\ref{thm:lcp-ratio} are
\[
    \lambda=\frac23,\qquad
    a=(a_0,a_1)=\left(\frac23,1\right),
    \qquad
    b=(b_0,b_1)=\left(1,\frac13\right).
\]
The proof below uses the corresponding specialization: a direct match assigns
$(\alpha_i,\beta_j)=(a_0,b_0)$, a $3$-path update assigns
$(\alpha_i,\alpha_x,\beta_j,\beta_y)=(a_1,1,b_1,1)$, and the postprocessing
step changes $(\alpha_i,\beta_j)=(a_0,b_0)$ to $(1,a_0)$ on the affected
matched edge. We give the feasibility check directly.

When reassignment budgets are finite, the final matching produced by a \spgreedy
algorithm may still contain augmenting $3$-paths, but these paths can be
\emph{infeasible} because their middle edge is frozen (executing the
augmentation would violate a budget). Thus, unlike the case $t=\infty$, we
cannot argue dual feasibility by ruling out all augmenting $3$-paths.
Instead, we use a different dual assignment together with an additional
structural property of augmenting $3$-paths.

A key observation for the case $s=t=1$ is as follows. Suppose \lcp matches an
arriving vertex $j\in R$ by augmenting along a $3$-path $\apath{j}{x}{y}{i}$, where $xy$ is an
edge of the current matching and $i\in L$ is free. The augmentation removes $xy$
and adds $xj$ and $iy$, so both internal vertices $x$ and $y$ change partners
once. Since $s=t=1$, this uses up their entire reassignment budget: $x$ and $y$
become exhausted, and the new incident edges (namely $jx$ and $yi$) are frozen.

Moreover, since $x$ and $y$ were not exhausted before this augmentation, neither
has been reassigned earlier. We claim that in this case the edge $xy$ must have
entered the matching when $y$ first became matched. (Indeed, if $xy$ had entered
the matching via an earlier $3$-path augmentation, then either $x$ or $y$ would
already have been reassigned.)

\begin{proposition}\label{prop:lcp11}

    The \lcp algorithm is $\frac{3}{5}=0.60$-competitive for $\omp_3(1,1)$.
\end{proposition}
\begin{proof}
Fix an instance $(G,\pi)$ where $G=(L,R,E)$, and let $\mathcal A$ denote the
\lcp algorithm for $\omp_3(1,1)$. Let $M$ be the final matching produced by
$\mathcal A$. We construct, alongside the
execution of $\mathcal A$, a dual solution $(\alpha,\beta)$ for $D_G$, starting
from the all-zero solution; the construction is given in
Algorithm~\ref{alg:dual-lcp-11}.

The dual updates are designed so that each successful matching step increases
the dual objective by at most $5/3$. The final postprocessing step only
redistributes dual weight within edges of $M$ (preserving the objective value)
to guarantee feasibility. The dual updates and the postprocessing step are also
shown in Figure~\ref{fig:lcp11-updates}
(\subref{fig:lcp11-direct}--\subref{fig:lcp11-post}).

\begin{figure}[ht]
\centering

\begin{minipage}[b]{0.48\textwidth}
    \centering
    \begin{subfigure}[b]{\textwidth}
        \centering
        \scalebox{0.75}{
        \begin{tikzpicture}
            \useasboundingbox (-2, -0.4) rectangle (4, 5.9);

            \Vertex[x=0,y=4.0, style={color=white}, label={$\alpha_i=0$}, position=left]{yi0_top}
            \Vertex[x=2,y=4.0, style={color=white}, label={$\beta_j=0$}, position=right]{yj0_top}
            \Vertex[x=0,y=5.5, style={color=white}, label={$\alpha_x=\frac23$}, position=left]{yx0_top}
            \Vertex[x=2,y=5.5, style={color=white}, label={$\beta_y=1$}, position=right]{yy0_top}

            \Vertex[x=0,y=4.0, size=.5, color=white, style=solid, label=$i$, position=above]{i0_top}
            \Vertex[x=2,y=4.0, size=.5, color=white, style=solid, label=$j$, position=above]{j0_top}
            \Vertex[x=0,y=5.5, size=.5, color=black, style=solid, label=$x$, position=above]{x0_top}
            \Vertex[x=2,y=5.5, size=.5, color=black, style=solid, label=$y$, position=above]{y0_top}

            \Edge[style=solid](x0_top)(y0_top)
            \Edge[style=dashed](x0_top)(j0_top)
            \Edge[style=dashed](y0_top)(i0_top)

            \draw[->, >=stealth, line width=1.5pt] (1, 3.5) -- (1, 2.0);

            \Vertex[x=0,y=0, style={color=white}, label={$\alpha_i=1$}, position=left]{yi1_bot}
            \Vertex[x=2,y=0, style={color=white}, label={$\beta_j=\frac13$}, position=right]{yj1_bot}
            \Vertex[x=0,y=1.5, style={color=white}, label={$\alpha_x=1$}, position=left]{yx1_bot}
            \Vertex[x=2,y=1.5, style={color=white}, label={$\beta_y=1$}, position=right]{yy1_bot}

            \Vertex[x=0,y=0, size=.5, color=black, style=solid, label=$i$, position=above]{i1_bot}
            \Vertex[x=2,y=0, size=.5, color=black, style=solid, label=$j$, position=above]{j1_bot}
            \Vertex[x=0,y=1.5, size=.5, color=black, style=solid, label=$x$, position=above]{x1_bot}
            \Vertex[x=2,y=1.5, size=.5, color=black, style=solid, label=$y$, position=above]{y1_bot}

            \Edge[style=solid](y1_bot)(i1_bot)
            \Edge[style=solid](x1_bot)(j1_bot)
            \Edge[style=dashed](x1_bot)(y1_bot)

        \end{tikzpicture}
        }
        \caption{$3$-path update: $j$ arrives and is matched by augmenting along $\apath{j}{x}{y}{i}$.}
        \label{fig:lcp11-aug3}
    \end{subfigure}
\end{minipage}%
\hfill
\begin{minipage}[b]{0.48\textwidth}
    \centering
    \begin{subfigure}[b]{\textwidth}
        \centering
        \scalebox{0.75}{
        \begin{tikzpicture}
            \useasboundingbox (-2, -0.4) rectangle (4, 2.4);

            \Vertex[x=0, y=2.0, style={color=white}, label={$\alpha_i=0$}, position=left]{yi0_top}
            \Vertex[x=2, y=2.0, style={color=white}, label={$\beta_j=0$}, position=right]{yj0_top}
            \Vertex[x=0, y=2.0, size=.5, color=white, style=solid, label=$i$, position=above]{i0_top}
            \Vertex[x=2, y=2.0, size=.5, color=white, style=solid, label=$j$, position=above]{j0_top}
            \Edge[style=dashed](i0_top)(j0_top)

            \draw[->, >=stealth, line width=1.5pt] (1, 1.6) -- (1, 0.4);

            \Vertex[x=0, y=0, style={color=white}, label={$\alpha_i=\frac23$}, position=left]{yi1_bot}
            \Vertex[x=2, y=0, style={color=white}, label={$\beta_j=1$}, position=right]{yj1_bot}
            \Vertex[x=0, y=0, size=.5, color=black, style=solid, label=$i$, position=above]{i1_bot}
            \Vertex[x=2, y=0, size=.5, color=black, style=solid, label=$j$, position=above]{j1_bot}
            \Edge[style=solid](i1_bot)(j1_bot)
        \end{tikzpicture}
        }
        \caption{Direct match update: $j$ arrives and is matched directly to $i$.}
        \label{fig:lcp11-direct}
    \end{subfigure}

    \vspace{0.5cm}

    \begin{subfigure}[b]{\textwidth}
        \centering
        \scalebox{0.75}{
        \begin{tikzpicture}
            \useasboundingbox (-2, -0.4) rectangle (4, 2.4);

            \Vertex[x=0,y=2.0, style={color=white}, label={$\alpha_i=\frac23$}, position=left]{yi0_top}
            \Vertex[x=2,y=2.0, style={color=white}, label={$\beta_j=1$}, position=right]{yj0_top}
            \Vertex[x=0,y=2.0, size=.5, color=black, style=solid, label=$i$, position=above]{i0_top}
            \Vertex[x=2,y=2.0, size=.5, color=black, style=solid, label=$j$, position=above]{j0_top}
            \Edge[style=solid](i0_top)(j0_top)

            \draw[->, >=stealth, line width=1.5pt] (1, 1.6) -- (1, 0.4);

            \Vertex[x=0,y=0, style={color=white}, label={$\alpha_i=1$}, position=left]{yi1_bot}
            \Vertex[x=2,y=0, style={color=white}, label={$\beta_j=\frac23$}, position=right]{yj1_bot}
            \Vertex[x=0,y=0, size=.5, color=black, style=solid, label=$i$, position=above]{i1_bot}
            \Vertex[x=2,y=0, size=.5, color=black, style=solid, label=$j$, position=above]{j1_bot}
            \Edge[style=solid](i1_bot)(j1_bot)
        \end{tikzpicture}
        }
        \caption{Postprocessing update on a matched edge $ij\in M$: $(\alpha_i,\beta_j)$ is changed from $(\frac23,1)$ to $(1,\frac23)$.}
        \label{fig:lcp11-post}
    \end{subfigure}
\end{minipage}

\caption{Dual updates used in Algorithm~\ref{alg:dual-lcp-11}.}
\label{fig:lcp11-updates}
\end{figure}
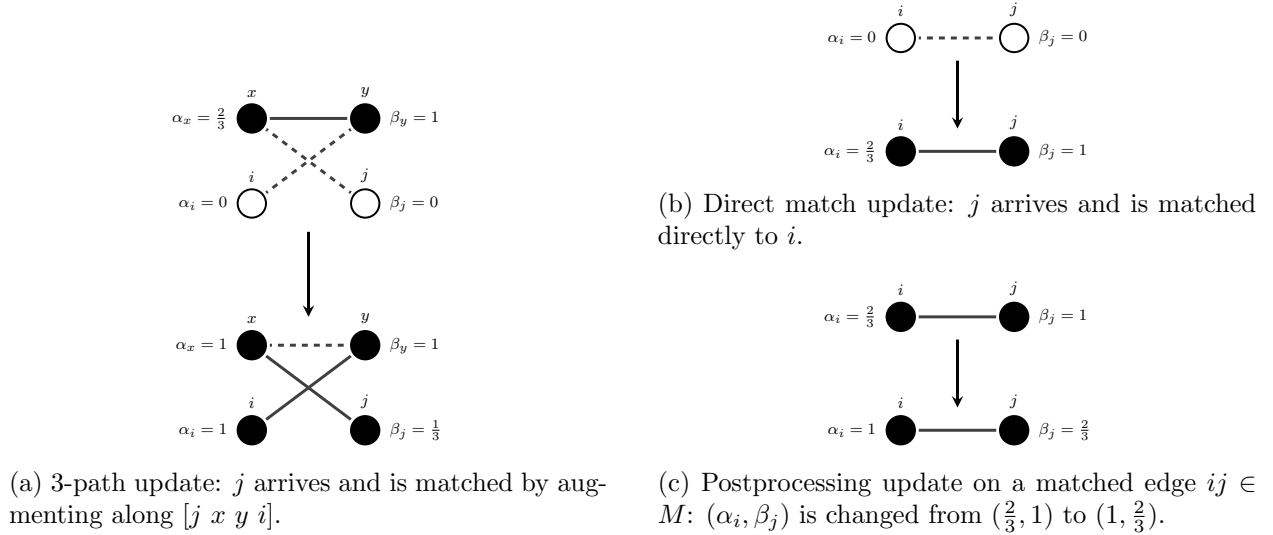

 \begin{algorithm}[ht!]
\SetAlgoLined
\DontPrintSemicolon
\KwIn{Instance $(G,\pi)$, algorithm $\mathcal A$ (\lcp for $\omp_3(1,1)$)}
\KwOut{Dual vector $(\alpha_i,\beta_j)_{i\in L,\,j\in R}$}
\lForAll{$i\in L, j\in R$}{Set $\alpha_i\gets 0$, $\beta_j\gets 0$}
\ForAll(\tcp*[f]{dual update steps}){$j\in R$ in arrival order, while running $\mathcal A$}{
  \lIf{$\mathcal A$ matches $j$ directly to some $i\in L$}{
    Set $\alpha_i\gets \frac23$ and $\beta_j\gets 1$
  }
  \lElseIf{$\mathcal A$ matches $j$ using a feasible $3$-path $\apath{j}{x}{y}{i}$}{
    Set $\alpha_i\gets 1$, $\alpha_x\gets 1$, $\beta_j\gets \frac13$ and $\beta_y\gets 1$
  }
}
Let $M$ be the final matching constructed by $\mathcal A$\\
\ForAll(\tcp*[f]{postprocessing step}){$ij\in M$ with $j\in R$}{
  \If{neither $j$ nor $i$ is exhausted, and $j$ has no free neighbor in $L$ at the end}{
    Set $\alpha_i\gets 1$ and $\beta_j\gets \frac23$\;
  }
}
\caption{$\dual$ for $\omp_3(1,1)$}
\label{alg:dual-lcp-11}
\end{algorithm}

We first verify the change in the dual objective at every step.
Each time $\mathcal A$ matches an arrival $j$ directly, the dual objective
increases by exactly $5/3$: $j$ is new so $\beta_j$ increases from $0$ to $1$,
and the chosen neighbor $i$ is free so $\alpha_i$ increases from $0$ to $2/3$.

Each time $\mathcal A$ matches an arrival $j$ using a $3$-path $\apath{j}{x}{y}{i}$, the dual
objective again increases by exactly $5/3$: $j$ is new so $\beta_j$ increases
from $0$ to $1/3$; the endpoint $i\in L$ is free so $\alpha_i$ increases from
$0$ to $1$; and $\alpha_x$ increases by $1/3$ (by the observation before the
proposition, $x$ was previously matched directly with $y$, so $\alpha_x=2/3$
before this update).
No dual variable ever decreases during the construction.

The postprocessing step only redistributes $1/3$ within a matched edge $ij\in M$,
changing $(\alpha_i,\beta_j)$ from $(2/3,1)$ to $(1,2/3)$, and hence it does not
change the total dual objective value.
To conclude via Lemma~\ref{lemma:pd-template} with $\rho=3/5$ (the reciprocal of
$5/3$), it remains to show that $(\alpha,\beta)$ is feasible for $D_G$, namely
that $\alpha_i+\beta_j\ge 1$ for every edge $ij\in E$.

Fix an arbitrary edge $ij\in E$. We rule out all possibilities with
$\alpha_i+\beta_j<1$.

\begin{itemize}
\item $(\alpha_i,\beta_j)=(0,0)$.
This would mean that both $i$ and $j$ are free in the final matching $M$.
But then the edge $ij$ would be an augmenting $1$-path for $M$, contradicting
that $\mathcal A$ is \spgreedy (a direct match never violates budgets).

\item $(\alpha_i,\beta_j)=(0,\frac23)$.
If $\beta_j = 2/3$, then $j$ was modified during postprocessing, implying it had no free neighbors in $L$ upon termination. However, if $i$ is a neighbor of $j$ with $\alpha_i = 0$, then $i$ must be free in $M$, providing $j$ with a free neighbor and creating a contradiction.

\item $(\alpha_i,\beta_j)=(0,\frac13)$.
If $\beta_j=\frac13$, then $j$ was matched upon arrival using a $3$-path, which
is only possible if $j$ had no free neighbor in $L$ at its arrival time.
However, $\alpha_i=0$ implies that $i$ is free in the final matching $M$.
By the basic properties of the model, offline vertices never become free after being matched; hence $i$ was free when $j$ arrived.
This contradicts that $j$ had no free neighbor at its arrival.

\item $(\alpha_i,\beta_j)=(\frac23,0)$. Since $\alpha_i\neq 0$, $i$ is matched; let
$l\in R$ be such that $il\in M$.
Since $\alpha_i=\frac23$, the edge $il$ entered the matching via a direct match
when $l$ arrived, and $i$ was never reassigned (otherwise a $3$-path update would have
set $\alpha_i$ to $1$). Hence $i$ is not exhausted.
Moreover, the edge $il$ stayed in the matching until the end, which implies
that $l$ was never reassigned; thus $l$ is not exhausted either.
Finally, because $\alpha_i$ remains $\frac23$ after the postprocessing step,
that step did not act on the matched edge $il$; therefore, $l$ has a free
neighbor $k\in L$ at the end.

Since $\beta_j=0$, the vertex $j$ is free in $M$. Thus $\apath{j}{i}{l}{k}$ is an
$M$-augmenting $3$-path with middle edge $il$. By Lemma~\ref{lemma:middle}
(middle-edge lemma), this implies that $il$ must be frozen, contradicting that
neither endpoint of $il$ is exhausted.
\end{itemize}

Thus $\alpha_i+\beta_j\ge 1$ for every edge $ij\in E$, so $(\alpha,\beta)$ is
feasible for $D_G$. Lemma~\ref{lemma:pd-template} yields $|M|\ge \frac35\,\nu(G)$,
proving that $\mathcal A$ is $\frac35$-competitive. \qedhere

\end{proof}

\end{document}